\documentclass[format=acmsmall]{acmart}

\usepackage{etoolbox}
\newtoggle{arxiv}
\toggletrue{arxiv} 

\renewcommand\footnotetextcopyrightpermission[1]{} 

\usepackage[ruled]{algorithm2e}

\SetAlFnt{\small}
\SetAlCapFnt{\small}
\SetAlCapNameFnt{\small}
\SetAlCapHSkip{0pt}
\IncMargin{-\parindent}

\usepackage{array}
\newcolumntype{L}[1]{>{\raggedright\let\newline\\\arraybackslash\hspace{0pt}}m{#1}}
\newcolumntype{C}[1]{>{\centering\let\newline\\\arraybackslash\hspace{0pt}}m{#1}}
\newcolumntype{R}[1]{>{\raggedleft\let\newline\\\arraybackslash\hspace{0pt}}m{#1}}

\usepackage{enumitem}

\newcommand{\E}{\mathbb{E}} 
\renewcommand{\P}{\mathbb{P}} 
\newcommand{\R}{\mathbb{R}} 
\newcommand{\N}{\mathbb{N}} 
\newcommand{\1}{\mathbbm{1}}
\newcommand{\trans}{\mathsf{T}}

\newcommand{\srank}{\textrm{srank}}
\newcommand{\rank}{\textrm{rank}}
\newcommand{\diag}{\textrm{diag}}
\newcommand{\surr}{\textrm{surr}}
\newcommand{\avg}{\textrm{avg}}
\newcommand{\outT}{\textrm{out}}

\newcommand{\fail}{\textrm{fail}}
\newcommand{\tsum}{\textstyle{\sum}}
\DeclareMathOperator*{\argmin}{arg\,min}

\newcommand{\plotHeight}{1.46in}
\usepackage{bbm,subcaption}
\graphicspath{{./figures/}}
\usepackage[titletoc,title]{appendix}
\usepackage{mathtools}
\mathtoolsset{showonlyrefs}
\iftoggle{arxiv}{ 
\newcommand{\new}[1]{#1} 
}{
\newcommand{\new}[1]{{\color{red}#1}} 
}

\begin{document}
\title{On the role of clustering in Personalized PageRank estimation}
\author{Daniel Vial}
\affiliation{%
  \institution{University of Michigan}
}
\email{dvial@umich.edu}
\author{Vijay Subramanian}
\affiliation{%
  \institution{University of Michigan}
}
\email{vgsubram@umich.edu}

\begin{abstract}
Personalized PageRank (PPR) is a measure of the importance of a node from the perspective of another (we call these nodes the \textit{target} and the \textit{source}, respectively). PPR has been used in many applications, such as offering a Twitter user (the source) recommendations of who to follow (targets deemed important by PPR); additionally, PPR has been used in graph-theoretic problems such as community detection. However, computing PPR is infeasible for large networks like Twitter, so efficient estimation algorithms are necessary. 

In this work, we analyze the relationship between PPR estimation complexity and clustering. First, we devise algorithms to estimate PPR for many source/target pairs. In particular, we propose an enhanced version of the existing single pair estimator \texttt{Bidirectional-PPR} that is more useful as a primitive for many pair estimation. We then show that the common underlying graph can be leveraged to efficiently and jointly estimate PPR for many pairs, rather than treating each pair separately using the primitive algorithm. Next, we show the complexity of our joint estimation scheme relates closely to the degree of clustering among the sources and targets at hand, indicating that estimating PPR for many pairs is easier when clustering occurs. Finally, we consider estimating PPR when several machines are available for parallel computation, devising a method that leverages our clustering findings, specifically the quantities computed \textit{in situ}, to assign tasks to machines in a manner that reduces computation time. This demonstrates that the relationship between complexity and clustering has important consequences in a practical distributed setting.
\end{abstract}

\maketitle

\section{Introduction}

Many systems, ranging from social networks to financial markets to the human brain, can be represented as graphs. When analyzing such systems, questions regarding the importance of nodes and the relationships between them arise. Which nodes are most influential, globally and locally? From the perspective of a given node, which nodes are important, and which nodes consider the given node to be important? How can these notions be quantified?

PageRank and Personalized PageRank (PPR) can help answer such questions. PageRank is a measure of the importance or centrality of a \textit{target} node; PPR ``personalizes" this measure to the perspective of a \textit{source} node. Proposed to rank Internet search results \cite{page1999pagerank}, and later to personalize these rankings \cite{haveliwala2002topic}, PageRank and PPR have been used in applications such as recommending Twitter followers \cite{gupta2013wtf} and YouTube videos \cite{baluja2008video}, as well as ``beyond the web'' \cite{gleich2015pagerank}, in fields such as bioinformatics \cite{morrison2005generank, freschi2007protein}. In graph theory, PPR has been used as a primitive for tasks such as detecting communities near a seed node \cite{andersen2006local} and assessing similarity between graphs \cite{koutra2013deltacon}.

The widespread use of PageRank and PPR can be attributed to the notion of relational ``importance" they convey, as well as the simplicity of the model from which they are derived. However, the scale of modern networks often makes them difficult (or impossible) to compute. As such, strategies for efficient estimation of PageRank and PPR are necessary.

\begin{figure}
\centering
\begin{subfigure}[b]{\columnwidth}
\centering\
\includegraphics[height=\plotHeight]{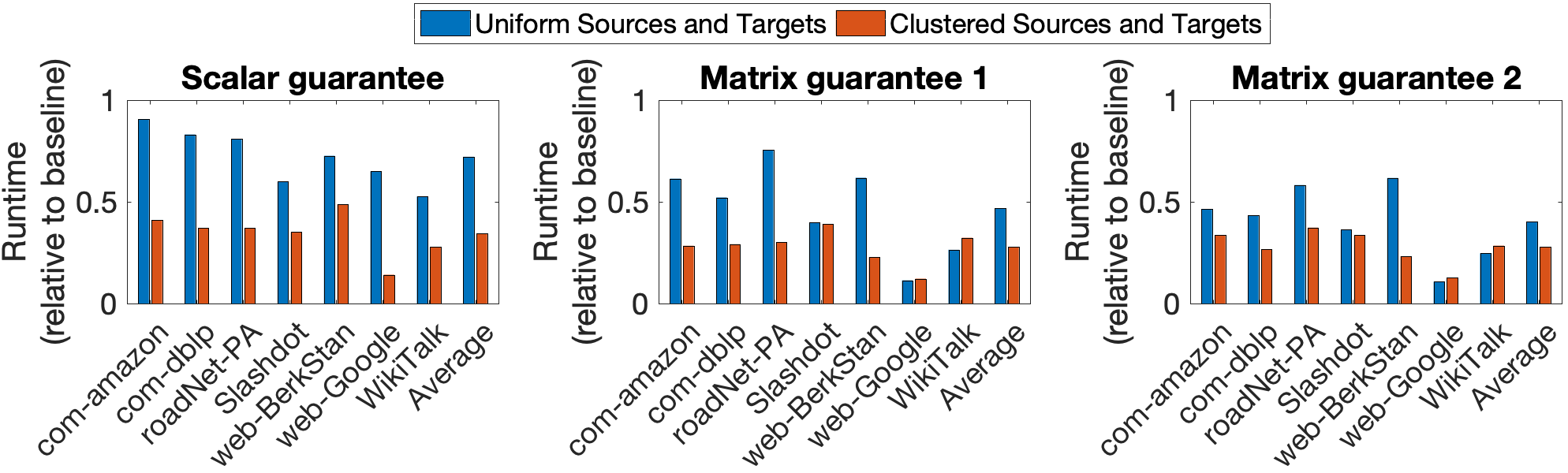} 
\caption{Across a diverse set of real graphs, our algorithms accelerate baseline methods; these accelerations are most significant when the sources and targets are clustered (experiment details in Section \ref{secExpReal}).} \label{figIntroAlgos}
\end{subfigure}
\begin{subfigure}[b]{0.64\columnwidth}
\centering\
\includegraphics[height=\plotHeight]{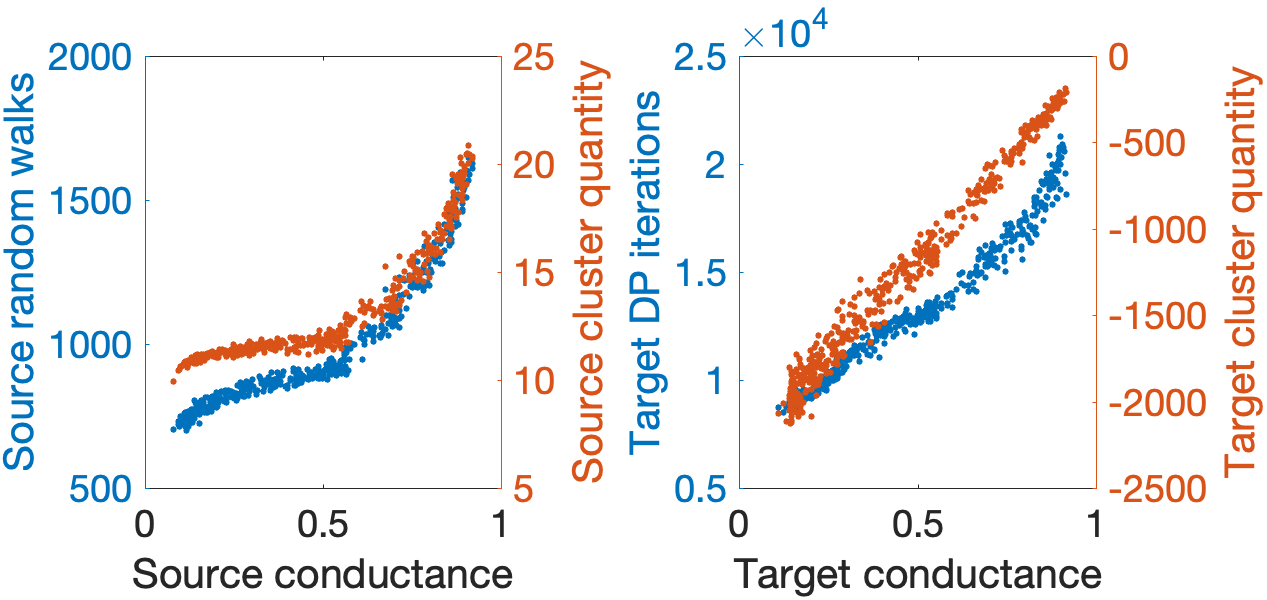} 
\caption{The complexity of the source and target stages of our methods scale with quantities that describe clustering of sources and targets, and that behave like conductance (experiment details in Section \ref{secExpSynScal}).} \label{figIntroClustQuant}
\end{subfigure}%
\hspace{0.02\columnwidth}%
\begin{subfigure}[b]{0.34\columnwidth}
\centering\
\includegraphics[height=\plotHeight]{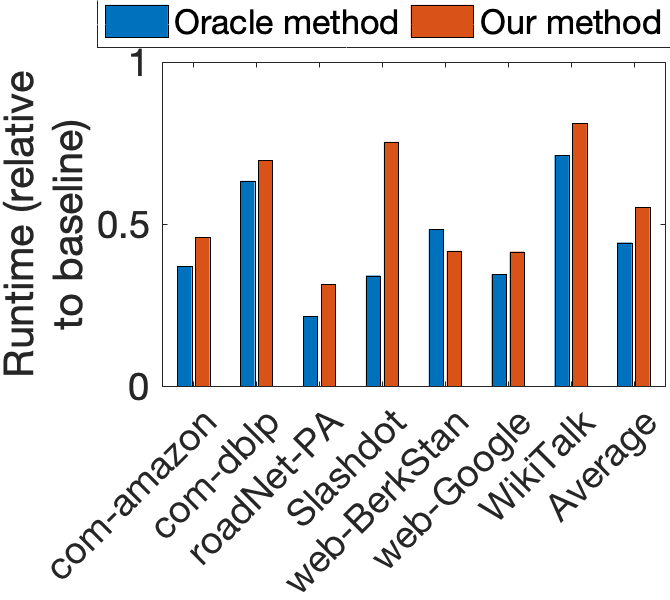}
\caption{Our findings can be used to identify clustering at runtime and accelerate PPR estimation (see Section \ref{secDistributed}).} \label{figIntroDist}
\end{subfigure}
\caption{Key empirical results} \label{figIntroResults}
\end{figure}

\textbf{Our contributions:} In this work, we analyze the relationship between clustering and PPR estimation complexity. In particular, we devise algorithms to estimate PPR for many source/target pairs, and we show that the complexity of these methods decreases with increased clustering among the sources and targets at hand. To demonstrate the consequences of our findings, we consider a distributed  setting in which this relationship between complexity and clustering can be leveraged to design more efficient algorithms. More specifically, our contributions are as follows: 
\begin{enumerate}

\item In Section \ref{secSinglePairDescription}, we propose an enhancement of \texttt{Bidirectional-PPR} \cite{lofgren2016personalized}, the state-of-the-art PPR estimator for a single source/target pair. As the name suggests, \texttt{Bidirectional-PPR} estimates PPR in two stages: random walks forward from the source node and dynamic programming (DP) backward from the target node. Our algorithm, called \texttt{FW-BW-MCMC}, adds a DP stage forward from the source that allows it to serve as a primitive in the many pair setting. In Appendix \ref{secSinglePairComparison}, we establish similar guarantees to those for \texttt{Bidirectional-PPR}.

\item In Section \ref{secManyPair}, we use \texttt{FW-BW-MCMC} as a primitive to estimate PPR for many pairs, proposing methods that accelerate the naive scheme of separately sampling random walks for each source and separately running DP for each target. For the sources, we show the forward DP allows random walk samples to be shared, decreasing the number of samples required. For the targets, we define a new iterative update for the backward DP, which eliminates repeated computations that may occur when treating each target separately. Using these ideas, we devise an algorithm with accuracy guarantees on each scalar estimate, and two algorithms with accuracy guarantees on the matrix containing all estimates. Across a diverse set of graphs, our methods are roughly 1.1 to 9.3 times faster than baseline methods (Fig.\ \ref{figIntroAlgos}).

\item We show analytically in Section \ref{secManyPair} and empirically in Section \ref{secExperiments} that the accelerations offered by our algorithms are most significant when the sources and targets are each clustered together in the graph, i.e.\ PPR estimation is ``easier'' when clustering occurs. For example, our algorithms typically accelerate baseline methods by factors of 3-4 when clustering occurs (Fig.\ \ref{figIntroAlgos}). More specifically, we prove the number of random walks for the sources and the number of DP iterations for the targets scale with quantities that describe clustering among the sources and targets, respectively; we find empirically that these clustering quantities scale with a more traditional clustering quantity, conductance (Fig.\ \ref{figIntroClustQuant}). \new{Also, while these clustering quantities are difficult to analyze in general, we provide analytical results for the stochastic block model, the prototypical model for networks with community structure.}

\item Finally, in Section \ref{secDistributed}, we demonstrate an application of our results, showing that our findings can be used to devise efficient algorithms when estimating PPR in a distributed setting. Specifically, we show that quantities computed during the forward DP can be used to predict the random walk sampling time for different assignments of tasks to machines, and we propose a natural but heuristic method to compute an assignment that (locally) minimizes this time. At a high level, our method ``learns'' the clustering structure present at runtime; empirically, this learning is quite successful, in the sense that our method performs nearly as well as an oracle method that knows the clustering structure \textit{a priori} (Fig.\ \ref{figIntroDist}).
\end{enumerate}
The remainder of the paper is organized as follows. We begin with preliminaries and related work in Sections \ref{secPrelim} and \ref{secRelated}, respectively. Sections \ref{secSinglePairDescription}-\ref{secDistributed} follow the outline above. We close in Section \ref{secConclusions}.

\section{Preliminaries} \label{secPrelim}

We begin with some preliminary definitions. Let $G = (V,E)$ be a directed graph, and define $n = |V|, m = |E|$. For $v \in V$, let $N_{\textrm{out}}(v) = \{ u \in V : v \rightarrow u \in E \}$ denote $v$'s outgoing neighbors, and let $d_{\textrm{out}}(v) = | N_{\textrm{out}}(v) |$ denote the out-degree of $v$. For simplicity, we assume $d_{\textrm{out}}(v) > 0$ $\forall\ v \in V$. Similarly define $N_{\textrm{in}}(v)$ and $d_{\textrm{in}}(v)$ as $v$'s incoming neighbors and in-degree. Finally, let $A$ denote the adjacency matrix of $G$, let $D$ be the diagonal matrix with $D(v,v) = d_{\textrm{out}}(v)$, and define $P = D^{-1} A$.

PageRank \cite{page1999pagerank} is the stationary distribution $\{ \pi(v) \}_{v \in V }$ of the Markov chain with transition matrix $(1-\alpha) P + \alpha \frac{1}{n} 1_n 1_n^T$, where $\alpha \in (0,1)$ and $1_n$ denotes the all ones vector of length $n$. In words, this chain is a random walk on $G$ for which, with probability $\alpha$ at each step, the random walker ``jumps'' to a uniform node, rather than following the walk. Let us denote this chain by $\{ X_i \}_{i \in \N}$. Clearly, $\{ X_i \}_{i \in \N}$ is irreducible and aperiodic, as $P(u,v) \geq \frac{\alpha}{n} > 0\ \forall\ u , v \in V$. Assuming $V$ is finite, positive recurrence follows, so $\{ \pi(v) \}_{ v \in V }$ exists, is unique, and satisfies
\begin{equation} \label{eqAvgNumberVisits}
\pi(v) = \lim_{N \rightarrow \infty} \frac{1}{N} \sum_{i=0}^{N-1} \mathbbm{1}_{ \{ X_i = v \} }\ \forall\ v \in V .
\end{equation}
It is by \eqref{eqAvgNumberVisits} that PageRank gives a measure of ``importance": we consider $v$ important when $\pi(v)$ is large, which occurs when $v$ is visited often on the chain $\{ X_i \}_{i \in \N}$.

We will henceforth refer to PageRank as \textit{global} PageRank to distinguish it from the generalization to PPR. Formally, PPR is the stationary distribution $\{ \pi_{\sigma}(v) \}_{ v \in V }$ of the Markov chain with transition matrix $P_{\sigma} = (1-\alpha) P + \alpha 1_n \sigma^T$. Here $\sigma$ is a nonnegative vector that sums to 1; hence, it yields a distribution on the jump locations, generalizing the uniform jumps of global PageRank. This gives $\pi_{\sigma}$ an interpretation like \eqref{eqAvgNumberVisits}, while also accounting for the preference on jump locations given by $\sigma$.

There are two important mathematical viewpoints of PPR that serve as the foundation for many estimation techniques; we will make use of both in the sections that follow. The first viewpoint is linear algebraic. Here we let $\pi_{\sigma}$ denote the stationary distribution as a row vector. By global balance, $\pi_{\sigma}$ satisfies $\pi_{\sigma} = \pi_{\sigma} P_{\sigma}$; solving for $\pi_{\sigma}$ (and assuming $\pi_{\sigma}$ is normalized to sum to 1) gives
\begin{equation} \label{eqPprLinAlg}
\pi_{\sigma} = \alpha \sigma^T ( I - (1-\alpha) P)^{-1} = \alpha \sigma^T \sum_{i=0}^{\infty} (1-\alpha)^i P^i .
\end{equation}
Note this immediately suggests estimating PPR by computing the first $k$ terms of the summation. The second PPR viewpoint is probabilistic. Denoting by $\{Y_i \}_{i \in \N}$ the Markov chain with transition matrix $P$, and letting $L \sim \textrm{geometric}(\alpha)$, \cite{athreya2003perfect} shows
\begin{equation} \label{eqPprEndpoint}
\pi_{\sigma}(v) = \P [ Y_L = v | Y_0 \sim \sigma ]\ \forall\ v \in V ,
\end{equation}
which again suggests estimation techniques; in this case, using Markov chain Monte Carlo (MCMC). We note this viewpoint is closely related to exact sampling \cite{propp1996exact} and Doeblin chains \cite{athreya2003perfect,meyn2012markov}.

Often, we let $\sigma = e_s$ for some $s \in V$, where $e_s \in \{ 0,1 \}^n$ satisfies $e_s(v) = \mathbbm{1}_{ \{ v = s \} }$. In such a case, we denote PPR as $\{ \pi_s(v) \}_{v \in V}$ and the  transition matrix as $P_s$. In fact, using \eqref{eqPprLinAlg}, one can show
\begin{equation} \label{eqLinearity}
\pi_{\sigma}(v) = \sum_{s \in V} \sigma(s) \pi_s(v)\ \forall\ v \in V .
\end{equation}
 Due to \eqref{eqLinearity}, many PPR estimation algorithms focus on estimating $\pi_s(v)$, from which extensions to estimating $\pi_{\sigma}(v)$ naturally follow; we focus on this as well throughout the paper.

Finally, we argue existence and uniqueness of PPR is not a concern. Indeed, the PPR Markov chain is aperiodic (since $P_s(s,s) \geq \alpha > 0$), so to guarantee existence and uniqueness, we need only verify irreducibility. For this, let $V_s =  \{ v \in V : \textrm{$\exists$ a path from $s$ to $v$ in $G$} \}$, so that $\forall\ u , v \in V_s$, $\exists\ i \in \N$ s.t.\ $P_s^i ( u , v ) > 0$ (we can jump from $u$ to $s$, then reach $v$ from $s$). Hence, if $P_s$ is \textit{not} irreducible, we can define a modified chain with states $V_s$ that \textit{is} irreducible and obtain the stationary distribution $\pi_s$ for this chain. We can then set $\pi_s(v) = 0\ \forall\ v \notin V_s$ -- intuitively, if $s$ cannot reach $v$, $v$ is not ``important" to $s$, so its PPR should be zero. Given this simple fix, we assume existence and uniqueness of PPR.

\begin{table}
\caption{Summary of notation} \label{tabSummaryNotation}
\begin{tabular}{|l|l|l|}
\hline
\textbf{Notation} & \textbf{Defined in} & \textbf{Description} \\ \hline
$\pi, \pi_s$ & Section \ref{secPrelim} & Global PageRank vector, PPR vector for $s \in V$  \\ \hline
$p^s, r^s$ & Section \ref{secSinglePairDescription} & Output vectors of Algorithm \ref{algApproxPR} (forward DP); satisfy \eqref{eqFwInvariant} \\ \hline
$\sigma_s$ & Section \ref{secSinglePairDescription} & Starting distribution for walks in \texttt{FW-BW-MCMC}; $\sigma_s = r^s / \| r^s \|_1$ \\ \hline
$r^s_{\max}$ & Section \ref{secSinglePairDescription} & Upper bound on $\| r^s \|_1$ at conclusion of Algorithm \ref{algApproxPR} \\ \hline
$p^t, r^t$ & Section \ref{secSinglePairDescription} & Output vectors of Algorithm \ref{algApproxCont} (backward DP); satisfy \eqref{eqBwInvariant} \\ \hline
$r^t_{\max}$ & Section \ref{secSinglePairDescription} & Upper bound on $\| r^t \|_{\infty}$ at conclusion of Algorithm \ref{algApproxCont} \\ \hline
$\Sigma_S$ & Section \ref{secManySource} & Matrix with rows $\{ \sigma_s \}_{s \in S}$ (subscript $S$ omitted at times) \\ \hline
$\| \Sigma_S \|_{\infty,1}$ & Section \ref{secManySource} & Source clustering quantity; see \eqref{eqSourceClustQuant} \\ \hline
$c_T$ & Section \ref{secManyTarget} & Target clustering quantity; see \eqref{eqTargetClustQuant}  \\ \hline
$\srank(A)$ & Section \ref{secManyPairMatrix} & Stable rank of matrix $A$; $\srank(A) = (\|A\|_F/\|A\|_2)^2$ \\ \hline
$P_S, R_S, P_T, R_T$ & Section \ref{secManyPairMatrix} & Matrices containing $\{ p^s \}_{s \in S}$, $\{ r^s \}_{s \in S}$, $\{ p^t \}_{t \in T}$, $\{ r^t \}_{t \in T}$, resp.\ \\ \hline
$\Pi(S,T)$ & Section \ref{secManyPairMatrix} & Matrix containing $\{ \pi_s(t) \}_{s \in S, t \in T}$ \\ \hline
$\sigma_{\avg}, \sigma_{\max}$ & Section \ref{secManyPairMatrix} & Starting distributions for walks in matrix estimators; see \eqref{eqSigmaAvgSigmaMaxDefn} \\ \hline
$\Phi(U)$ & Section \ref{secExpSynScal} & Conductance of $U \subset V$; see \eqref{eqConductance} \\ \hline
$\sigma_S$ & Section \ref{secDistributed} & Vector with $\sigma_s(v) = \max_{s \in S} \sigma_s(v)$ \\ \hline
\end{tabular}
\end{table}

\section{Related Work} \label{secRelated}

Before proceeding, we discuss some existing PPR estimation algorithms. Broadly speaking, these can be organized hierarchically: first, those that estimate the entire PPR matrix $\{ \pi_s(t) \}_{s \in V, t \in V}$; second, those that estimate a single row $\{ \pi_s(t) \}_{t \in V}$ or column $\{ \pi_s(t) \}_{s \in V}$ of this matrix, or its column sums (i.e.\ global PageRank); and third, those that estimate a single entry $\pi_s(t)$.

At the first level, several algorithms have been proposed to accelerate the power iteration or matrix inversion in \eqref{eqPprLinAlg}. To accelerate the power iteration, \cite{jeh2003scaling} provides a decomposition that allows a single row of the PPR matrix to be estimated using previously-computed rows; hence, this yields a procedure of first computing a small number of rows and then using these to estimate other rows. To obtain less costly matrix inversions, several works, e.g.\ \cite{tong2008random,shin2015bear}, have leveraged structural assumptions of the graph at hand. For example, Tong \textit{et al.}\ in \cite{tong2008random} propose a decomposition of $P$ into a block diagonal matrix $P_1$ and $P_2 : = P - P_1$; for graphs like social networks, $P_2$ can be extremely sparse. From the probabilistic viewpoint \eqref{eqPprEndpoint}, \cite{fogaras2005towards} gives an algorithm to estimate any entry of the PPR matrix at runtime using a precomputed database of random walk samples. 

At the second level, algorithms include the dynamic programming methods in \cite{andersen2006local} and \cite{andersen2008local} that estimate a row and a column of the PPR matrix, respectively; both can be viewed as localized versions of the power iteration in \eqref{eqPprLinAlg}. The algorithm in \cite{andersen2006local} yields $l_1$ and $l_{\infty}$ error guarantees on the row estimate with complexity $O(m)$, while \cite{andersen2008local} gives an $l_{\infty}$ guarantee on the column estimate with complexity $O(m)$. We make use of these algorithms in our methods and will discuss them in more detail in Section \ref{secSinglePairDescription}. We also note the approach in \cite{andersen2006local,andersen2008local} is closely related to work by Lee and co-authors \cite{lee2013computing,lee2014asynchronous,lee2014solving} that focuses on estimation of the stationary distribution of countable state-space Markov chains, as well as estimation in the context of general linear systems. From the probabilistic viewpoint, an important work is \cite{avrachenkov2007monte}, which analyzes Monte Carlo methods for global PageRank estimation, based on both the final step of sampled random walks (as given by \eqref{eqPprEndpoint}) and the number of visits along the entire walk. In \cite{avrachenkov2007monte}, it is shown that a single walk from each node (i.e.\ $n$ walks total) suffices to obtain estimates with small relative error for nodes with high global PageRank. Another work in this category is \cite{borgs2014multiscale}, which uses random walk-based methods to detect all nodes with global PageRank exceeding $n^{-\delta}, \delta \in (0,1)$ with complexity sublinear in $n$. \cite{borgs2014multiscale} also contains an algorithm to estimate a row of the PPR matrix with each estimate satisfying a multiplicative plus additive error guarantee; the complexity is linear in $n$ (if the error tolerance is set to match \cite{lofgren2016personalized}).

At the third level, the aforementioned \texttt{Bidirectional-PPR} algorithm from \cite{lofgren2016personalized} combines existing dynamic programming and Monte Carlo methods to estimate a single PPR value with worst-case and average-case complexity $O(n)$ and $O(\sqrt{m})$, respectively. From an accuracy perspective, this algorithm achieves a relative error bound for PPR values exceeding $1/n$, and an absolute error bound otherwise. We discuss this algorithm in more detail in Section \ref{secSinglePairDescription}.

In the context of this body of work, we will consider estimation of a small set of PPR values, $\{ \pi_s(t) \}_{s \in S, t \in T}$ for some $S,T \subset V$. While we do not precisely quantify ``small'', we implicitly assume $|S| \approx |T| = o( \sqrt{m} )$. In this setting, the existing methods described above can be applied in two ways. First, using methods such as the power iteration or the dynamic programming schemes (i.e.\ the first two levels of the above hierarchy), one can estimate entire rows and/or columns of the PPR matrix and then discard unwanted estimates. Such approaches typically have complexity $O( |S| m )$ or $O( |T| m )$ (depending on exactly which approach is used). Second, one can run the single pair estimator \texttt{Bidirectional-PPR} separately for each pair $(s,t) \in S \times T$. This approach has typical complexity $O(|S| |T| \sqrt{m} )$. When $|S| \approx |T| = o( \sqrt{m} )$, the second approach is more efficient. Hence, we will treat this approach as a baseline for comparison to our methods. Our primary contribution is to show that this baseline can be accelerated by exploiting clustering among $S$ and $T$ to estimate PPR values \textit{jointly}, rather than running \texttt{Bidirectional-PPR} \textit{separately} for each $(s,t) \in S \times T$.

\section{Single node pair estimation} \label{secSinglePairDescription}

We begin by proposing an enhancement of \texttt{Bidirectional-PPR} \cite{lofgren2016personalized}, the state-of-the-art single pair PPR estimator; we will introduce our algorithm and then describe \texttt{Bidirectional-PPR} as a special case. As mentioned in Section \ref{secRelated}, the idea behind these estimators is to combine dynamic programming (DP) and Markov chain Monte Carlo (MCMC) to estimate $\pi_s(t)$ for some $s, t \in V$. Our algorithm uses two DP stages and one MCMC stage. We will refer to these stages as the forward DP, backward DP, and MCMC stages; hence, we call our estimator \texttt{FW-BW-MCMC}. It is depicted pictorially in Fig.\ \ref{figAlgDepiction} and defined formally in Algorithm \ref{algOurEstimator}. Before proceeding, we briefly describe each stage.

\begin{figure}
\centering
\includegraphics[height=1in]{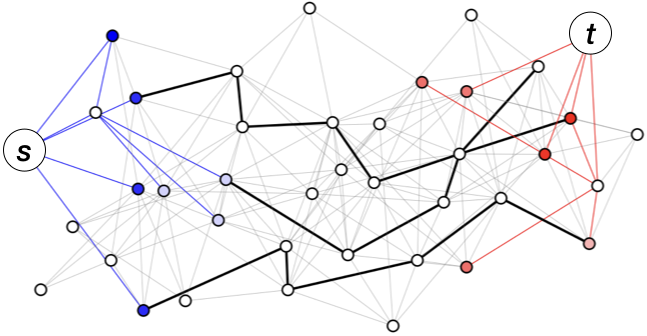}
\caption{Depiction of our algorithm \texttt{FW-BW-MCMC}. Blue and red nodes and edges show forward and backward dynamic programming, respectively; thick black edges show random walks.} \label{figAlgDepiction}
\end{figure}

The forward DP stage is Algorithm \ref{algApproxPR}. This is nearly identical to the \texttt{Approximate-PageRank} algorithm of \cite{andersen2006local}, so we use the same name here; however, we change the termination criteria from $\| D^{-1} r^s \|_{\infty} \leq r_{\max}^s$ to $\| r^s \|_{1} \leq r_{\max}^s$, where $r^s_{\max} \in (0,1)$ is an input to the algorithm (we describe our motivation for this change shortly). The algorithm takes as input $s \in V$ and produces $p^s, r^s \in \R_+^n$, shown in \cite{andersen2006local} to satisfy the invariant \eqref{eqFwInvariant} at each iteration.
\begin{equation} \label{eqFwInvariant}
\pi_s(u) = p^s(u) + \sum_{w \in V} r^s(w) \pi_w(u)\ \forall\ u \in V .
\end{equation}
As mentioned Section \ref{secRelated}, Algorithm \ref{algApproxPR} can be viewed as a ``localized" power iteration. At a high level, it computes elements of the matrices in \eqref{eqPprLinAlg} corresponding to high probability paths from $s$ to $u$ (the $p^s(u)$ term) while tracking the error from ``uncomputed" paths (the $\sum_{w \in V} r^s(w) \pi_w(u)$ term). These high probability paths are shown as blue edges in Fig.\ \ref{figAlgDepiction}.

The backward DP stage is \texttt{Approximate-Contributions} (Algorithm \ref{algApproxCont}, from \cite{andersen2008local}), which can be viewed as the dual of Algorithm \ref{algApproxPR}: while Algorithm \ref{algApproxPR} works forwards (along outgoing edges), Algorithm \ref{algApproxCont} works backwards (along incoming edges).  In \cite{andersen2008local}, it is shown that Algorithm \ref{algApproxCont} maintains the invariant \eqref{eqBwInvariant}, which can be interpreted similarly to \eqref{eqFwInvariant}. This stage is shown in red in Fig.\ \ref{figAlgDepiction}.
\begin{equation} \label{eqBwInvariant}
\pi_v(t) = p^t(v) + \sum_{w \in V} \pi_v(w) r^t(w)\ \forall\ v \in V.
\end{equation}

To motivate the MCMC stage, first observe that combining \eqref{eqFwInvariant} and \eqref{eqBwInvariant} with $u = t$ and $v = s$ gives
\begin{equation} \label{eqFwBwInvariant}
\pi_s(t) = p^t(s) + \langle p^s , r^t \rangle +  \sum_{w, w' \in V} r^s(w) \pi_w(w') r^t(w') ,
\end{equation}
and so, after running the DP stages, only the third term in \eqref{eqFwBwInvariant} is unknown. The goal of the MCMC stage is to estimate this term. Towards this end, let $\sigma_s = r^s / \| r^s \|_1$ and use \eqref{eqLinearity} to write this term as
\begin{equation}
\| r^s \|_1 \sum_{w' \in V}  \sum_{w \in V} \sigma_s(w) \pi_w(w')  r^t (w') = \| r^s \|_1 \sum_{w' \in V} \pi_{\sigma_s}(w') r^t(w') = \| r^s \|_1  \E_{U\sim \pi_{\sigma_s}} \left[ r^t(U) \right] .
\end{equation}
Leveraging the probabilistic PPR interpretation \eqref{eqPprEndpoint}, we can then estimate this term by sampling random walks. More specifically, we first sample a starting node from $\sigma_s$ (blue nodes in Fig.\ \ref{figAlgDepiction}), and we then sample a random walk beginning at the starting node (black edges in Fig.\ \ref{figAlgDepiction}). This process of sampling random walks is the MCMC stage of our algorithm.

As mentioned above, the forward DP stage uses termination criteria $\| r^s \|_{1} \leq r_{\max}^s$, rather than the $\| D^{-1} r^s \|_{\infty} \leq r_{\max}^s$ criteria used in \cite{andersen2006local}. This is because we will require a uniform bound on $\{ \| r^s \|_1 \}_{s \in S}$ when proving results pertaining to a set sources $S$ in later sections. However, this bound is not needed in practice, where we can instead use $\| D^{-1} r^s \|_{\infty} \leq r_{\max}^s$ termination. We call this variant of our algorithm \texttt{FW-BW-MCMC-Practical}; see Algorithm \ref{algOurEstimatorPractical} in Appendix \ref{appPracticalVersion} for a formal definition.

Having defined \texttt{FW-BW-MCMC}, we can describe the existing algorithm \texttt{Bidirectional-PPR}, which (in brief) operates as follows: run the backward DP from $t$, take $v = s$ in \eqref{eqBwInvariant}, and estimate the unknown term $\E_{U\sim \pi_s} [ r^t(U) ]$ by sampling random walks from $s$. We observe this is a special case of \texttt{FW-BW-MCMC}; specifically, the case $r^s_{\max} = 1$. We reemphasize that walks are sampled from $\nu \sim \sigma_s$ in \texttt{FW-BW-MCMC} and from $s$ in \texttt{Bidirectional-PPR}, which will be an important distinction later.

In the next sections, we will propose many pair estimators that use either \texttt{Bidirectional-PPR} or our enhancement as a primitive. We will show that using our enhancement as the primitive offers runtime accelerations not possible when using \texttt{Bidirectional-PPR}. Implicit in this discussion will be an understanding that using either primitive yields similar performance when these accelerations are ignored (so that using our enhancement as the primitive offers better performance when the accelerations are accounted for). In particular, we can prove the following results (as single pair estimation is not our focus, we defer formal statements and proofs to Appendix \ref{secSinglePairComparison}):
\begin{enumerate}
\item \texttt{FW-BW-MCMC}, \texttt{FW-BW-MCMC-Practical}, and \texttt{Bidirectional-PPR} offer the same accuracy guarantee (except for mild differences in assumptions) 

\item \texttt{FW-BW-MCMC} and \texttt{Bidirectional-PPR} have worst-case complexity $O(n)$

\item \texttt{FW-BW-MCMC-Practical} and \texttt{Bidirectional-PPR} have average-case complexity $O(\sqrt{m})$
\end{enumerate}

\begin{algorithm}
\caption{$(p^s,r^s) = \texttt{Approximate-PageRank} (G,s,\alpha,r_{\max}^s)$} \label{algApproxPR}
Initialize $p^s = 0, r^s = e_s$ \\
\While{$\| r^s \|_{1} > r_{\max}^s$}{
	 Let $v^* \in \arg \max_{v \in V} r^s(v) / d_{\textrm{out}}(v)$ \\
	Set $r^s(u) \leftarrow r^s(u) + (1-\alpha) r^s(v^*) / d_{\textrm{out}}(v^*)\ \forall\ u \in N_{\textrm{out}}(v^*)$,  $p^s(v^*) \leftarrow p^s(v^*) + \alpha r^s(v^*), r^s(v^*) = 0$
}
\end{algorithm}
\begin{algorithm}
\caption{$(p^t,r^t) = \texttt{Approximate-Contributions} (G,t,\alpha,r_{\max}^t)$} \label{algApproxCont}
Initialize $p^t = 0, r^t = e_t$ \\
\While{$\| r^t \|_{\infty} > r_{\max}^t$}{
	 Let $v^* \in \arg \max_{v \in V} r^t(v)$ \\
	Set $r^t(u) \leftarrow r^t(u) + (1-\alpha) r^t(v^*) / d_{\textrm{out}}(u)\ \forall\ u \in N_{\textrm{in}}(v^*)$, $p^t(v^*) \leftarrow p^t(v^*) + \alpha r^t(v^*), r^t(v^*) = 0$
}
\end{algorithm}
\begin{algorithm}
\caption{$\hat{\pi}_s(t) = \texttt{FW-BW-MCMC} (G,s,t,\alpha,r_{\max}^s,r_{\max}^t,w)$} \label{algOurEstimator}
 Let $(p^s,r^s) = \texttt{Approximate-PageRank} (G,s,\alpha,r_{\max}^s)$ (Algorithm \ref{algApproxPR}); set $\sigma_s = r^s / \| r^s \|_1$ \\
 Let $(p^t,r^t) = \texttt{Approximate-Contributions} (G,t,\alpha,r_{\max}^t)$ (Algorithm \ref{algApproxCont}) \\
\For{ $i = 1$ \KwTo $w$}{
	 Sample random walk starting at $\nu \sim \sigma_s$ of length $\sim$ geom($\alpha$); let $X_i = r^t ( U_i )$, where $U_i$ is endpoint of walk
}
 Let $\hat{\pi}_s(t) = p^t(s) + \langle p^s , r^t \rangle +  \frac{\| r^s \|_1}{w} \sum_{i=1}^w X_i$
\end{algorithm}

\section{Many node pair estimation} \label{secManyPair}

In this section, we consider the problem of estimating PPR for many node pairs, $\{ \pi_s(t) \}_{s \in S, t \in T}$ for some $S, T \subset V$. We consider two variants of this problem. First, in Section \ref{secManyPairScalar}, we view $\{ \pi_s(t) \}_{s \in S, t \in T}$ as a set of scalars, each of which we aim to accurately estimate. Second, in Section \ref{secManyPairMatrix}, we view $\{ \pi_s(t) \}_{s \in S, t \in T}$ as a matrix, which we aim to approximate accurately in the operator norm. For both variants, we propose algorithms that accelerate existing approaches, and we show the accelerations scale with quantities that can be interpreted as clustering measures of $S$ and $T$. In addition to these algorithms, we briefly discuss variants that use precomputation in Section \ref{secPrecomp}.

\subsection{Scalar estimation viewpoint} \label{secManyPairScalar}

Given $S,T \subset V$, a natural approach to estimate $\pi_s(t)\ \forall\ (s,t) \in S \times T$ is to use single pair estimators from Section \ref{secSinglePairDescription} as primitives. In particular, we could use either of the following approaches:
 \begin{itemize}
\item Run forward DP and sample random walks from $\nu \sim \sigma_s$ for each $s \in S$. Run backward DP from each $t \in T$. Compute estimates as in \texttt{FW-BW-MCMC}.
\item Sample random walks from each $s \in S$. Run backward DP from each $t \in T$. Compute estimates as in \texttt{Bidirectional-PPR}. 
\end{itemize}
As argued in Appendix \ref{secSinglePairComparison}, the primitives \texttt{FW-BW-MCMC} and \texttt{Bidirectional-PPR} are roughly equivalent in terms of complexity and accuracy; hence, both approaches have similar complexity. However, in Section \ref{secManySource}, we show the source stage of the first approach (forward DP and random walks) can be accelerated in a way not possible for the second approach. Further, in Section \ref{secManyTarget}, we show the target stage (backward DP) can be accelerated as well. Hence, using primitive method \texttt{FW-BW-MCMC} and the accelerations of Sections \ref{secManySource}-\ref{secManyTarget}, we can more efficiently estimate $\{ \pi_s(t) \}_{s \in S, t \in T}$.

\begin{figure}
\centering
\begin{subfigure}[b]{0.48\columnwidth}
\centering
\includegraphics[height=1in]{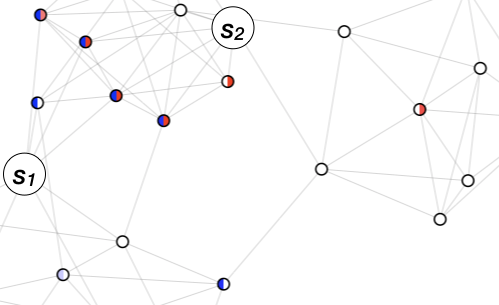}
\caption{Walks are sampled from blue nodes for $s_1$ and from red nodes for $s_2$; walks from blue \textit{and} red nodes are shared between $s_1$ and $s_2$.} \label{figTwoSource}
\end{subfigure}%
\hspace{0.04\columnwidth}%
\begin{subfigure}[b]{0.48\columnwidth}
   \centering
  \includegraphics[height=1in]{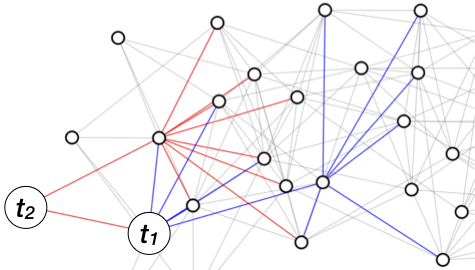}
   \caption{Red paths are computed via \texttt{Extend} during $t_2$ DP; blue paths can be computed via \texttt{Merge} during $t_2$ DP, rather than recomputed via \texttt{Extend}.} \label{figTwoTargets}
\end{subfigure}
\caption{Accelerations for \texttt{FW-BW-MCMC} when $|S| = |T| = 2$} \label{figAccelerationDepictions}
\end{figure}

\subsubsection{Source stage acceleration} \label{secManySource}

To accelerate the source stage, we define a unified MCMC stage for a set of sources $S$. At a high level, this scheme allows us to share walks across multiple $s \in S$, thereby decreasing the total number of walks required. We motivate the scheme pictorially in Fig.\ \ref{figTwoSource}, for the simple case $S = \{s_1, s_2\}$. Here blue and red depict $\sigma_{s_1}$ and $\sigma_{s_2}$ values, i.e.\ blue and red nodes are the starting nodes of random walks used in the $\pi_{s_1}$ and $\pi_{s_2}$ estimates, respectively. Observe several nodes have nonzero $\sigma_{s_1}$ and $\sigma_{s_2}$ values. The unified MCMC stage allows us to use random walks sampled from such nodes towards both estimates ($\pi_{s_1}$ and $\pi_{s_2}$).

To define the unified MCMC stage, we first define an equivalent MCMC stage for a single source. Recall that in Algorithm \ref{algOurEstimator} we sample each of $w$ random walks in two stages: first, we sample starting node $\nu_s \sim \sigma_s$, and second, we sample a walk starting at $\nu_s$. Equivalently, we can first sample starting nodes $\{ \nu_s^{(i)} \}_{i=1}^w$ i.i.d.\ from $\sigma_s$, and then sample $X^{(w)}_s(v) : = \sum_{i=1}^w \1_{  \{ \nu_s^{(i)} = v \} }$ walks starting at $v$, for each $v \in V$. With this in mind, the unified MCMC stage proceeds as follows. First, for each $s \in S$ we sample starting nodes $\{ \nu_s^{(i)} \}_{i=1}^w$ i.i.d.\ from $\sigma_s$ (as in the single source case), and we define $X^{(w)}_s(v)$ as above. Next, we sample $X^{(w)}(v) : = \max_{s \in S} X^{(w)}_s(v)$ walks starting at each $v \in V$. Letting $U_i^{v}$ denote the endpoint of the $i$-th walk from $v$,  we then estimate $\pi_s(t)$ as
\begin{equation} \label{eqMaxSamplingEstimate}
\hat{\pi}_s(t) = p^t(s) + \langle p^s , r^t \rangle +  \frac{\| r^s \|_1}{w} \sum_{v \in V : X^{(w)}_s(v) > 0} \sum_{i=1}^{X^{(w)}_s(v)} r^t( U_i^{v} ) .
\end{equation}
The final term in \eqref{eqMaxSamplingEstimate} is an unbiased estimate of $\E_{U \sim \pi_{\sigma_s}} [ r^t  ( U ) ]$ using $\sum_{v \in V} X_s^{(w)}(v) = w$ random walks, so the accuracy guarantee of Algorithm \ref{algOurEstimator} holds. To analyze the complexity of this scheme, first note we must sample $\sum_{v \in V} X^{(w)}(v)$ walks in total. We can then prove the following.

\begin{theorem} \label{STATE_UNIFIED_MCMC_COMP}
Fix $\epsilon, p_{\fail} \in (0,1)$. Assume $w$, the number of random walks required for each $s \in S$ in the unified MCMC stage from Section \ref{secManySource}, satisfies
\begin{equation}\label{eqUnifiedMcmcSampleComp}
w > \frac{3 \log( 2 \sum_{s \in S, v \in V} \mathbbm{1}_{ \{ \sigma_s(v) > 0 \} }  /p_{\fail}) }{ \epsilon^2 \min_{s \in S, v \in V : \sigma_s(v) > 0} \sigma_s(v) } .
\end{equation} 
Then with probability at least $1-p_{\fail}$, the total number of walks $\sum_{v \in V} X^{(w)}(v)$ sampled satisfies
\begin{equation} \label{eqUnifiedMcmcCompGuarantee}
\left| \sum_{v \in V} X^{(w)}(v) - w \sum_{v \in V} \max_{s \in S} \sigma_s(v) \right| \leq \epsilon w \sum_{v \in V} \max_{s \in S} \sigma_s(v) .
\end{equation}
\end{theorem}
\begin{proof}
See Appendix \ref{PROOF_UNIFIED_MCMC_COMP}.
\end{proof}
Before proceeding, we offer several remarks on this result:
\begin{itemize}
\item A lower bound on $w$ is given by \eqref{eqFwBwMcmcAccAss} in Theorem \ref{STATE_FWBWMCMC_ACC} to guarantee accuracy of each estimate. Thus, if $w$ exceeds both \eqref{eqUnifiedMcmcSampleComp} and \eqref{eqFwBwMcmcAccAss}, we obtain guarantees for both scalar accuracy and complexity of the walks. (In general, though, it is unclear which of \eqref{eqUnifiedMcmcSampleComp} and \eqref{eqFwBwMcmcAccAss} is larger.)
\item \new{In the worst case, the denominator on the right side of \eqref{eqUnifiedMcmcSampleComp} may be quite small, so the assumption on $w$ in Theorem \ref{STATE_UNIFIED_MCMC_COMP} may be restrictive. However, this only means that the concentration in \eqref{eqUnifiedMcmcCompGuarantee} may not provably occur, \textit{not} that the scheme will necessarily have poor performance. Furthermore, we find that this concentration essentially occurs for the values of $w$ used in practice, see e.g.\ leftmost plot in Fig.\ \ref{figClusteringPlot} and left two plots in Fig.\ \ref{figRealMultiPairDetails}.}
\item We will denote the matrix with rows $\{ \sigma_s \}_{s \in S}$ by $\Sigma$ (or by $\Sigma_S$, if we wish to emphasize the sources $S$ at hand) and will write the bound in \eqref{eqUnifiedMcmcCompGuarantee} as
\begin{equation} \label{eqSourceClustQuant}
\| \Sigma \|_{\infty,1} = \sum_{v \in V} \max_{s \in S} \sigma_s(v) 
\end{equation}
Here we have used the notation of the $L_{p,q}$ matrix norm, defined for a matrix $A$ as
\begin{equation} \label{eqLpqNorm}
\| A \|_{p,q} = \left( \textstyle \sum_j \left( \textstyle \sum_i | A(i,j) |^p \right)^{q/p} \right)^{1/q} .
\end{equation}
\end{itemize}

From Theorem \ref{STATE_UNIFIED_MCMC_COMP}, we expect to sample approximately $w \| \Sigma \|_{\infty,1}$ walks when $w$ is large. It is easily verified that $\| \Sigma \|_{\infty,1} \in [1,|S|]$, so our approach requires $w |S|$ random walks in the worst case, but only $w$ in the best case. In contrast, if we use \texttt{Bidirectional-PPR} as a primitive for many pair estimation, the unified MCMC stage is not possible (all random walks used to estimate $\pi_s$ begin at $s$, so sharing walks is not possible), and $w |S|$ walks are \textit{always} required. In short, \texttt{FW-BW-MCMC} with the unified MCMC stage accelerates the source stage of our many pair estimation approach.

\new{Unfortunately, it is difficult to quantify the degree of this acceleration in general, in part because $\| \Sigma_S \|_{\infty,1}$ depends on the forward DP, which itself is difficult to analyze. However, in Section \ref{secExperiments}, we offer empirical evidence that $\| \Sigma_S \|_{\infty,1}$ scales with the \textit{conductance} of $S$, a common measure of the clustering of $S$ in the underlying graph (see \eqref{eqConductance} for a formal definition). Furthermore, as will be discussed next, this quantity provably scales with clustering for the \textit{stochastic block model}, a common model for networks with community structure. In short, when $S$ is clustered in the graph, $\| \Sigma_S \|_{\infty,1}$ is typically small, and estimating PPR for many sources is easier.

We now turn to our result for the stochastic block model. We consider the special case for which $n$ is a perfect square and the graph is composed of $\sqrt{n}$ communities, each containing $\sqrt{n}$ nodes. (This allows us to compare the extremes of choosing $\sqrt{n}$ sources from the same community or from distinct communities; however, the analysis can be modified for other cases.) More specifically, we define $V_{n,i} = \{ 1 + (i-1) \sqrt{n} , \ldots , i \sqrt{n} \}$ and set $V_n = \cup_{i=1}^{\sqrt{n}} V_{n,i}$; we will view each $V_{n,i}$ as a community. For $v \in V_n$, we denote by $i(v)$ the unique $i \in \{1,\ldots,\sqrt{n}\}$ satisfying $v \in V_{n,i}$, i.e.\ $i(v)$ is the community that $v$ belongs to. We then construct a graph $G_n = (V_n, E_n)$ as follows: for any $u, v \in V_n$ s.t.\ $u \neq v$, edge $u \rightarrow v$ is present with probability $p_n$ if $i(u) = i(v)$ (i.e.\ if $u,v$ are in the same community), and is present with probability $q_n$ if $i(u) \neq i(v)$ (i.e.\ if $u,v$ are in different communities), independent of other edges. We also define
\begin{equation}
d_{\outT}(v) = | \{ u \in V_n : v \rightarrow u \in E_n \} | , \quad d_{\outT}^-(v) = | \{ u \in V_n \setminus V_{n,i(v)} : v \rightarrow u \in E_n \} | \quad \forall\ v \in V_n .
\end{equation}
In words, $d_{\outT}(v)$ is $v$'s out-degree (as before, though here it is a random variable), and $d_{\outT}^-(v)$ is the number of edges pointing from $v$ to other communities.

Our analysis will assume $p_n = p$ is a constant and $q_n = o ( 1 / \sqrt{n} )$. In this case, $\E [ d_{\outT}(v) ] = \Theta(\sqrt{n})$ (i.e.\ the graph is dense) and $\E [ d_{\outT}^-(v) ] = o ( \sqrt{n} )$ (i.e.\ nodes prefer to connect to their own community). Also, we assume the forward DP is run for at most $o(\sqrt{n})$ iterations. Since all nodes have out-degree $\Theta(\sqrt{n})$ with high probability (see proof of Theorem \ref{thmSbmSources}), this means we dedicate at most $o( n )$ complexity to the forward DP. This is consistent with the fact that our algorithm has average-case complexity $O ( \sqrt{m} )$, since $\sqrt{m} = n^{3/4}$ when all out-degrees are $\Theta(\sqrt{n})$. Hence, this assumption on the number of iterations is minor. Under these assumptions, we can prove the following:
\begin{theorem} \label{thmSbmSources}
Let $\{ G_n = (V_n,E_n) \}_{n \in \N : \sqrt{n} \in \N}$ be the sequence of stochastic block models described above, with $p_n = p$ for some constant $p \in (0,1)$ and $q_n = o ( 1 / \sqrt{n} )$. Assume we run the forward DP for at least one iteration, but at most $o ( \sqrt{n} )$ iterations. Then the following hold:
\begin{itemize}
\item For each $n$, let $S_n = V_{n,i}$ for some $i \in \{1,\ldots,\sqrt{n}\}$, i.e.\ all sources belong to the same community. If $q_n = \Omega ( \log n / n )$ (i.e.\ cross-community connections are dense), then for some constant $C > 0$,
\begin{equation}
\lim_{n \rightarrow \infty} \P \left( \| \Sigma_{S_n} \|_{\infty,1}  \leq C q_n n \right) = 1 ,\\
\end{equation}
i.e.\ $\| \Sigma_{S_n} \|_{\infty,1} = O ( q_n n ) = o ( \sqrt{n} )$ with high probability. If instead $q_n = \Theta(1/n)$ (i.e.\ cross-community connections are sparse), then for some constant $C > 0$,
\begin{equation}
\lim_{n \rightarrow \infty} \P \left( \| \Sigma_{S_n} \|_{\infty,1} \leq C \log n / \log \log n \right) = 1 ,
\end{equation}
i.e.\ $\| \Sigma_{S_n} \|_{\infty,1} = O ( \log n / \log \log n )$ with high probability.
\item For each $n$, let $S_n \subset V_n$ with $|S_n| = \sqrt{n}$ and $i(s) \neq i(s')\ \forall\ s , s' \in S_n$ s.t.\ $s \neq s'$, i.e.\ each source belongs to a distcint community. Then for any constant $\delta \in (0,1)$,
\begin{equation}
\lim_{n \rightarrow \infty} \P \left( \| \Sigma_{S_n} \|_{\infty,1} \geq (1-\delta) \sqrt{n} \right) = 1 ,
\end{equation}
i.e.\ $\| \Sigma_{S_n} \|_{\infty,1} \in [ ( 1-\delta) \sqrt{n} , \sqrt{n} ]$ with high probability.
\end{itemize}
\end{theorem}
\begin{proof}
\iftoggle{arxiv}{%
See Appendix \ref{appProofThmSbmSources}.%
}{%
See Appendix H in \cite{fullPaper}.%
}%
\end{proof}
}

\subsubsection{Target stage acceleration} \label{secManyTarget}

Our next goal is to accelerate the target stage of the many pair estimation approach. For this, we propose a unified DP stage that avoids repeated computations that may occur when running backward DP for each target separately. We motivate our approach in the simple case $T = \{ t_1, t_2 \}$. Assume that $p^{t_1}, r^{t_1}$ have been computed by Algorithm \ref{algApproxCont}, and that $p^{t_2}, r^{t_2}$ are currently being computed. If $r^{t_2}(t_1) > r^t_{\max}$ at some iteration, we can use the alternate update rule given by \eqref{eqMergeUpdate} (rather than that given in Algorithm \ref{algApproxCont}).
\begin{equation} \label{eqMergeUpdate}
p^{t_2} \leftarrow p^{t_2} + r^{t_2}(t_1) p^{t_1}, \quad r^{t_2} \leftarrow r^{t_2} + r^{t_2}(t_1) ( r^{t_1} - e_{t_1} ) .
\end{equation}
When $p^{t_2}, r^{t_2}$ are updated via \eqref{eqMergeUpdate}, the invariant \eqref{eqBwInvariant} is maintained. Indeed, for any $s \in V$,
\begin{align}
&  p^{t_2}(s) + r^{t_2}(t_1) p^{t_1}(s) + \sum_{u \in V} \pi_s(u) ( r^{t_2}(u) + r^{t_2}(t_1) ( r^{t_1}(u) - e_{t_1}(u) ) )  \\
& \quad =  p^{t_2}(s) + \sum_{u \in V} \pi_s(u) r^{t_2}(u) + r^{t_2}(t_1) ( ( p^{t_1}(s) + \sum_{u \in V} \pi_s(u) r^{t_1}(u) ) - \pi_s(t_1) ) ) \label{eqMergeInvariant1}  \\
& \quad = \pi_s(t_2) + r^{t_2}(t_1) ( \pi_s(t_1) - \pi_s(t_1) ) = \pi_s(t_2) , \label{eqMergeInvariant2} 
\end{align}
where in \eqref{eqMergeInvariant1} we simply rearranged terms and in \eqref{eqMergeInvariant2} we assume $p^{t_1}, r^{t_1}$ and $p^{t_2}, r^{t_2}$ satisfy \eqref{eqBwInvariant}. 

We can interpret the update rule \eqref{eqMergeUpdate} as follows. As discussed in Section \ref{secSinglePairDescription}, we view Algorithm \ref{algApproxCont} as a method of traversing paths to $t$ and computing the probability of these paths. For the update in Algorithm \ref{algApproxCont}, specific paths are extended by a single step at each iteration; we call this update \texttt{Extend}. In contrast, the alternate update rule \eqref{eqMergeUpdate} extends paths by (potentially) many steps in an iteration; specifically, by appending paths to $t_1$, with paths from $t_1$ to $t_2$, to obtain paths to $t_2$. We call this update \texttt{Merge} to highlight that longer paths are combined to obtain new paths.

The utility of \texttt{Merge} is that the probability of paths to $t_2$ through $t_1$ need not be recomputed one step at a time via \texttt{Extend}. This is depicted in Fig.\ \ref{figTwoTargets}: red paths are computed via \texttt{Extend} during $t_2$ DP; blue paths, having already been computed via \texttt{Extend} during $t_1$ DP, are used to compute longer paths in a single iteration via \texttt{Merge} during $t_2$ DP. In contrast, blue paths would be recomputed one step at a time via \texttt{Extend} during $t_2$ DP, if separate DP was used. In short, \texttt{Merge} may allow Algorithm \ref{algApproxCont} to terminate in fewer iterations. This is made more specific in Proposition \ref{STATE_BW_ITER_MERGE}.

\begin{proposition} \label{STATE_BW_ITER_MERGE}
Suppose $T = \{t_1,t_2\}$ and $\pi_{t_1}(t_2) > r^t_{\max}$. If we run Algorithm \ref{algApproxCont} for $t_2$ and use \texttt{Merge} at iterations for which $v^* = t_1$, the algorithm terminates in at most $\frac{n \pi(t_2)}{\alpha r^t_{\max}} - \frac{( \| p^{t_1} \|_1 - \alpha )}{\alpha}$ iterations. If \texttt{Merge} is not used, the number of iterations for termination is at most $\frac{n \pi(t_2)}{\alpha r^t_{\max}}$. 
\end{proposition}
\begin{proof}
See Appendix \ref{PROOF_BW_ITER_MERGE}.
\end{proof}

From Algorithm \ref{algApproxCont}, $\| p^{t_1} \|_1 \geq \alpha$. Hence, Proposition \ref{STATE_BW_ITER_MERGE} allows us to tighten the iteration bound by $\frac{( \| p^{t_1} \|_1 - \alpha )}{\alpha} \geq 0$ (with equality if and only if the algorithm terminates in a single iteration for $t_1$). In the more general case, the iterations we save roughly scales with the quantity  
\begin{equation} \label{eqTargetClustQuant}
c_T = \sum_{i=1}^{|T|} \left| \left\{ j \in \{ 1 , 2 , ... , i-1 \} : \pi_{t_j}(t_i) > r^t_{\max} \right\} \right| ,
\end{equation}
assuming the nodes in $T$ are chosen in order $\{t_1,t_2,\dotsc, t_{|T|}\}$. We note the choice of this order has a clear impact on performance, but optimizing it at runtime is difficult; we discuss this more in %
\iftoggle{arxiv}{%
Appendix \ref{secTargetOrder}. %
}{%
Appendx I of \cite{fullPaper}. %
}%
See Algorithm \ref{algApproxContMany} for our many target algorithm.

We next offer a clustering interpretation of the quantity $c_T$. For this, note $\pi_{t_j}(t_i) > r^t_{\max}$ is a notion of ``closeness" between $t_i$ and $t_j$; hence, $c_T$ is a notion of clustering of the set $T$, and our analysis suggests estimating PPR for many targets is easier when the targets are clustered. Note that, while the source clustering quantity $\| \Sigma \|_{\infty,1}$ from Section \ref{secManySource} is \textit{smaller} when clustering among sources is more significant, the target clustering quantity $c_T$ is \textit{larger} when clustering among targets is more significant; in Section \ref{secExperiments}, we show $-c_T$ scales with the conductance of $T$ in practice.

\begin{algorithm}
\new{
\caption{$\{ (p^t,r^t) \}_{t \in T} = \texttt{Approx-Cont-Many-Targets} (G,T = \{t_1,t_2,\ldots,t_{|T|}\} ,\alpha,r_{\max}^t)$} \label{algApproxContMany}
\For{$i=1$ \KwTo $|T|$}{
 $p^{t_i} = 0, r^{t_i} = e_{t_i}$ \\
\While{$\| r^{t_i} \|_{\infty} > r_{\max}^t$}{
	  $v^* \in \arg \max_{v \in V} r^{t_i}(v)$   \\
	\uIf{$ v^* \in \{ t_1, \ldots , t_{i-1} \}$}{
		 $p^{t_i} \leftarrow p^{t_i} + r^{t_i}(v^*) p^{v^*},  r^{t_i} \leftarrow r^{t_i} + r^{t_i}(v^*) ( r^{v^*} - e_{v^*} )$  (i.e.\ use \texttt{Merge} update \eqref{eqMergeUpdate})
}
\Else{
	 $r^{t_i}(u) \leftarrow r^{t_i}(u) + (1-\alpha) \frac{r^{t_i}(v^*)}{ d_{\textrm{out}}(u)}\ \forall\ u \in N_{\textrm{in}}(v^*)$, $p^{t_i}(v^*) \leftarrow p^{t_i}(v^*) + \alpha r^{t_i}(v^*), r^{t_i}(v^*) = 0$
}	
}
}
}
\end{algorithm}

\subsection{Matrix approximation viewpoint} \label{secManyPairMatrix}

For the second variant of the many pair estimation problem, we view $\{ \pi_s(t) \}_{s \in S, t \in T}$ as a matrix that we aim to accurately approximate. For simplicity, we assume $|S| = |T| = l$, and we denote these sets $S = \{ s_i \}_{i=1}^l, T = \{ t_i \}_{i=1}^l$. We also assume $V = \{ 1 , 2 , ... , n \}$, and we let $\Pi$ denote the matrix of dimension $n \times n$ whose $(i,j)$-th element is $\pi_i(j)$. In this notation, we seek an estimate $\hat{\Pi}(S,T)$ of $\Pi(S,T)$ that minimizes $\| \hat{\Pi}(S,T) - \Pi(S,T) \|_2$, where for a matrix $A$, $A(I,J)$ denotes the submatrix of $A$ containing rows $I$ and columns $J$, and where $\| A \|_2 = \max_{\|x\|_2 = 1} \| A x \|_2$ is the operator norm.

Before proceeding, we introduce additional notation used in this section. Similar to the $A(I,J)$ notation, $A(I,:)$ and $A(:,J)$ are the submatrices with rows $I$ and all columns, and all rows and columns $J$, respectively. For a vector $x$, $x(I)$ is the vector with elements $I$; when $x$ has nonzero entries, $\diag(x)$ and $\diag(1/x)$ are the diagonal matrices whose $i$-th diagonal elements are $x(i)$ and $1/x(i)$, respectively. Finally, we will encounter stable rank, which for a matrix $A$ is defined as $\srank(A) = (  \| A \|_F / \| A \|_2  )^2$, where $\| \cdot \|_F = \| \cdot \|_{2,2}$ is the Frobenius norm, with $\| \cdot \|_{2,2}$ defined as in \eqref{eqLpqNorm}. It is straightforward to verify $1 \leq \srank ( A ) \leq \rank ( A )$ by writing $\| A \|_F^2$ and $\| A \|_2^2$ in terms the singular values of $A$ (see, for example, Section 2.1.15 of \cite{tropp2015introduction}).

With this notation in mind, we define the following matrices:
\begin{gather}
P_S \in \R^{n \times l} \textrm{ s.t.\ } P_S(i,j) = p^{s_j}(i), \quad R_S \in \R^{n \times l} \textrm{ s.t.\ } R_S(i,j) = r^{s_j}(i) , \label{eqSmatDefn} \\
P_T \in \R^{n \times l} \textrm{ s.t.\ } P_T(i,j) = p^{t_j}(i) , \quad R_T \in \R^{n \times l} \textrm{ s.t.\ } R_T(i,j) = r^{t_j}(i) . \label{eqTmatDefn} 
\end{gather}
Here $p^{s_j}, r^{s_j}$ and $p^{t_j}, r^{t_j}$ are assumed to have been computed via Algorithms \ref{algApproxPR} and \ref{algApproxContMany}, respectively. We may then collect the invariant \eqref{eqFwBwInvariant} for each $(s_i, t_j)$ pair in matrix form as
\begin{equation} \label{eqMatrixInvariant}
\Pi(S,T) = P_T(S,:) + P_S^{\trans} R_T + R_S^{\trans} \Pi R_T .
\end{equation}
Observe only $R_S^{\trans} \Pi R_T $ is unknown in \eqref{eqMatrixInvariant}. Hence, we consider estimation of this term. To this end, let $\sigma$ be any $n$-length vector satisfying $\sigma(i) > 0\ \forall\ i \in \{1,2,...,n\}$ and $\sum_{i=1}^n \sigma(i) =1$; note we may view $\sigma$ as a distribution on $V$. We then rewrite the unknown term in \eqref{eqMatrixInvariant} as
\begin{equation} \label{eqMatrixInvariantUnknown}
R_S^{\trans} \Pi R_T  = R_S^{\trans} \diag( 1 / \sigma )  \diag( \sigma ) \Pi  R_T .
\end{equation}
Using \eqref{eqMatrixInvariantUnknown}, we can obtain unbiased estimates of $R_S^{\trans} \Pi R_T$ as follows. Let $\{ \mu_i \}_{i=1}^w$ be i.i.d. samples from $\sigma$. For $i \in \{ 1 , 2 , ... , w \}$, let $\nu_i \sim \pi_{\mu_i}$ independently (where we sample from $\pi_{\mu_i}$ using a random walk, as given by \eqref{eqPprEndpoint}), and let $X_i = R_S^{\trans} \diag( 1 / \sigma ) e_{\mu_i} e_{\nu_i}^{\trans} R_T$. It is straightforward to see $\E [ e_{\mu_i} e_{\nu_i}^{\trans} ] = \diag( \sigma ) \Pi$; hence, $\E [ X_i ] = R_S^{\trans} \Pi R_T$. We may then estimate $\Pi(S,T)$ as $\hat{\Pi}(S,T) = P_T(S,:) + P_S^{\trans} R_T + \frac{1}{w} \sum_{i=1}^w X_i$.

We will consider two forms of $\sigma$ for this approach. Specifically, let us define
\begin{equation} \label{eqSigmaAvgSigmaMaxDefn}
\sigma_{\avg}(i) = \frac{1}{l} \sum_{s \in S} \sigma_s(i)  , \quad \sigma_{\textrm{max}}(i) = \frac{1}{\| \Sigma \|_{\infty,1}} \max_{s \in S} \sigma_s(i) ,
\end{equation}
where $\sigma_s = r^s / \| r^s \|_1$ as before. Observe that when $\sigma \in \{ \sigma_{\avg} , \sigma_{\max} \}$, the assumption $\sum_{i=1}^n \sigma(i) = 1$ is satisfied. Furthermore, we argue that the assumption $\sigma(i) > 0$ is without loss of generality in these cases. Indeed, suppose $\sigma(j) = 0$ for some $j$ and $\sigma(i) > 0$ for $i \neq j$. Then $\P [ \mu_i = j ] = 0$ by definition, and by \eqref{eqSigmaAvgSigmaMaxDefn}, $r^s(j) = 0\ \forall\ s \in S$. It is then readily verified that $R_S ( V \setminus \{ j \} , : )^{\trans} \diag( 1 / \sigma ( V \setminus \{ j \} ) ) e_{\mu_i} e_{\nu_i}^{\trans} R_T$ is an unbiased estimate of $R_S^{\trans} \Pi R_T$. Given this simple fix, we assume $\sigma(i) > 0$ moving forward.

\begin{algorithm}
\caption{$\hat{\Pi}(S,T) = \texttt{FW-BW-MCMC-Many-Pair} (G,S =  \{ s_i \}_{i=1}^l,T= \{ t_i \}_{i=1}^l ,\alpha,r_{\max}^t,r_{\max}^s,w)$} \label{algMatrixEstimator}
\For{ $i = 1$ \KwTo $l$}{
	 $(p^{s_i},r^{s_i}) = \texttt{Approximate-PageRank} (G,s_i,\alpha,r_{\max}^s)$ (Algorithm \ref{algApproxPR})
}
 $\{ (p^t,r^t) \}_{t \in T} = \texttt{Approx-Cont-Many-Targets} (G,T,\alpha,r_{\max}^t)$ (Algorithm \ref{algApproxContMany}) \\
 Construct $P_S, R_S, P_T, R_T$ via \eqref{eqSmatDefn}, \eqref{eqTmatDefn}; compute $\sigma = \sigma_{\avg}$ or $\sigma = \sigma_{\max}$ via \eqref{eqSigmaAvgSigmaMaxDefn} \\
\For{ $i = 1$ \KwTo $w$}{
	 Let $X_i = R_S^{\trans} \diag(1/\sigma) e_{\mu_i} e_{\nu_i}^{\trans} R_T$, where $\nu_i$ is endpoint of walk starting at $\mu_i \sim \sigma$ of length $\sim$ geom($\alpha$)
}
Let $\hat{\Pi}(S,T) = P_T(S,:) + P_S^{\trans} R_T + \frac{1}{w} \sum_{i=1}^w X_i$
\end{algorithm}

To summarize, we have proposed the matrix approximation scheme formally defined in Algorithm \ref{algMatrixEstimator}. Theorem \ref{STATE_MATRIX_ACC} provides a guarantee for the accuracy of this scheme.

\begin{theorem} \label{STATE_MATRIX_ACC}
Fix $\epsilon > 0$. If $\sigma = \sigma_{\avg}$ in Algorithm \ref{algMatrixEstimator}, assume the number of walks $w$ satisfies
\begin{equation}
w \geq  l^2 \sqrt{\srank(\Pi(S,T))} \log ( 2 l / p_{\textrm{fail}} ) r^s_{\max} r^t_{\max} (6+4\epsilon) / (3\epsilon^2) .
\end{equation}
If instead $\sigma = \sigma_{\max}$ in Algorithm \ref{algMatrixEstimator}, assume $w$ satisfies
\begin{equation}
w \geq  l^{3/2} \| \Sigma \|_{\infty,1} \log ( 2 l / p_{\textrm{fail}} ) r^s_{\max} r^t_{\max} (6+4\epsilon) / (3\epsilon^2) .
\end{equation}
Then for both choices of $\sigma$, and with probability at least $1-p_{\textrm{fail}}$, Algorithm \ref{algMatrixEstimator} returns an estimate $\hat{\Pi}(S,T)$ satisfying $\| \Pi(S,T) - \hat{\Pi}(S,T) \|_2 \leq \epsilon \max \{ \| \Pi(S,T) \|_2 , 1 \}$.
\end{theorem}

\begin{proof} See Appendix \ref{PROOF_MATRIX_ACC} \end{proof}

We note that, neglecting common factors, Theorem \ref{STATE_MATRIX_ACC} states $w$ scales with $l^2$ and $l^{3/2}$ in the best case for $\sigma_{\avg}$ and $\sigma_{\max}$, respectively; in the worst case, $w$ scales with $l^{5/2}$ for both approaches. In the next section, we compare $\sqrt{ l\ \srank(\Pi(S,T)) }$ with $\| \Sigma \|_{\infty,1}$ empirically to compare the ``typical" case.

Next, we observe Theorem \ref{STATE_MATRIX_ACC} shows that, as in the scalar estimation viewpoint of Section \ref{secManyPairScalar}, PPR matrix approximation is easier when clustering occurs. This is because, when $\sigma = \sigma_{\max}$, complexity scales with $\| \Sigma \|_{\infty,1}$ (which we have argued is measure of clustering of $S$); when $\sigma = \sigma_{\avg}$, complexity scales with $\srank(\Pi(S,T))$, a measure of matrix dimensionality. Additionally, stable rank is unique from the clustering quantities introduced thus far in that it takes into account both $S$ and $T$ (unlike $\| \Sigma \|_{\infty,1}$, which only accounts for $S$, or $c_T$, which only accounts for $T$).

Finally, we comment on a difference for the choices of $\sigma$. In particular, when $\sigma = \sigma_{\max}$, one can set $w$ proportional to $\| \Sigma \|_{\infty,1}$ before sampling random walks, leveraging clustering at runtime to increase efficiency. In contrast, when $\sigma = \sigma_{\avg}$, the scaling factor in the $w$ lower bound is the unknown quantity $\srank(\Pi(S,T))$. However, we propose using $\srank( P_T(S,:) + P_S^{\trans} R_T )$ (known at runtime) as a surrogate for $\srank(\Pi(S,T))$. In Section \ref{secExperiments}, we show empirically that using this surrogate yields performance similar to using $\srank(\Pi(S,T))$.

\subsection{Precomputation variants} \label{secPrecomp}

While we have thus far assumed all computations are done online, one can also consider variants for which some computations are done offline, with the results stored for later use. In fact, in Section 4 of \cite{lofgren2016personalized}, the authors propose several such algorithms for the case of one source $s \in V$ and many targets $T \subset V$, using \texttt{Bidirectional-PPR} as a primitive. Each of these variants proceeds as follows. For the offline stage, \texttt{Approx-Contributions} is run for every $t \in V$, and the vectors $\{ p^t, r^t \}_{t \in V}$ are stored. For the online stage, random walks are sampled from $s$, and $\{ \pi_s(t) \}_{t \in T}$ are estimated using the endpoints of these walks and $\{ p^t, r^t \}_{t \in T}$. As mentioned, several such algorithms are proposed; these only differ in how the vectors are stored and how the walks and vectors are combined to generate estimates. In particular, the basic framework of running \texttt{Approx-Contributions} offline and sampling walks from $s$ online is used in all of the precomputation algorithms from \cite{lofgren2016personalized}.

Analogous to our extension of \texttt{Bidirectional-PPR} from the single pair case to the many pairs case, we can extend these precomputation algorithms from the single source case to the many sources case. Specifically, we can modify each of these algorithms in two ways (but otherwise leave them unchanged). First, we can modify the offline stage by also precomputing and storing $\{ p^s , r^s \}_{s \in V}$ via \texttt{Approx-PageRank}. Second, we can modify the online stage by sampling walks using the precomputed vectors $\{ r^s \}_{s \in S}$ and the walk sharing scheme from Section \ref{secManySource}.

To assess the performance of this approach, we compare against the naive extension of \cite{lofgren2016personalized}'s precomputation algorithms to the case $|S| > 1$; namely, leaving the offline stage unchanged and sampling walks separately from each $s \in S$ online. Clearly, our approach requires more storage (due to running \texttt{Approx-PageRank} offline); however, this storage will be roughly double that of the naive extension and thus will not increase the order of the space complexity. On the other hand, our approach will accelerate the online stage of this naive extension, since fewer random walks will typically be sampled. Specifically, per Section \ref{secManySource}, we expect to sample $w \| \Sigma \|_{\infty,1}$ walks instead of $w |S|$ walks; as discussed previously, the former quantity can be much smaller if $S$ is clustered.

We also note that Algorithm \ref{algApproxContMany} can be used to compute $\{ p^t, r^t \}_{t \in V}$ offline, though this is a minor point, since offline computational complexity is generally not a concern. However, this raises another point. When  precomputation is not allowed, our source and target accelerations are both used at runtime; when precomputation is allowed, only our source acceleration is used at runtime. Hence, the runtime savings of our schemes may be less significant in the precomputation setting. In spite of this, we believe the savings will still be considerable in general. This belief follows from the fact that, in our experiments, the source acceleration is generally at least as significant as the target acceleration. For example, Fig.\ \ref{figManySTplots} shows that the number of random walks sampled grows more slowly in $|S|$ than the number of DP iterations grows in $|T|$. Additionally, Fig.\ \ref{figRealMultiPair} shows that for fixed $|S|,|T|$, walk savings and DP iteration savings are comparable across a wide range of graphs.

\section{Experiments} \label{secExperiments}

In this section, we demonstrate the empirical performance of our algorithms and the role of clustering in their performance. We conduct experiments using both synthetic and real graphs. On the synthetic side, we use a directed Erd\H{o}s-R\'enyi graph and directed stochastic block model (referred to hereafter as \texttt{Direct-ER} and \texttt{Direct-SBM}, respectively), each with $n = 2 \times 10^3$ and $\E [ m ] = 2 \times 10^4$. For the real datasets, we use a set of graphs from the Stanford Network Analysis Platform \cite{snapnets} including social networks (\texttt{Slashdot}, \texttt{Wiki-Talk}), partial web crawls (\texttt{web-BerkStan}, \texttt{web-Google}), co-purchasing and co-authoring graphs (\texttt{com-amazon}, \texttt{com-dblp}), and a road network (\texttt{roadNet-PA}). In addition to the diverse domains of these datasets, they differ in terms of sparsity (in order of magnitude, each has $10^6$ edges, but the number of nodes ranges from $10^4$ to $10^6$), so we believe our empirical results are robust across different graph structures. We also note that error bars depict standard deviation across experimental trials, while for scatter plots without error bars, each dot represents a single trial. For further experimental documentation, we point the reader to %
\iftoggle{arxiv}{%
Appendix \ref{secExperimentDetails}. %
}{%
Appendix J in \cite{fullPaper}. %
}%
\new{In particular, 
\iftoggle{arxiv}{%
Table \ref{tabParam} in Appendix \ref{secExperimentDetails} %
}{%
Table 3 in \cite{fullPaper} %
}%
documents algorithmic parameters used. We chose these parameters so the primitive algorithms \texttt{FW-BW-MCMC} and \texttt{Bidirectional-PPR} yield similar accuracy ($\approx 10 \%$ relative error) while balancing runtime between the DP and MCMC stages of the algorithm in the single pair case. Note the analysis in Appendix \ref{secSinglePairComparison} shows that balancing runtime in this manner minimizes overall complexity; hence, for both algorithms, our chosen parameters optimize runtime subject to an accuracy constraint, providing a fair comparison. Finally, the implementation of our algorithms is available at https://github.com/danielvial/clusteringPpr.
}

\subsection{Synthetic data}

\subsubsection{Scalar estimation} \label{secExpSynScal}

We first compare \texttt{FW-BW-MCMC} with \texttt{Bidirectional-PPR} when computing $\pi_s(t)\ \forall\ (s,t) \in S \times T$ as $|S|$ and $|T|$ grow on \texttt{Direct-ER}. More specifically, for \texttt{FW-BW-MCMC} we use the $\| D^{-1} r^s \|_{\infty} \leq r_{\max}^s$ forward DP scheme as in \texttt{FW-BW-MCMC-Practical}, sample walks using the scheme from Section \ref{secManySource}, and use Algorithm \ref{algApproxContMany} for backward DP; for \texttt{Bidirectional-PPR}, we sample walks separately from each $s \in S$ and run backward DP separately for each $t \in T$. Results are shown in Fig.\ \ref{figManySTplots}. Note the number of random walks sampled and number of backward DP iterations grow more slowly with $|S| = |T|$ using \texttt{FW-BW-MCMC}, due to the accelerations proposed in Sections \ref{secManySource} and \ref{secManyTarget}, respectively. As a result, runtime grows more slowly using \texttt{FW-BW-MCMC}. %
In Fig.\ \ref{figManySTplots}, we also show the clustering quantities \eqref{eqSourceClustQuant} and \eqref{eqTargetClustQuant}. We observe the source clustering quantity $\| \Sigma \|_{\infty,1}$ has a concave shape, which corresponds to the apparent sublinear growth of random walks as $|S|$ grows. Additionally, the target clustering quantity $c_T$ has a convex shape; since backward DP iteration \textit{savings} scale with $c_T$, we expect DP iterations to correspondingly ``flatten", which indeed occurs. These observations empirically validate the key insights of Section \ref{secManyPairScalar}: namely, that the estimation schemes proposed have complexities that scale with the identified clustering quantities $\| \Sigma \|_{\infty,1}$ and $c_T$. We also plot $\srank(\Pi(S,T))$ on the runtime plot; note it appears to flatten along with runtime as $|S|, |T|$ grow. \new{Finally, these plots remain qualitatively similar as $n$ grows, while the improvement of our scheme over the existing one increases; see %
\iftoggle{arxiv}{%
Appendix \ref{secExperimentDetails}.%
}{%
Appendix J in \cite{fullPaper}.%
}%
}

\begin{figure}
\centering
\includegraphics[height=\plotHeight]{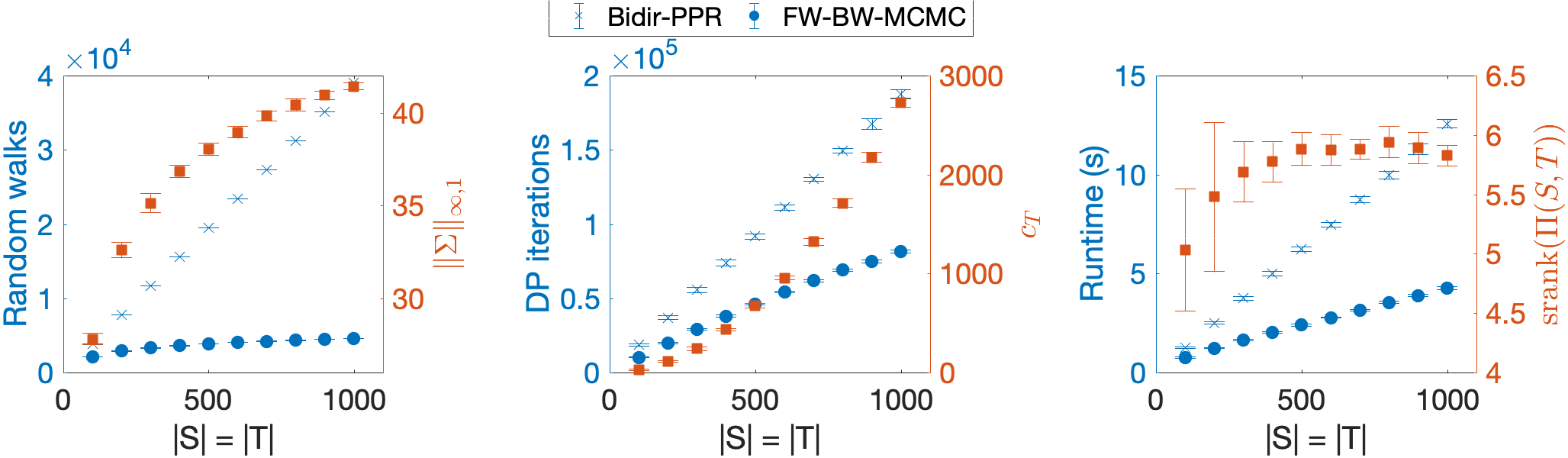}
\caption{On \texttt{Direct-ER}, random walks, backward DP iterations, and runtime scale more slowly in $|S|, |T|$ for our method \texttt{FW-BW-MCMC} when compared to the existing method \texttt{Bidir-PPR}.} \label{figManySTplots}
\end{figure}

\begin{figure}
\centering
\includegraphics[height=\plotHeight]{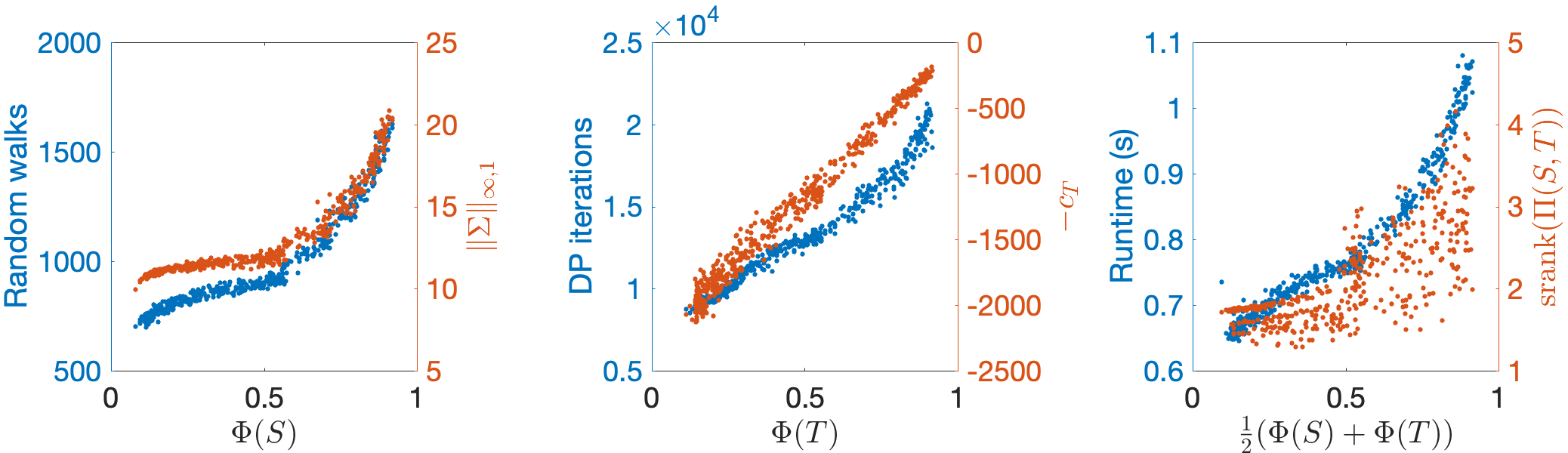}
\caption{When clustering is significant, fewer random walks and backward DP iterations yield faster runtime for our method on \texttt{Direct-SBM}; additionally, our clustering measures roughly scale with conductance.} \label{figClusteringPlot}
\end{figure}

Next, to further examine the effect of clustering, we use \texttt{Direct-SBM}. 
We fix $|S| = |T| = 100$ and sample $S$ and $T$ from decreasingly clustered sets via the following scheme: we first sample $S,T$ from a single community, we then sample $S,T$ from two communities, etc., until we sample $S,T$ from the entire graph, allowing us to observe a wide range of clustering. As in the previous experiment, we are interested in how algorithmic performance relates to $\| \Sigma \|_{\infty,1}$ and $c_T$. Here, we also compare these quantities to a clustering measure commonly used in the graph theory literature (see e.g.\ the aforementioned \cite{andersen2006local}), \textit{conductance}, defined for $U \subset V$ as
\begin{equation} \label{eqConductance}
\Phi(U) = \frac{\sum_{i \in U, j \notin U} A_{ij} }{ \min \{ \sum_{u \in U} d_{\textrm{out}}(u) , \sum_{u \notin U} d_{\textrm{out}}(u) \} } .
\end{equation}
In Fig.\ \ref{figClusteringPlot}, we observe fewer random walks are sampled when $\Phi(S)$ is small (when $S$ is significantly clustered); similarly, the backward DP converges in fewer iterations when $\Phi(T)$ is small (when $T$ is significantly clustered). Furthermore, Fig.\ \ref{figClusteringPlot} shows that $\| \Sigma \|_{\infty,1}$ grows with $\Phi(S)$ and $-c_T$ grows with $\Phi(T)$. In short, our identified clustering quantities behave similar to conductance. In the runtime plot, we again show $\srank(\Pi(S,T))$ as a measure of overall complexity; this quantity (roughly) grows with the average conductance $\frac{1}{2} (\Phi(S)+\Phi(T))$, as does runtime. 

\subsubsection{Matrix approximation} \label{secExpSynMat}

We now document performance of our matrix approximation scheme (Algorithm \ref{algMatrixEstimator}) using \texttt{Direct-SBM} and the $S,T$ sampling strategy from the previous experiment. We compare three cases: $\sigma = \sigma_{\max}$ with $w \propto \| \Sigma \|_{\infty,1}$, $\sigma = \sigma_{\avg}$ with $w \propto \sqrt{ l\ \srank(\Pi(S,T)) }$, and $\sigma = \sigma_{\avg}$ with $w \propto \sqrt{l\ \srank(P_T(S,:)+P_S^{\trans} R_T)}$. These cases are motivated by Theorem \ref{STATE_MATRIX_ACC}, which states that the sample requirements for $\sigma = \sigma_{\max}$ and $\sigma = \sigma_{\avg}$ are $\| \Sigma \|_{\infty,1}$ and $\sqrt{l\ \srank(\Pi(S,T))}$, respectively (neglecting common factors); additionally, since $\srank(\Pi(S,T))$ is unknown in practice, we proposed using $\srank(P_T(S,:)+P_S^{\trans} R_T)$ as a surrogate in the discussion following the theorem. Results are shown in Fig.\ \ref{figMatrixApprox}. Observe that for all three cases, fewer walks are sampled when $S$ and $T$ are clustered (i.e.\ when $\frac{1}{2} (\Phi(S) + \Phi(T))$ is small; nevertheless, error remains roughly constant (in fact, when clustering is present, error is somewhat lower despite fewer walks being sampled). Further, we observe $\sigma_{\max}$ and $\sigma_{\avg}$ have similar performance, in terms of complexity and accuracy. Finally, we note the results for the $\srank(\Pi(S,T))$ and $\srank(P_T(S,:)+P_S^{\trans} R_T)$ cases are quite similar, suggesting that $\srank(P_T(S,:)+P_S^{\trans} R_T)$ is an appropriate surrogate for $\srank(\Pi(S,T))$.

\begin{figure}
\centering
\includegraphics[height=\plotHeight]{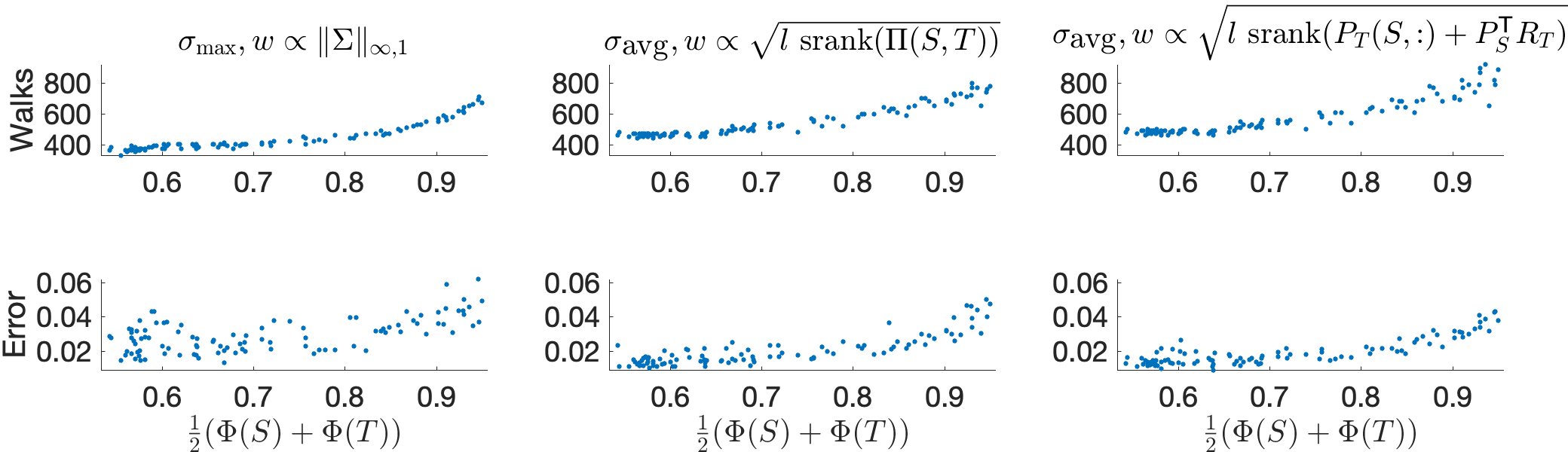}
\caption{On \texttt{Direct-SBM}, our matrix approximation schemes are most efficient when clustering is significant; additionally, the surrogate $\srank(P_T(S,:)+P_S^{\trans} R_T))$ performs similar to $\srank(\Pi(S,T))$.} \label{figMatrixApprox}
\end{figure}

\subsection{Real data} \label{secExpReal}

\subsubsection{Scalar estimation} \label{secExpRealScal}

We next compare \texttt{FW-BW-MCMC} with \texttt{Bidirectional-PPR} as in Section \ref{secExpSynScal}, but here using real datasets. We fix $|S| = |T| = 1000$ and randomly sample $S, T$ using two different schemes: sampling uniformly among all nodes and using an algorithm described in %
\iftoggle{arxiv}{%
Appendix \ref{secExperimentDetails} %
}{%
Appendix J of \cite{fullPaper} %
}%
to build clustered subsets of nodes; we find these schemes typically give conductance values $\approx 0.99$ and $\approx 0.5$, respectively, allowing us to observe two degrees of clustering. In Fig.\ \ref{figRealMultiPair}, we show random walk count, DP iteration count, and runtime for our method relative to the corresponding values using \texttt{Bidirectional-PPR}. Averaging across the diverse set of graphs considered, our method is approximately 1.4 times faster in the uniform case and 2.9 times faster in the clustered case, highlighting the efficiency of our algorithms and the impact of clustering on their performance. Additionally, we note our method is at least twice as fast for all datasets in the clustered case. For the same experiment, we also show random walk count (normalized to $w$) and the number of \texttt{Merge} updates (i.e.\ the number of DP iterations saved when compared to existing methods) in Fig.\ \ref{figRealMultiPairDetails}. From Theorem \ref{STATE_UNIFIED_MCMC_COMP} and Proposition \ref{STATE_BW_ITER_MERGE}, we expect these quantities to scale linearly with the identified clustering quantities $\| \Sigma \|_{\infty,1}$ and $c_T$, respectively; from Fig.\ \ref{figRealMultiPairDetails}, we observe this scaling roughly occurs. This verifies our analysis empirically on real datasets.

\begin{figure}
\centering
\includegraphics[height=\plotHeight]{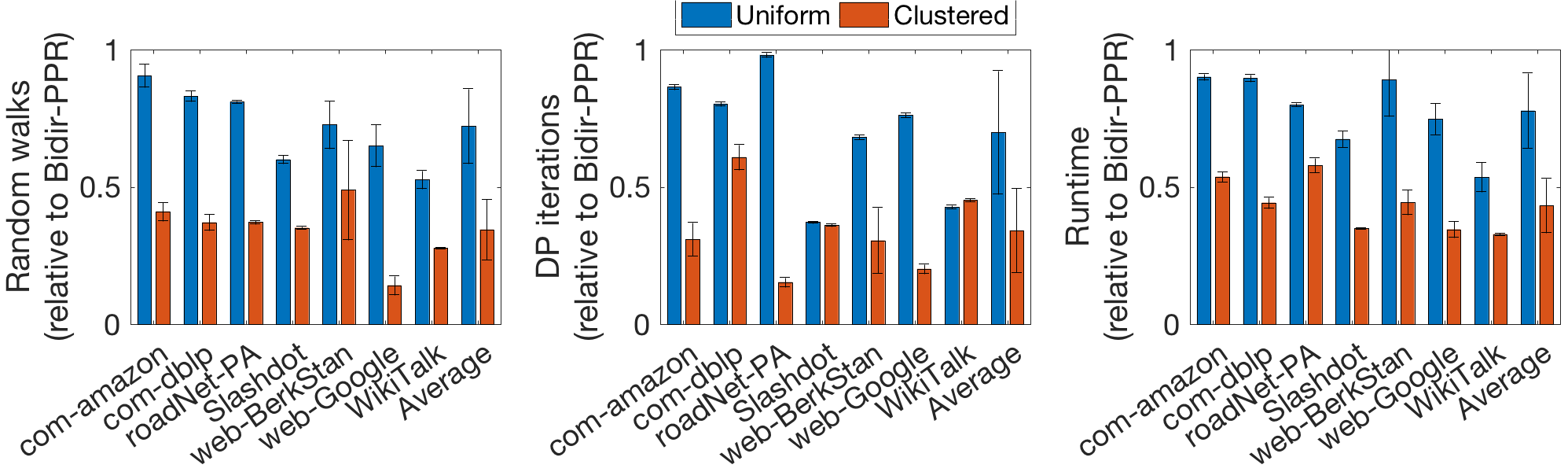}
\caption{On real graphs, our scalar methods are typically 1.4 and 2.9 times faster than existing methods when $S,T$ are chosen uniformly and clustered, respectively, due to fewer random walks and DP iterations.} \label{figRealMultiPair}
\end{figure}
\begin{figure}
\centering
\includegraphics[height=\plotHeight]{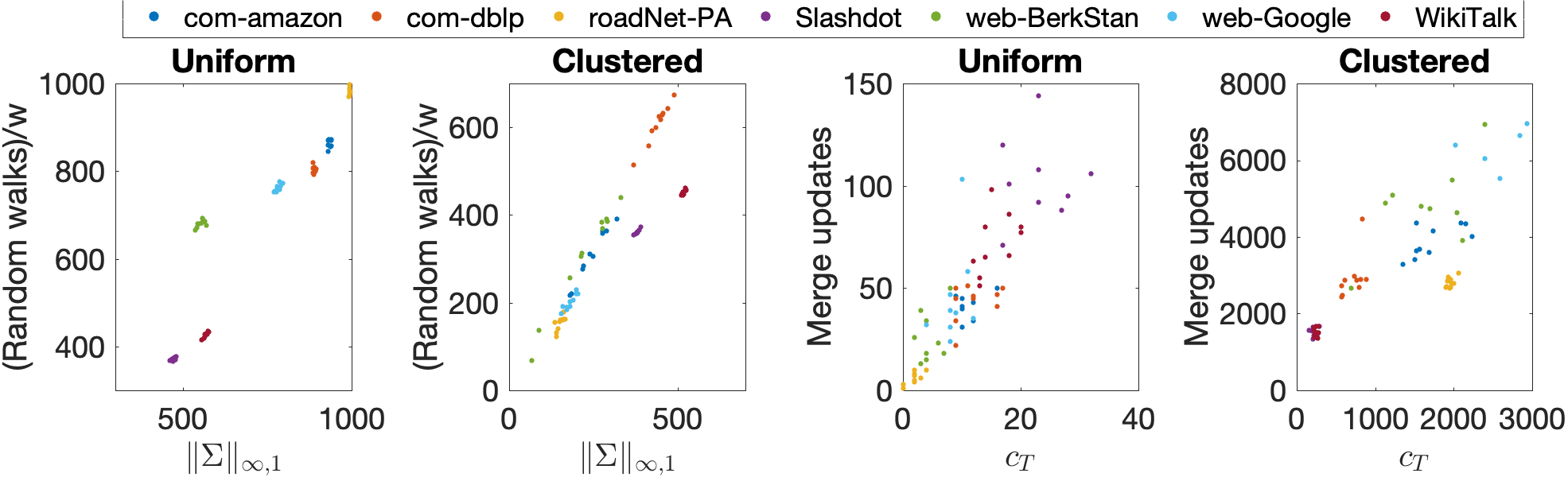}
\caption{On real graphs, random walks and \texttt{Merge} updates scale with clustering quantities $\| \Sigma \|_{\infty,1}$ and $c_T$, empirically validating the analysis of Section \ref{secManyPairScalar} across diverse datasets.} \label{figRealMultiPairDetails}
\end{figure}

\subsubsection{Matrix approximation} \label{secExpRealMat}

Finally, we test our matrix approximation scheme (Algorithm \ref{algMatrixEstimator}) on real graphs. Here we also compare to a baseline method that does not leverage clustering among targets and sources. In particular, we run backward DP separately for each target, rather than using the accelerated scheme as in Algorithm \ref{algMatrixEstimator}. Additionally, the baseline method uses no forward DP, i.e.\ we set $r^s_{\max} = 1$ in Algorithm \ref{algMatrixEstimator}, so that $p^s = 0, r^s = e_s\ \forall\ s \in S$. Note that, in this case, both the $\sigma_{\max}$ and $\sigma_{\avg}$ schemes reduce to sampling $\mu_i \sim S$ uniformly, sampling $\nu_i \sim \pi_{\nu_i}$ using a random walk, and estimating $\Pi(S,T)$ as $\hat{\Pi}(S,T) = P_T(S,:) + \frac{1}{w} \sum_{i=1}^w X_i$, where $X_i = [ e_{s_1}\  e_{s_2}\ \cdots\ e_{s_l} ]^{\trans} e_{\mu_i} e_{\nu_i}^{\trans} R_T$ is an unbiased estimate of $\Pi(S,:) R_T$. We reemphasize that walks are \textit{not} shared among sources for this baseline scheme, i.e.\ clustering among sources is not leveraged to improve performance. For the baseline scheme, we set $w \propto l$, and we compare performance to the $\sigma_{\max}$ scheme with $w \propto \| \Sigma \|_{\infty,1}$ and the $\sigma_{\avg}$ scheme with $w \propto \sqrt{l\ \srank(P_T(S,:)+P_S^{\trans} R_T)}$. Results are shown in Fig.\ \ref{figRealMat}, with quantities shown for the $\sigma_{\max}$ and $\sigma_{\avg}$ schemes relative to the baseline scheme. Averaging across datasets, the $\sigma_{\max}$ and $\sigma_{\avg}$ schemes are over twice as fast as the baseline scheme when $S,T$ are chosen uniformly and 3.4 times faster when $S,T$ are clustered; additionally, the accuracy of both schemes is comparable to the baseline across datasets (and slightly better on average). We also note both our schemes are at least twice as fast as the baseline for all graphs in the clustered case.

\begin{figure}
\centering
\includegraphics[height=\plotHeight]{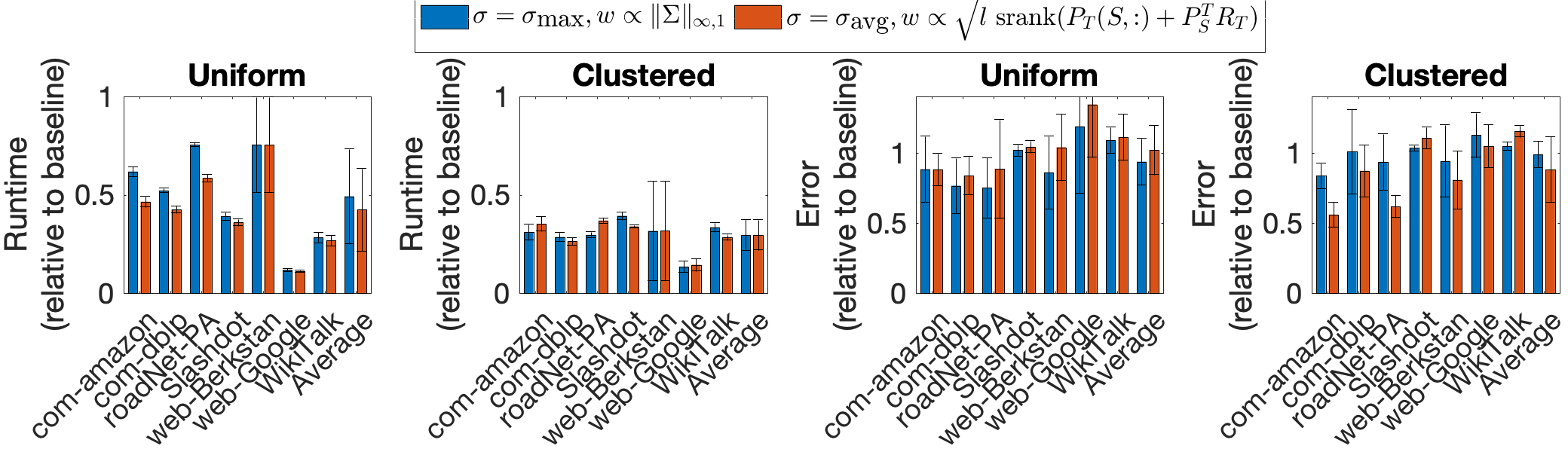}
\caption{On real graphs, our matrix approximation schemes are significantly faster than the baseline method (which uses no forward DP) with comparable accuracy; this is most notable when $S,T$ are clustered.} \label{figRealMat}
\end{figure}

\section{Application: distributed random walk sampling} \label{secDistributed}

Thus far, our key finding has been that PPR estimation complexity scales with quantities that describe clustering among sources and/or targets. In this section, we demonstrate one application of these findings; namely, that these findings can be used to efficiently estimate $\{ \pi_s \}_{s \in S}$ online when several machines are available and when offline storage is permitted. More specifically, we consider a natural distributed computational setting with the following features:
\begin{itemize}
\item $k$ machines are available for parallel computation and a central machine is available to facilitate the parallel computation (for simplicity, we assume $k \in \{|S|, |S|/2, |S|/3, \ldots \}$)
\item $\{ p^t, r^t \}_{t \in V}$ have been precomputed via Algorithm \ref{algApproxCont} and are stored offline
\end{itemize}

Using the existing method \texttt{Bidirectional-PPR} as a primitive, a baseline strategy for this estimation task is as follows: arbitrarily partition $S$ into $k$ subsets of size $|S|/k$, use the $i$-th machine to sample random walks from each source $s$ belonging to the $i$-th subset, and estimate $\pi_s$ using the endpoints of walks from $s$ and $\{ p^t, r^t \}_{t \in V}$ (as in the primitive method \texttt{Bidirectional-PPR}). 

Our goal is to devise a strategy more efficient than this baseline. In particular, we propose the following approach. First, we arbitrarily partition $S$ into $k$ subsets of size $|S|/k$, and we use the $i$-th machine to run forward DP (Algorithm \ref{algApproxPR}) for each source $s$ belonging to the $i$-th subset. Second, we use the central machine to construct another partition $\{ S_i \}_{i=1}^k$ of $S$, in a manner we discuss shortly. Third, we use the $i$-th machine to run the accelerated source stage from Section \ref{secManySource} for the subset of sources $S_i$. Finally, we estimate $\pi_s$ as in the primitive method \texttt{FW-BW-MCMC}.

It remains to specify how to construct the partition $\{ S_i \}_{i=1}^k$. For this, we turn to Theorem \ref{STATE_UNIFIED_MCMC_COMP} and the results of Section \ref{secExperiments}, which indicate that the number of random walks sampled on the $i$-th machine scales with $\| \Sigma_{S_i} \|_{\infty,1}$, where $\Sigma_{S_i}$ is the matrix with rows $\{ \sigma_s \}_{s \in S_i}$. Hence, as the random walk stage in our approach runs in parallel across $i$, the runtime of the this stage scales with 
 \begin{equation} \label{eqObjectiveOriginal}
\max_{i \in \{1,\ldots,k\} } \| \Sigma_{S_i} \|_{\infty,1} .
\end{equation}
Our goal is thus to construct the partition $\{ S_i \}_{i=1}^k$ so as to minimize \eqref{eqObjectiveOriginal}. However, as this is a combinatorial optimization problem, we devise a heuristic method to approximate the solution. To simplify the discussion of this  method, we introduce some notation. For $S' \subset S$, let $\sigma_{S'}$ be s.t.\ $\sigma_{S'}(v) = \max_{s' \in S'} \sigma_{s'}(v)\ \forall\ v \in V$; note  $\| \Sigma_{S'} \|_{\infty,1} = \| \sigma_{S'} \|_1$. For $S' \subset S$ and $s \in S \setminus S'$, let
\begin{align}
d(s,S')  = \sum_{v \in V} \max \left\{ \sigma_s(v) - \sigma_{S'}(v) , 0 \right\} . \label{eqDSdefn} 
\end{align}
It is straightforward to show \eqref{eqDSderiv}, i.e.\ \eqref{eqDSdefn} gives the increase in $\| \sigma_{S'} \|_1$  if $s$ is added to $S'$. 
\begin{equation}\label{eqDSderiv}
d ( s, S' ) = \| \sigma_{ S' \cup \{ s \} } \|_1 -  \| \sigma_{S'} \|_1 .
\end{equation}
With this notation in place, we may restate the objective function \eqref{eqObjectiveOriginal} as
 \begin{equation}\label{eqObjective}
\max_{i \in \{1,\ldots,k\} } \| \sigma_{S_i} \|_1 .
\end{equation}

Our heuristic method to approximate the minimizer of \eqref{eqObjective} proceeds as follows. First, we assign one node to each $S_i$, $i \in \{1,\ldots,k\}$, using an initialization method similar to $k$-means++ \cite{arthur2007k}: we choose the $i$-th of these nodes with probability proportional to its distance from the first $(i-1)$ of them, in hopes of choosing initial nodes with $\sigma_s$ vectors far apart. Next, we iteratively assign the remaining $|S|-k$ nodes to some $S_j$. In particular, we assign $s$ to $S_j$ such that $d(s,S_j) + \| \sigma_{S_j} \|_1$ is minimal; from \eqref{eqDSderiv}, this can be viewed as minimizing the increase in the objective function \eqref{eqObjective} incurred by assigning $s$ to some $S_j$. This heuristic method is formally defined in Algorithm \ref{algHeuristic}.

\begin{algorithm}
\caption{ $\{ S_i \}_{i=1}^k$ = \texttt{Source-Partition} ($\{ \sigma_s \}_{s \in S}$, $k$)  } \label{algHeuristic}
 Draw $s \sim S$ uniformly, set $S_1 = \{ s \}$, $\sigma_{S_1} = \sigma_s$; set $S_i = \emptyset\ \forall\ i \in \{2,\ldots,k\}$ \\
\For{$i = 2$ \KwTo $k$}{
	 Draw $s \sim S$ with probability proportional to $\min_{j \in \{1,\ldots,i-1\}} \| \sigma_s - \sigma_{S_j} \|_1$; set $S_i = \{ s \}$, $\sigma_{S_i} = \sigma_s$
}
\For{$i = k+1$ \KwTo $|S|$}{
	 Choose any $s \in S \setminus ( \cup_{j=1}^k S_j )$ (any $s$ not yet assigned); compute $d(s,S_j)\ \forall\ j \in \{1,\ldots,k\}$ \\
	 Let $j^* \in \argmin_j  d(s,S_j) + \| \sigma_{S_j} \|_1$ , $\sigma_{S_{j^*}}(v) = \max \{ \sigma_{S_{j^*}} (v), \sigma_s(v)   \}\ \forall\ v \in V$, $S_{j^*} = S_{j^*} \cup \{ s \}$
}
\end{algorithm}

We now empirically compare our approach with the baseline scheme. For this experiment, we set $S = \{ \tilde{S}_i \}_{i = 1}^k$, where each $\tilde{S}_i$ is a clustered subset of nodes constructed as in Section \ref{secExperiments} (with $k = 10$ and $| \tilde{S}_i | = 100\ \forall\ i$). This yields a set of sources $S$ that is not highly clustered itself, but that contains $k$ subsets that are densely connected internally and sparsely connected to other subsets. In addition to comparing to the baseline, we also test the performance of an ``oracle'' scheme, which knows the clustering information of the input set $S$. More specifically, the oracle scheme proceeds in the same manner as our scheme, except instead of using Algorithm \ref{algHeuristic} to construct the partition $\{ S_i \}_{i = 1}^k$, it simply sets $S_i = \tilde{S}_i\ \forall\ i \in \{1,2,\ldots,k\}$. Put differently, while the heuristic scheme attempts to learn an assignment of sources to machines for which each machine is assigned a clustered set of sources (in the sense that \eqref{eqObjective} is minimal), the oracle scheme knows such an assignment \textit{a priori}.

Results for this experiment are shown in Fig.\ \ref{figDistSetting}, using the set of real graphs from Section \ref{secExperiments}. Averaging across graphs, the oracle and heuristic methods are roughly 1.8 and 2.2 faster than the baseline scheme, respectively (left). (Here total runtime is computed as maximum walk sampling time across machines for the baseline; sum of maximum forward DP time and maximum walk time for the oracle; and sum of maximum forward DP time, maximum walk time, and time to run Algorithm \ref{algHeuristic} for the heuristic.) Additionally, both methods sample approximately $\frac{1}{4}$ of the random walks sampled by the baseline scheme, across graphs (middle). Finally, the heuristic method typically produces a partition $\{ S_i \}_{i =1}^k$ of $S$ with objective function value \eqref{eqObjective} similar to that produced by the oracle method (right). Interestingly, the heuristic outperforms the oracle for several datasets. This suggests that the cluster information known by the oracle does not necessarily produce an optimal assignment of sources to machines; rather, the source clustering quantity $\| \sigma_{S_i} \|_1$ identified in Section \ref{secManySource} may be what truly dictates performance.

Before closing, we offer several remarks on this distributed setting. First, while we focused on the scalar estimation scheme from Section \ref{secManySource}, the framework extends to the $\sigma_{\max}$ matrix approximation scheme from Section \ref{secManyPairMatrix}. In particular, using the latter scheme in this setting would also involve construction of a partition so as to minimize \eqref{eqObjective}, per Theorem \ref{STATE_MATRIX_ACC}. For this reason, we expect the performance of this scheme to be similar to Fig.\ \ref{figDistSetting}. Second, we note that using the $\sigma_{\avg}$ matrix approximation scheme in this setting requires a partition that minimizes a different objective function. In %
\iftoggle{arxiv}{%
Appendix \ref{secAdditionalDist}, %
}{%
Appendix K of \cite{fullPaper}, %
}%
we present an algorithm to construct such a partition, as well as empirical results describing performance (in short, our scheme performs similarly to the oracle and noticeably outperforms the baseline, as in Fig.\ \ref{figDistSetting}). Third, we find in practice that our heuristic partitioning schemes naturally balance the number of sources assigned to each machine (see %
\iftoggle{arxiv}{%
Appendix \ref{secAdditionalDist}). %
}{%
Appendix K of \cite{fullPaper}). %
}%
Such balance is crucial in the performance of our scheme. This is because we require $\| \sigma_{S_i} \|_1 = O ( |S| / k )\ \forall\ i$ to perform as well as the baseline, which may in turn require an extreme degree of clustering if the partition is unbalanced (for example, if $|S_i| = O(|S|)$ for some $i$). It is worth noting that we also tried to partition $\{ \sigma_s \}_{s \in S}$ using $k$-means++ (an off-the-shelf vector partitioning algorithm), but this led to highly unbalanced assignments (and thus poor performance). Finally, we note than one limitation of our scheme is that, if $|S|,|T| = \Theta(n)$, Algorithm \ref{algHeuristic} essentially partitions the entire graph and thus may be slower than directly estimating PPR. However, we recall from Section \ref{secRelated} that our focus is $|S|,|T| = o(\sqrt{m}) = o(n)$, so this is not a concern. Indeed, for the Fig.\ \ref{figDistSetting} experiment, Algorithm \ref{algHeuristic} accounted for only 12\% of runtime (averaged across graphs).

\begin{figure}
\centering
\includegraphics[height=\plotHeight]{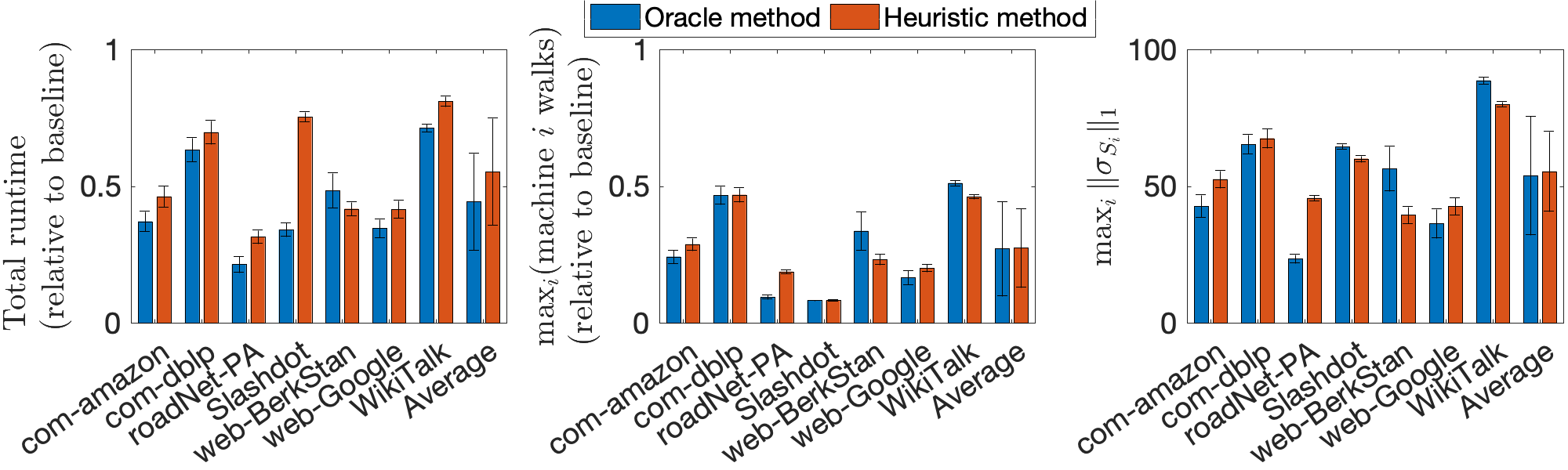}
\caption{In the distributed setting, our heuristic method is typically 1.8 times faster than the baseline, samples $\frac{1}{4}$ of the walks, and produces a low objective function value, with performance similar to an oracle method.} \label{figDistSetting}
\end{figure}

\section{Conclusions} \label{secConclusions}

In this work, we analyzed the relationship between PPR estimation complexity and clustering by devising estimation algorithms for many node pairs and showing the complexity of these methods scales with quantities interpretable as clustering measures. To demonstrate the utility of these findings, we considered a distributed setting for which the clustering quantities computed \emph{in situ} could be leveraged to reduce computation time. We believe this setting and the algorithms we designed for it are just one example of how our findings can be used to accelerate PPR estimation; hence, an avenue for future work would be to further explore applications of our results.

\newpage

\begin{appendices}

\section{Analysis of \texttt{FW-BW-MCMC} and comparison to \texttt{Bidirectional-PPR}} \label{secSinglePairComparison}

Here we state and prove the guarantees that were stated informally at the end of Section \ref{secSinglePairDescription}. We include the corresponding results for \texttt{Bidirectional-PPR} for comparison. We first present the accuracy guarantee, Theorem \ref{STATE_FWBWMCMC_ACC}. The idea is to bound relative error when $\pi_s(t) \geq \delta$ and to bound absolute error when $\pi_s(t) < \delta$. The authors of \cite{lofgren2016personalized} suggest choosing $\delta = O ( \frac{1}{n} )$. This choice dictates that we desire the relative bound when $t$'s PPR exceeds a uniform distribution over all nodes, which suggests that $t$ is ``significant'' to $s$ in this case. The proof applies the Chernoff bound to a variety of cases; this approach is similar \cite{lofgren2016personalized}, but we must address cases that do not arise in that work.

\begin{theorem} \label{STATE_FWBWMCMC_ACC} 
Fix minimum PPR threshold $\delta$, relative error tolerance $\epsilon$, and failure probability $p_{\textrm{fail}}$. For \texttt{FW-BW-MCMC}, assume the following hold:
\begin{equation}\label{eqFwBwMcmcAccAss}
\epsilon \in \left( 0,\tfrac{1}{\sqrt{2e}} \right) , \quad w = \tfrac{c r^s_{\max} r^t_{\max}}{\delta} , \quad c > \tfrac{3 (2e)^{1/3} \log( 2 / p_{\textrm{fail}} )}{\epsilon^{7/3}} .
\end{equation} 
For \texttt{Bidirectional-PPR}, assume the following hold:
\begin{equation}\label{eqBidirAccAss}
r^t_{\max} \in \left( \tfrac{2 e \delta}{\alpha \epsilon} , 1 \right), \quad w = \tfrac{c r^t_{\max}}{\delta} , \quad c > \tfrac{ 3 \log( 2 / p_{\textrm{fail}} ) }{ \epsilon^2 } .
\end{equation}
Then the estimate $\hat{\pi}_s(t)$ produced by either algorithm satisfies the following with probability $ \geq 1-p_{\textrm{fail}}$: 
\begin{equation} \label{eqBidirTypeErrGuar}
| \pi_s(t) - \hat{\pi}_s(t) | \leq \begin{cases} \epsilon \pi_s(t), & \pi_s(t) \geq \delta \textrm{ (significant case)} \\ 2 e \delta , & \pi_s(t) < \delta \textrm{ (insignificant case)} \end{cases} .
\end{equation}
\end{theorem}
\begin{proof}
See \cite{lofgren2016personalized} for \texttt{Bidirectional-PPR}; see Appendix \ref{PROOF_FWBWMCMC_ACC} for \texttt{FW-BW-MCMC}.
\end{proof}
From Theorem \ref{STATE_FWBWMCMC_ACC}, \texttt{FW-BW-MCMC} offers the same accuracy as \texttt{Bidirectional-PPR}. However, our assumptions on $\epsilon$ and $c$ are stronger than those required for \texttt{Bidirectional-PPR}. The first assumption is mild, since $\frac{1}{\sqrt{2e}} \approx 0.43$ and we typically desire a tighter relative error bound. The second affects complexity and will be discussed next. Note also that our guarantee holds for any $r^t_{\max} \in (0,1)$, while proving the theorem for \texttt{Bidirectional-PPR} requires a lower bound on $r^t_{\max}$.

Next, we have a worst-case complexity result in Theorem \ref{STATE_FWBWMCMC_COMP} (by worst case, we mean the result holds for when the algorithm is run for any $s,t \in V$). The idea is to choose $r^s_{\max}, r^t_{\max}$ to balance the complexity of the DP and MCMC stages of the algorithm. The result requires the additional assumption $m \delta <  \log(1/p_{\textrm{fail}} ) / \epsilon^2 $, which guarantees that these $r^s_{\max}, r^t_{\max}$ values lie in $(0,1)$. Note that with $\delta = O( \frac{1}{n} )$, this implies $m = O(n)$, i.e.\ nodes have constant degrees as $n$ grows.

\begin{theorem} \label{STATE_FWBWMCMC_COMP}
Fix minimum PPR threshold $\delta$, relative error tolerance $\epsilon$, and failure probability $p_{\textrm{fail}}$. Assume \eqref{eqFwBwMcmcAccAss}-\eqref{eqBidirAccAss} hold and $m \delta <  \log(1/p_{\textrm{fail}} ) / \epsilon^2$. Then setting $r^s_{\max} = r^t_{\max} = \frac{ m^{1/3} \delta^{1/3} \epsilon^{7/9} }{ ( \log(1/p_{\textrm{fail}}) )^{1/3} }$ in \texttt{FW-BW-MCMC} yields minimal complexity $O \left(  \frac{ m^{2/3} ( \log(1/p_{\textrm{fail}} ) )^{1/3} }{ \alpha \epsilon^{7/9} \delta^{1/3} }  \right)$. Furthermore, setting $r^t_{\max} = \frac{ \sqrt{m \delta} \epsilon }{ \sqrt{ \log(1/p_{\textrm{fail}}) } }$ in \texttt{Bidirectional-PPR} yields minimal complexity $O \left( \frac{ \sqrt{ m \log(1/ p_{\textrm{fail} } ) } }{ \alpha \epsilon \sqrt{\delta} } \right)$.
\end{theorem}
\begin{proof}
See Appendix \ref{PROOF_FWBWMCMC_COMP}.
\end{proof}
Note that, with $\delta = O(\frac{1}{n})$, so that $m = O(n)$, both algorithms have complexity linear in $n$, while \texttt{FW-BW-MCMC} has strictly better dependence on the parameters $p_{\textrm{fail}}$ and $\epsilon$.

Finally, we present an average-case complexity result for \texttt{FW-BW-MCMC-Practical} (Algorithm \ref{algOurEstimatorPractical}), which changes the termination criteria to $\| D^{-1} r^s \|_{\infty} \leq r_{\max}^s$ in the forward DP.
\begin{theorem}
For any $s \in V$ and for $t \sim V$ uniformly, \texttt{FW-BW-MCMC-Practical} produces an estimate satisfying the accuracy guarantee of Theorem \ref{STATE_FWBWMCMC_ACC} and has complexity $O \left(  \frac{\sqrt{m \log(1/p_{\textrm{fail}})}}{\sqrt{n \delta} \alpha \epsilon^{7/6} }  \right)$.
\end{theorem}
\begin{proof}
See Appendix \ref{appPracticalVersion}.
\end{proof}
With $\delta = O(\frac{1}{n})$, this establishes the $O(\sqrt{m})$ average case complexity claimed at the end of Section \ref{secSinglePairDescription}. The guarantee for \texttt{Bidirectional-PPR} in \cite{lofgren2016personalized} has $\epsilon$ instead of $\epsilon^{7/6}$ but is otherwise identical.

\section{Proof of Theorem \ref{STATE_FWBWMCMC_ACC}}  \label{PROOF_FWBWMCMC_ACC}

We will use the following result from \cite{dubhashi2009concentration}.

\begin{theorem} \label{thmChernoff}
(from Theorem 1.1 in \cite{dubhashi2009concentration}) Let $\{ Z_i \}$ be a set of independent random variables with $Z_i \in [0,1]\ \forall\ i$, and let $Z = \sum_i Z_i$. Then for any $\eta \in (0,1)$ and any $d > 2 e \E [ Z ]$,
\begin{gather}
\P [ Z > (1+\eta)  \E[Z] ] \leq \exp( -\eta^2 \E[Z] / 3 ) , \quad  \P [ Z < (1-\eta)  \E[Z] ] \leq \exp( -\eta^2 \E[Z] / 2 ) , \label{eqChernoff1} \\
\P [ Z > d ] \leq 2^{-d} . \label{eqChernoff2}
\end{gather}
\end{theorem}

To begin the proof, we define $Y_i = X_i / r_{\max}^t$ and $Y = \sum_{i=1}^w Y_i$, where $X_i$ is from Algorithm \ref{algOurEstimator}. Observe the $Y_i$'s are independent and $Y_i \in [0,1]$ (by the terminating condition of Algorithm \ref{algApproxCont}), so Theorem \ref{thmChernoff} applies for appropriate choices of $\eta$ and $d$. We also observe that \eqref{eqExpScalings} holds, which follows by linearity and $w = \frac{c r^s_{\max} r^t_{\max}}{\delta}$ in the statement of the theorem. 
\begin{equation} \label{eqExpScalings}
\E[Y] = \tfrac{w}{r^t_{\max}} \E [ X_i ] = \tfrac{c r^s_{\max}}{\delta} \E [ X_i ] . 
\end{equation}

We now turn to the case $\pi_s(t) \geq \delta$, for which we aim to show $\P [ | \hat{\pi}_s(t) - \pi_s(t)  | > \epsilon \pi_s(t) ] < p_{\textrm{fail}}\ \forall\ \epsilon \in ( 0 ,  \frac{1}{\sqrt{2e}})$. We will examine three sub-cases. The first two sub-cases depend on the constant $k := ( \frac{\epsilon}{2e} )^{1/3}$ (we motivate the choice of this constant at the conclusion of the proof). We also observe the following, which follows from the assumption $c > \frac{3 (2e)^{1/3} \log(2/p_{\textrm{fail}})}{\epsilon^{7/3}}$:
\begin{equation} \label{eqKCinequality}
\tfrac{ k }{3} = \tfrac{\epsilon }{6 e k^2} = \tfrac{ \epsilon^{1/3} }{ 3 (2e)^{1/3} } > \tfrac{\log(2/p_{\textrm{fail}})}{\epsilon^2 c} .
\end{equation}

For the first sub-case, assume $\E [Y] \geq kc$. Then we have the following:
\begin{align}
& \P [ | \hat{\pi}_s(t) - \pi_s(t)  | > \epsilon \pi_s(t) ]  \leq \P \left[ \left| \tfrac{1}{w} \tsum_{i=1}^w X_i -  \E [X_i ] \right| > \epsilon \E [X_i] \right]  \\
&  \quad = \P \left[ | Y - \E [ Y ] | > \epsilon \E [ Y ] \right]  \leq 2 \exp \left( - \epsilon^2 \E [Y] / 3 \right)  \leq 2 \exp \left( - \epsilon^2 k c / 3 \right)  < p_{\textrm{fail}}  .
\end{align}
Here the first inequality holds by definition of $\hat{\pi}_s(t)$ in Algorithm \ref{algOurEstimator} and the invariant \eqref{eqFwBwInvariant}; the equality holds by \eqref{eqExpScalings} and the definition of $Y$; the second inequality uses Theorem \ref{thmChernoff} (note $\epsilon < \frac{1}{\sqrt{2e}} < 1$, so \eqref{eqChernoff1} applies); and the final two inequalities hold by $\E[Y] \geq kc$ and \eqref{eqKCinequality}, respectively.

For the second sub-case, assume $\E [ Y ] \in [ \frac{\epsilon c}{2e} , kc)$. First, observe that by \eqref{eqExpScalings}, the assumption $\E [ Y ] < k c$, and the Algorithm \ref{algApproxPR} terminating condition,
\begin{equation}
\| r^s \|_1 \E [ X _i ] = \tfrac{\| r^s \|_1 \delta \E[Y]}{c r^s_{\max}} <  \tfrac{\| r^s \|_1  k \delta}{ r^s_{\max}} \leq k \delta .
\end{equation}
and so $\pi_s(t) \geq \| r^s \|_1 \E [ X _i ]  + (1-k) \delta$ (else, $\pi_s(t) < \delta$ by \eqref{eqFwBwInvariant}, a contradiction). We then write:
\begin{align}
& \P [ | \hat{\pi}_s(t) - \pi_s(t)  | > \epsilon \pi_s(t) ]  \leq \P \left[ \left| \tfrac{1}{w} \tsum_{i=1}^w X_i -  \E [X_i] \right| > \epsilon \left( \E [ X ] + \tfrac{ (1-k) \delta}{\|r^s\|_1} \right) \right]  \\
&  \quad = \P \left[ \left| Y - \E [ Y ] \right| > \epsilon \left( \E [ Y ] + \tfrac{  (1-k) \delta w}{\|r^s\|_1 r^t_{\max}} \right) \right]  \leq \P \left[ \left| Y - \E [ Y ] \right| > \epsilon \left( \E [ Y ] + (1-k) c \right) \right]  \\
&  \quad = \P \left[ \left| Y - \E [ Y ] \right| > \epsilon \left( \E [ Y ] + \left( \tfrac{1-k}{k} \right) k c \right) \right] < \P \left[ \left| Y - \E [ Y ] \right| > \tfrac{\epsilon}{k} \E [ Y ]  \right]  \\
&  \quad \leq 2 \exp \left( - \epsilon^2 \E [ Y ] / \left(3 k^2 \right) \right) < 2 \exp \left( - \epsilon^3 c / \left(6 e k^2 \right) \right)  < p_{\textrm{fail}}  .
\end{align}
Here the first inequality and first equality follow similar arguments as Case 1; the second inequality is by the Algorithm \ref{algApproxPR} terminating condition and $w = \frac{c r^s_{\max} r^t_{\max}}{\delta}$; the second equality simply multiplies and divides $k$; the third inequality holds by assumption $\E [ Y ] \in [ \frac{\epsilon c}{2e} , kc)$; the fourth inequality holds by Theorem \ref{thmChernoff} (note $\frac{\epsilon}{k} = \epsilon^{2/3} (2e)^{1/3} < 1$ by assumption $\epsilon < \frac{1}{\sqrt{2e}}$, so \eqref{eqChernoff1} applies); the fifth inequality follows from $\E [ Y ] \in [ \frac{\epsilon c}{2e} , kc)$; and the final inequality holds by \eqref{eqKCinequality}. Note we have assumed $1-k > 0$ in the third and fifth inequality; this follows from $\epsilon < \frac{1}{\sqrt{2e}}$. 

For the third and final sub-case, assume $\E[ Y ] < \frac{\epsilon c}{2e}$. We have the following:
\begin{align}
& \P [ | \hat{\pi}_s(t) - \pi_s(t)  | > \epsilon \pi_s(t) ] = \P \left[ \left| \tfrac{1}{w} \tsum_{i=1}^w X_i -  \E [X_i] \right| > \tfrac{\epsilon \pi_s(t)}{ \|r^s\|_1 } \right] = \P \left[ | Y - \E [ Y ] | >  \tfrac{\epsilon \pi_s(t) w}{ \|r^s\|_1 r^t_{\max} } \right] \\
& \quad  \leq \P \left[ | Y - \E [ Y ] | > \tfrac{\epsilon \delta w}{ r^s_{\max} r^t_{\max} } \right] = \P \left[ | Y - \E [Y] | > \epsilon c \right] \leq \P [ Y > \epsilon c ]  \leq 2^{-\epsilon c} . \label{eqCase3_application}
\end{align}
Here the first three equalities and first inequality follow similar arguments as previous cases; the penultimate inequality holds since $\{ | Y - \E [ Y ] | > \epsilon c \} \subset \{ Y > \epsilon c \}$ when $Y \geq \E [ Y ]$, whereas $\{ | Y - \E [ Y ] | > \epsilon c \} \subset \{ \E [ Y ] > \epsilon c \} \subset \{ 2 e \E [ Y ] > \epsilon c \} = \emptyset$ when $Y < \E [ Y ]$; and the final inequality holds by Theorem \ref{thmChernoff}; note $\epsilon c > 2 e \E[Y]$ by assumption, so \eqref{eqChernoff2} applies. Next, we observe
\begin{align}
\epsilon c & >  \tfrac{ 6 e \log ( 1 / p_{\textrm{fail}}  ) }{ ( \sqrt{2e} \epsilon )^{4/3} } > 6 e \log ( 1 / p_{\textrm{fail}} ) = \tfrac{6 e}{ \log_2(e) } \log_2( 1 / p_{\textrm{fail}} ) > \log_2( 1 / p_{\textrm{fail}} ) . \label{eqCase3_logToLog2}
\end{align}
where the first two inequalities hold by $c > \frac{3 (2e)^{1/3} \log(1/p_{\textrm{fail}})}{\epsilon^{7/3}}$ and $\epsilon < \frac{1}{\sqrt{2e}}$, and the final inequality holds since $e < 4 \Rightarrow \log_2(e) < 2 \Rightarrow \frac{6 e}{ \log_2(e) } > \frac{3 e}{2} > 1$. Combining \eqref{eqCase3_application} and \eqref{eqCase3_logToLog2} completes Case 3.

Finally, we note the bounds in Cases 1 and 3 grow with decreasing  and increasing $k$, respectively. Hence, our choice of $k = ( \frac{\epsilon}{2e} )^{1/3}$ comes from equating the two to minimize failure probability.

We now turn to the case $\pi_s(t) < \delta$. Observe that by $\pi_s(t) < \delta$ and the invariant \eqref{eqFwBwInvariant}, $\| r^s \|_1 \E [X_i] < \delta$. By \eqref{eqExpScalings}, this implies $2 e \E [Y] <  \frac{2 e w \delta}{r^t_{\max} \| r^s \|_1} =: b$. Then
\begin{align}
& \P [ |  \hat{\pi}_s(t) - \pi_s(t)  | > 2 e \delta ] = \P \left[ \left| \tfrac{1}{w} \tsum_{i = 1}^w X_i -  \E [ X_i ] \right| > \tfrac{2 e \delta}{ \| r^s \|_1 } \right] = \P \left[ | Y - \E [Y] | > \tfrac{2 e \delta w}{ \| r^s \|_1 r^t_{\max}}   \right]  \\
& \quad = \P \left[ | Y - \E [Y] | > b \right] \leq P [ Y > b ] \leq 2^{-b} . \label{eqCase4_application} 
\end{align}
Here the equalities follow similar steps as previous cases, the first inequality holds by the same argument in the Case 3 analysis, and the final inequality holds by Theorem \ref{thmChernoff} (note \eqref{eqChernoff2} applies since $b > 2 e \E [ Y ]$). We also observe
\begin{equation} \label{eqCase4_logToLog2}
b = \tfrac{2 e w \delta}{r^t_{\max} \| r^s \|_1} > \tfrac{2 e w \delta}{r^t_{\max} r^s_{\max}} = 2 e c > \epsilon c > \log_2 ( 1 / p_{\textrm{fail}} ) ,
\end{equation}
where the first inequality is by the Algorithm \ref{algApproxPR} terminating condition, the second inequality holds since $2 e > 1 > \epsilon$, and the third inequality follows from \eqref{eqCase3_logToLog2}; the equalities are by definition. Finally, we combine \eqref{eqCase4_application} and \eqref{eqCase4_logToLog2} to complete the proof.

\section{Proof of Theorem \ref{STATE_FWBWMCMC_COMP}} \label{PROOF_FWBWMCMC_COMP}

The complexity of Algorithm \ref{algOurEstimator} is the total complexity of Algorithm \ref{algApproxCont}, Algorithm \ref{algApproxPR}, and the random walks. Below, we show Algorithms \ref{algApproxCont} and \ref{algApproxPR} have complexity $\frac{m}{\alpha r^t_{\max}}$ and $\frac{m}{\alpha r^s_{\max}}$, respectively (using arguments from \cite{andersen2008local} and \cite{andersen2006local}). Furthermore, the complexity of the random walk stage is $O ( \frac{r^s_{\max} r^t_{\max} \log( 1 / p_{\textrm{fail}} ) }{ \alpha \delta \epsilon^{7/3} } )$, where $\frac{1}{\alpha}$ is the expected complexity of sampling a single random walk, and where the remaining factors give the number of walks required (recall in the statement of the theorem we assume \eqref{eqFwBwMcmcAccAss} holds). Hence, the complexity of Algorithm \ref{algOurEstimator} is $O(C(r^s_{\max} r^t_{\max}) /\alpha)$, where
\begin{equation}\label{eqOurEstimatorComp}
C(r^s_{\max} r^t_{\max}) = \tfrac{m}{ r^s_{\max}} + \tfrac{r^s_{\max} r^t_{\max} \log( 1 / p_{\textrm{fail}} ) }{ \delta \epsilon^{7/3} } + \tfrac{m}{ r^t_{\max}}   . 
\end{equation}

We now aim choose $r^s_{\max}, r^t_{\max}$ to minimize $O(C(r^s_{\max} r^t_{\max}) /\alpha)$, or equivalently, to minimize $C(r^s_{\max} r^t_{\max})$. For this, we let $K = \frac{\log( 1 / p_{\textrm{fail}} )}{ \delta \epsilon^{7/3} } > 0$ and note $\tfrac{\partial C}{\partial r^s_{\max}} = K r^t_{\max} - \tfrac{m}{ ( r^s_{\max} )^2 }  = 0$ if and only if $( r^s_{\max} )^2 r^t_{\max} = \tfrac{m}{K}$, and similarly, $\frac{\partial C}{\partial r^t_{\max}} = 0$ if and only if $( r^t_{\max} )^2 r^s_{\max} = \frac{m}{K}$; hence, $( \frac{m}{K} )^{1/3},( \frac{m}{K} )^{1/3}$ is a stationary point of $C(r^s_{\max},r^t_{\max})$. To verify this is a minimizer, we observe
\begin{align}
\begin{bmatrix} \frac{\partial^2 C}{\partial (r^s_{\max})^2} & \frac{\partial^2 C}{\partial r^s_{\max} r^t_{\max}} \\ \frac{\partial^2 C}{\partial r^t_{\max} r^s_{\max}} & \frac{\partial^2 C}{\partial (r^t_{\max})^2} \end{bmatrix} = \begin{bmatrix} \frac{2 m}{ (r^s_{\max} )^3 } & K \\ K & \frac{2 m}{ (r^t_{\max} )^3 }  \end{bmatrix} ,
\end{align}
from which it follows that the Hessian of $C$ evaluated at $r^s_{\max} = r^t_{\max} = ( \frac{m}{K} )^{1/3}$ is $K ( I + 1 1^{\trans} )$. This is positive definite, since for any vector $z \neq 0$,
\begin{equation}
z^{\trans} \left( K \left( I + 1 1^{\trans} \right) \right) z = K \left( \| z \|_2^2 + \left( z^{\trans} 1 \right)^2 \right) > 0 .
\end{equation}
To summarize, we have shown $r^s_{\max} = r^t_{\max} = ( \frac{m}{K} )^{1/3}$ minimizes $C(r^s_{\max},r^t_{\max})$ and hence minimizes the complexity of Algorithm \ref{algOurEstimator}. This establishes that the choice of $r^s_{\max}, r^t_{\max}$ in the statement of the theorem minimizes complexity. Finally, substituting $r^s_{\max} = r^t_{\max} = ( \frac{m}{K} )^{1/3}$ into \eqref{eqOurEstimatorComp} and dividing by $\alpha$ gives the complexity expression given in the theorem. Following the same approach establishes the Algorithm \texttt{Bidirectional-PPR} complexity bound given in the theorem.  

We return to bound the complexities of Algorithms \ref{algApproxCont} and \ref{algApproxPR}. For Algorithm \ref{algApproxCont}, we use an argument from \cite{andersen2008local}. First, let $v \in V$. From Algorithm \ref{algApproxCont}, $p^t(v)$ increases by at least $\alpha r^t_{\max}$ at each iteration for which $v^*  = v$. By the invariant \eqref{eqBwInvariant}, $p^t(v) \leq \pi_v(t)$. Taken together, $v^* = v$ for at most $\frac{\pi_v(t)}{\alpha r^t_{\max}}$ iterations. Furthermore, the complexity of each iteration for which $v^* = v$ is $d_{\textrm{in}}(v)$. Hence, the complexity of all iterations for which $v^* = v$ is bounded by $d_{\textrm{in}}(v) \frac{\pi_v(t)}{\alpha r^t_{\max}}$. Finally, the complexity of Algorithm \ref{algApproxCont} can be bounded by summing over all $v \in V$, i.e.\ $\sum_{v \in V} d_{\textrm{in}}(v) \frac{\pi_v(t)}{\alpha r^t_{\max}}  \leq \frac{1}{\alpha r^t_{\max}} \sum_{v \in V} d_{\textrm{in}}(v) = \frac{m}{\alpha r^t_{\max}}$.

We next turn to Algorithm \ref{algApproxPR}. As mentioned in the main text, Algorithm \ref{algApproxPR} changes the termination criteria from the algorithm originally defined in \cite{andersen2006local}; for clarify, we include the original definition in Algorithm \ref{algApproxPROriginal}. Here we use tilde marks to distinguish quantities from those in Algorithm \ref{algApproxPR}, and we explicitly indicate iteration number $k$ to improve clarity of the arguments to follow. Besides these notational changes, the only difference between Algorithms \ref{algApproxPR} and \ref{algApproxPROriginal} is the termination criteria.

With this notation in place, the complexity of Algorithm \ref{algApproxPROriginal} can be bounded as follows (using arguments from \cite{andersen2006local}). First, observe that for any iteration $k$,
\begin{equation}\label{eqFDPl1Decrease}
\| \tilde{r}^s_k \|_1 = \sum_{v \in V \setminus ( \{ v_k \} \cup N_{\textrm{out}}(v_k) ) } \tilde{r}^s_{k-1}(v)  + \sum_{v \in N_{\textrm{out}} ( v_k ) } \left(  \tilde{r}^s_{k-1}(v) + \tfrac{ (1-\alpha) \tilde{r}^s_{k-1}(v_k) }{ d_{\textrm{out}}(v_k) } \right) = \| \tilde{r}^s_{k-1} \|_1 - \alpha \tilde{r}^s_{k-1}(v_k)  ,
\end{equation}
where the first equality holds via the iterative update in Algorithm \ref{algApproxPROriginal}. Next, let $k^*$ be the iteration at which Algorithm \ref{algApproxPROriginal} terminates. Then the complexity of the algorithm is $\sum_{k=1}^{k^*} d_{\textrm{out}}(v_k)$, and
\begin{align}
\tsum_{k=1}^{k^*} d_{\textrm{out}}(v_k) & = \tsum_{k=1}^{k^*} \tfrac{d_{\textrm{out}}(v_k)}{\tilde{r}^s_{k-1}(v_k)} \tilde{r}^s_{k-1}(v_k) < \tfrac{1}{\tilde{r}^s_{\max}} \tsum_{k=1}^{k^*} \tilde{r}^s_{k-1}(v_k) \\
& = \tfrac{1}{\alpha \tilde{r}^s_{\max}} \tsum_{k=1}^{k^*} \left( \| \tilde{r}^s_{k-1} \|_1 - \| \tilde{r}^s_{k} \|_1 \right) = t\frac{1}{\alpha \tilde{r}^s_{\max}} \left( \| \tilde{r}^s_{0} \|_1 - \| \tilde{r}^s_{k^*} \|_1 \right) \leq \tfrac{1}{\alpha \tilde{r}^s_{\max}} ,
\end{align}
where the first inequality holds since $\tilde{r}_{\max}^s < \| D^{-1} \tilde{r}^s_k \|_{\infty} = \frac{\tilde{r}^s_{k-1}(v_k)}{d_{\textrm{out}}(v_k)}$ for $k \leq k^*$ (i.e.\ for each $k$ until the algorithm terminates), the second equality holds by the previous display, and the final inequality holds since $\| \tilde{r}^s_{0} \|_1 = \| e_s \|_1 = 1$ and $\| \tilde{r}^s_{k^*} \|_1 \geq 0$ (the remaining steps are straightforward). 

Using this, we can bound the complexity of Algorithm \ref{algApproxPR}. First, observe that in Algorithm \ref{algApproxPR},
\begin{equation}
\| r^s \|_1 = \tsum_{v \in V} \tfrac{r^s(v)}{d_{\textrm{out}}(v)} d_{\textrm{out}}(v) \leq m \max_{v \in V} \frac{r^s(v)}{d_{\textrm{out}}(v)} = m \| D^{-1} r^s \|_{\infty} ,
\end{equation}
and so to guarantee termination of Algorithm \ref{algApproxPR} (i.e.\ to ensure $\| r^s \|_1 \leq r^s_{\max}$), it suffices to guarantee $\| D^{-1} r^s \|_{\infty} \leq \frac{r^s_{\max}}{m}$. But from the analysis of Algorithm \ref{algApproxPROriginal}, the complexity required to ensure $\| D^{-1} r^s \|_{\infty} \leq \frac{r^s_{\max}}{m}$ is $\frac{m}{\alpha r^s_{\max}}$; hence, the complexity of Algorithm \ref{algApproxPR} is bounded by $\frac{m}{\alpha r^s_{\max}}$ as well.

\begin{algorithm}
\caption{$(\tilde{p}^s,\tilde{r}^s) = \texttt{Approximate-PageRank-Original} (G,s,\alpha,\tilde{r}_{\max}^s)$} \label{algApproxPROriginal}
Set $k = 0, \tilde{p}^s_k = 0, \tilde{r}^s_k = e_s$ \\
\While{$\| D^{-1} \tilde{r}^s_k \|_{\infty} > \tilde{r}_{\max}^s$}{
	Set $k \leftarrow k+1$; let $v_k \in \arg \max_{v \in V} \tilde{r}^s_{k-1}(v) / d_{\textrm{out}}(v)$ \\
	Set $\tilde{r}^s_k(v) = \tilde{r}^s_{k-1}(v) + (1-\alpha) \tilde{r}^s_{k-1}(v_k) / d_{\textrm{out}}(v_k)$, $\tilde{p}^s_k(v) = \tilde{p}^s_{k-1}(v)$ $\forall\ v \in N_{\textrm{out}}(v_k)$ \\
%
	Set $\tilde{r}^s_k(v_k) = 0$, $\tilde{p}^s_k(v_k) = \tilde{p}^s_{k-1}(v_k) + \alpha \tilde{r}^s_{k-1}(v_k)$ \\
	Set $\tilde{p}^s_k(v) = \tilde{p}^s_{k-1}(v)$, $\tilde{r}^s_k(v) = \tilde{r}^s_{k-1}(v)$ $\forall\ v \in V \setminus ( \{ v_k \} \cup N_{\textrm{out}}(v_k) )$ \\
}
\end{algorithm}

\section{Practical version of \texttt{FW-BW-MCMC}} \label{appPracticalVersion}

In this appendix, we define and analyze a modified version of \texttt{FW-BW-MCMC} that is more useful in practice. Before proceeding to the formal definition and analysis, we first motivate the practical algorithm. First, suppose for an instance of \texttt{FW-BW-MCMC} we have already run the backward DP (Algorithm \ref{algApproxCont}) and we are currently running the forward DP (Algorithm \ref{algApproxPR}). Though \texttt{FW-BW-MCMC} dictates we run the forward DP until $\| r^s \|_1 < r^s_{\max}$ for some predefined $r^s_{\max}$, we could instead terminate the forward DP (even if $\| r^s \|_1 > r^s_{\max}$) and proceed to the random walks. In other words, we dynamically change $r^s_{\max}$ from the predefined value to the current value of $\| r^s \|_1$. Then, if the number of walks sampled is $w = c \| r^s \|_1 r^t_{\max} / \delta$, where
\begin{equation}\label{eqCrequire}
c = \tfrac{3 (2e)^{1/3} \log( 2 / p_{\textrm{fail}} )}{ \epsilon^{7/3}} ,
\end{equation}
the proof of Theorem \ref{STATE_FWBWMCMC_ACC} goes through, i.e.\ the accuracy guarantee is achieved. Furthermore, this argument holds at any iteration of the forward DP. In other words, we can terminate the forward DP at any iteration and achieve the accuracy guarantee, as long as we scale $w$ with the $\| r^s \|_1$ value obtained at termination. From this observation, we aim to terminate the forward DP at the ``optimal'' iteration, i.e.\ the iteration for which the overall complexity of the algorithm is minimized.

Towards determining this optimal iteration, let $C_{FDP}$ denote the complexity of the forward DP until the current iteration, and define $C_{MCMC} = \frac{3 (2e)^{1/3} r^t_{\max} \log( 2 / p_{\textrm{fail}} )}{\alpha \delta \epsilon^{7/3}}$, so that $\| r^s \|_1 C_{MCMC}$ gives the complexity of the MCMC stage (since $c \| r^s \|_1 r^t_{\max} / \delta$ walks are sampled, each in expected time $\frac{1}{\alpha}$, with $c$ satisfying \eqref{eqCrequire}). Then, if we terminate the forward DP at the current iteration, the combined complexity of forward DP and MCMC stages will be $C_{FDP} + \| r^s \|_1 C_{MCMC}$. Suppose instead that we decide to run one more iteration, i.e.\ to terminate the forward DP at the \textit{next} iteration. Then, by Algorithm \ref{algApproxPR}, the next iteration will have complexity $d_{\textrm{out}}(v^*)$. Furthermore, by \eqref{eqFDPl1Decrease} in Appendix \ref{PROOF_FWBWMCMC_COMP}, $\| r^s \|_1$ will decrease by $\alpha r^s(v^*)$ at the next iteration. Hence, if we run one more iteration, the combined complexity of forward DP and MCMC will be $\left( C_{FDP} + d_{\textrm{out}}(v^*) \right) + \left( \| r^s \|_1 - \alpha r^s(v^*) \right) C_{MCMC}$. Now clearly, we should terminate the forward DP if and only if the resulting complexity is less than the complexity resulting from running another iteration. Hence, from the previous argument, we should terminate if and only if
\begin{equation}\label{eqDlInfinityTerm}
C_{FDP} + \| r^s \|_1 C_{MCMC}  < \left( C_{FDP} + d_{\textrm{out}}(v^*) \right) + \left( \| r^s \|_1 - \alpha r^s(v^*) \right) C_{MCMC}  \Leftrightarrow \tfrac{r^s(v^*)}{d_{\textrm{out}}(v^*)} < \tfrac{1}{\alpha C_{MCMC}} .
\end{equation}

In other words, to optimize the tradeoff between forward DP and MCMC complexity, we should run the forward DP until $\| D^{-1} r^s \|_{\infty}$ falls below the threshold in \eqref{eqDlInfinityTerm}. This motivates the practical version of \texttt{FW-BW-MCMC}, given in Algorithm \ref{algOurEstimatorPractical}. Algorithm \ref{algOurEstimatorPractical} changes two aspects of \texttt{FW-BW-MCMC}. First, it replaces Algorithm \ref{algApproxPR} with Algorithm \ref{algApproxPROriginal} (which uses  $\| D^{-1} \tilde{r}^s \|_{\infty}$ termination, as suggested by \eqref{eqDlInfinityTerm}). Second, it scales the the number of random walks sampled with $\| \tilde{r}^s \|_1$, as discussed above.

\begin{algorithm}
\caption{$\hat{\pi}_s(t) = \texttt{FW-BW-MCMC-Practical} (G,s,t,\alpha,\tilde{r}^s_{\max}, r_{\max}^t,w)$} \label{algOurEstimatorPractical}
Let $(p^t,r^t) = \texttt{Approximate-Contributions} (G,t,\alpha,r_{\max}^t)$ (Algorithm \ref{algApproxCont}) \\
Let $(\tilde{p}^s,\tilde{r}^s) = \texttt{Approximate-PageRank-Original} (G,s,\alpha,\tilde{r}_{\max}^s)$ (Algorithm \ref{algApproxPROriginal}); set $\tilde{\sigma}_s = \tilde{r}^s / \| \tilde{r}^s \|_1$ \\
\For{ $i = 1$ \KwTo $w \| \tilde{r}^s \|_1$ }{
	 Sample random walk starting at $\nu \sim \tilde{\sigma}_s$ of length $\sim$ geom($\alpha$); let $X_i = r^t ( U_i )$, where $U_i$ is endpoint of walk
}
Let $\hat{\pi}_s(t) = p^t(s) + \langle \tilde{p}^s , r^t \rangle +  \frac{1}{w} \sum_{i=1}^{w \| \tilde{r}^s \|_1} X_i$
\end{algorithm}

We can now establish accuracy and average-case complexity guarantees for Algorithm \ref{algOurEstimatorPractical}.

\begin{theorem} \label{STATE_FWBWMCMC_PRAC_ACC}
Fix minimum PPR threshold $\delta$, relative error tolerance $\epsilon$, failure probability $p_{\textrm{fail}}$. Let
\begin{equation} \label{eqFwBwMcmcPracAccAss}
\epsilon \in \left( 0,\tfrac{1}{\sqrt{2e}} \right) , \quad w = \tfrac{c r^t_{\max}}{\delta} , \quad c > \tfrac{3 (2e)^{1/3} \log( 2 / p_{\textrm{fail}} )}{\epsilon^{7/3}}   .
\end{equation} 
Then the estimate $\hat{\pi}_s(t)$ produced by Algorithm \ref{algOurEstimatorPractical} satisfies \eqref{eqBidirTypeErrGuar}with probability $\geq 1-p_{\textrm{fail}}$.
\end{theorem}
\begin{proof}
As discussed above, the proof of Theorem \ref{STATE_FWBWMCMC_ACC} goes through to establish this result.
\end{proof}

\begin{theorem}
Fix minimum PPR threshold $\delta$, relative error tolerance $\epsilon$, and failure probability $p_{\textrm{fail}}$. Assume \eqref{eqFwBwMcmcPracAccAss} holds. Then for any $s \in V$ and for $t \sim V$ uniformly, setting $\tilde{r}^s_{\max} = \frac{ \delta \epsilon^{7/3}  }{  r^t_{\max} \log( 1 / p_{\textrm{fail}})  }$, $r^t_{\max} = \frac{ \sqrt{m \delta} \epsilon^{7/6} }{ \sqrt{n \log ( 1 / p_{\textrm{fail}} ) } }$ in Algorithm \ref{algOurEstimatorPractical} yields complexity $O \left(  \frac{\sqrt{m \log(1/p_{\textrm{fail}})}}{\sqrt{n \delta} \alpha \epsilon^{7/6} }  \right)$.
\end{theorem}
\begin{proof}
First, for the complexity of the backward DP (Algorithm \ref{algApproxCont}), we use the result from \cite{lofgren2013personalized}, which we include for completeness. Recall from Appendix \ref{PROOF_FWBWMCMC_COMP} that the complexity of Algorithm \ref{algApproxCont} for any $t \in V$ is bounded by $\tsum_{v \in V} d_{\textrm{in}}(v) \tfrac{\pi_v(t)}{\alpha r^t_{\max}}$. Hence, for $t \sim V$ uniformly, the expected complexity is
\begin{equation}
\tfrac{1}{n} \tsum_{t \in V} \tsum_{v \in V} d_{\textrm{in}}(v) \tfrac{\pi_v(t)}{\alpha r^t_{\max}} = \tfrac{1}{n \alpha r^t_{\max}} \tsum_{v \in V} d_{\textrm{in}}(v) \tsum_{t \in V} \pi_v(t) = \tfrac{m}{n \alpha r^t_{\max}} ,
\end{equation}
since $\sum_{t \in V} \pi_v(t) = 1$ by definition. Next, we consider the complexity of the forward DP (Algorithm \ref{algApproxPROriginal}). From Appendix \ref{PROOF_FWBWMCMC_COMP}, for any $s \in V$ we have complexity $\frac{1}{\alpha \tilde{r}^s_{\max}} = \frac{ r^t_{\max} \log( 1 / p_{\textrm{fail}}) }{ \alpha \delta \epsilon^{7/3} }$. Finally, for the MCMC stage, we sample $w \| \tilde{r}^s \|_1 \leq w$ walks, where $w = c r^t_{\max} / \delta$ with $c$ satisfying \eqref{eqFwBwMcmcPracAccAss}. Each walk is sampled in average time $\frac{1}{\alpha}$. Therefore, the MCMC stage complexity is $O ( \tfrac{ r^t_{\max} \log( 1 / p_{\textrm{fail}}) }{ \alpha \delta \epsilon^{7/3} } )$. Thus, the overall complexity of Algorithm \ref{algOurEstimatorPractical} is bounded by
\begin{equation}\label{eqPracCompRtmax}
O \left( \tfrac{m}{n \alpha r^t_{\max}} + \tfrac{ r^t_{\max} \log( 1 / p_{\textrm{fail}}) }{ \alpha \delta \epsilon^{7/3} } \right) .
\end{equation}
Substituting $r^t_{\max}$ given in the statement of the theorem yields the desired complexity bound. Further, viewing \eqref{eqPracCompRtmax} as a function of $r^t_{\max}$, it is straightforward to verify this $r^t_{\max}$ is the global minimizer.
\end{proof}

\section{Proof of Theorem \ref{STATE_UNIFIED_MCMC_COMP}} \label{PROOF_UNIFIED_MCMC_COMP}

We first observe
\begin{align}
& \P \left[  \tsum_{v \in V} \max_{s \in S} X_s^{(w)}(v) > (1+\epsilon) w \tsum_{v \in V} \max_{s \in S} \sigma_s(v) \right]  \\
& \leq \P \left[  \cup_{s \in S, v \in V} \{ X_s^{(w)}(v) > (1+\epsilon) w \sigma_s(v) \} \right] \leq \sum_{ s \in S, v \in V : \sigma_s(v) > 0  } \P \left[ X_s^{(w)}(v) > (1+\epsilon) w \sigma_s(v) \right] , \label{eqMaxSampleUnionBound}
\end{align}
where the second inequality uses the fact that $X^{(w)}_s(v) \sim \textrm{Binomial}(w,\sigma_s(v))$ (hence, $X^{(w)}_s(v) = 0$ when $\sigma_s(v) = 0$). Again using this fact, we have by \eqref{eqChernoff1} from Theorem \ref{thmChernoff} in Appendix \ref{PROOF_FWBWMCMC_ACC},
\begin{equation}\label{eqMaxSampleChernoff}
\P \left[ X_s^{(w)}(v) > (1+\epsilon) w \sigma_s(v) \right] \leq \exp \left( - \tfrac{\epsilon^2}{3} w \sigma_s(v) \right) .
\end{equation}
Combining \eqref{eqMaxSampleUnionBound} and \eqref{eqMaxSampleChernoff}, we obtain
\begin{align}
& \P \left[ \tsum_{v \in V} \max_{s \in S} X_s^{(w)}(v) > (1+\epsilon) w \tsum_{v \in V} \max_{s \in S} \sigma_s(v) \right]  \leq \sum_{ s \in S, v \in V : \sigma_s(v) > 0  } \exp \left( - \tfrac{\epsilon^2}{3} w \sigma_s(v) \right) \\
& \quad \leq  \left( \max_{s \in S, v \in V : \sigma_s(v) > 0} \left\{  \exp \left( - \tfrac{\epsilon^2}{3} w \sigma_s(v) \right) \right\} \right) \left( \tsum_{s \in S, v \in V} \mathbbm{1}_{ \{ \sigma_s(v) > 0 \} } \right) \\
& \quad  = \exp \left( \tfrac{- \epsilon^2}{3} w \min_{s \in S, v \in V : \sigma_s(v) > 0} \sigma_s(v) \right) \left( \tsum_{s \in S, v \in V} \mathbbm{1}_{ \{ \sigma_s(v) > 0 \} } \right) < p_{\fail} / 2 , \label{eqUnionBoundTerm1}
\end{align}
where the final inequality holds by the bound on $w$ in the statement of the theorem. For the lower tail, following the same steps used to obtain \eqref{eqUnionBoundTerm1} gives
\begin{equation}\label{eqUnionBoundTerm2}
\P \left[  \tsum_{v \in V} \max_{s \in S} X_s^{(w)}(v) < (1-\epsilon) w \tsum_{v \in V} \max_{s \in S} \sigma_s(v) \right] < p_{\fail} / 2 .
\end{equation}
Finally, by the union bound, \eqref{eqUnionBoundTerm1} and \eqref{eqUnionBoundTerm2} together establish the theorem.


\section{Proof of Proposition \ref{STATE_BW_ITER_MERGE}}  \label{PROOF_BW_ITER_MERGE}

First, assume \texttt{Merge} is used at each iteration for which $v^* = t_2$. By Algorithm \ref{algApproxCont}, $\| p^{t_2} \|_1$ increases by at least $\alpha r^t_{\max}$ at each iteration for which $v^* \neq t_1$. By \eqref{eqMergeUpdate}, $\| p^{t_2} \|_1$ increases by at least $\| p^{t_1} \|_1 r^t_{\max}$ at each iteration for which $v^* = t_1$. Let us define $I_1$ as the number of iterations for which $v^* \neq t_1$, $I_2$ as the number of iterations for which $v^* = t_1$, and $I = I_1 + I_2$ as the total number of iterations. Since $\| p^{t_2} \|_1 = 0$ at the start of Algorithm \ref{algApproxCont} and $\| p^{t_2} \|_1 \leq n \pi(t_2)$ by the invariant \eqref{eqBwInvariant}, we have
\begin{equation} \label{eqExtendMergeCompExpansion}
\tfrac{n \pi(t_2)}{r^t_{\max}} \geq \alpha I_1 + \| p^{t_1} \|_1 I_2 = \alpha I + ( \| p^{t_1} \|_1 - \alpha ) I_2 .
\end{equation}
Now at termination of Algorithm \ref{algApproxCont}, $\| r^{t_2} \|_{\infty} \leq r^t_{\max}$, so by the invariant \eqref{eqBwInvariant}, $\pi_{t_1}(t_2) \leq p^{t_2}(t_1) + r^t_{\max}$ at termination. Therefore, if $\pi_{t_1}(t_2) > r^t_{\max}$, $p^{t_2}(t_1) > 0$ at termination, which can only occur if $v^* = t_1$ at some iteration. Hence, $\pi_{t_1}(t_2) > r^t_{\max} \Rightarrow I_2 \geq 1$. Finally, from Algorithm \ref{algApproxCont}, $\| p^{t_1} \|_1 \geq \alpha$. Substituting into \eqref{eqExtendMergeCompExpansion} gives $I \leq \tfrac{n \pi(t_2)}{\alpha r^t_{\max}} - \tfrac{( \| p^{t_1} \|_1 - \alpha )}{\alpha}$.

If instead \texttt{Merge} is not used, $\| p^{t_2} \|_1$ increases by at least $\alpha r^t_{\max}$ at \textit{every} iteration. Hence, the same argument as above establishes that the total number of iterations is bounded by $\frac{n \pi(t_2)}{\alpha r^t_{\max}}$.

\section{Proof of Theorem \ref{STATE_MATRIX_ACC}} \label{PROOF_MATRIX_ACC}

We will use Corollary 6.2.1 from \cite{tropp2015introduction}, which (applied to our setting) states the following. Assume $\{ X_i \}_{i=1}^w$ are independent random matrices satisfying $\E [ X_i ] = R_S^{\trans} \Pi R_T$. Let $M$ be s.t.\ $\| X_i \|_2 \leq M$ a.s., and let $m_2(X_i) = \max \{ \| \E [ X_i X_i^{\trans} ] \|_2 , \| \E [ X_i^{\trans} X_i ] \|_2 \}$. Then $\forall\ \eta > 0$,
\begin{equation} 
\P \left[ \left\| R_S^{\trans} \Pi R_t - \tfrac{1}{w} \tsum_{i=1}^w X_i \right\|_2 > \eta \right] \leq 2l \exp \left(  \tfrac{-3 w \eta^2}{ 6 m_2(X_i) + 4 M \eta } \right) . \label{eqBernstein1}
\end{equation}
We have verified the independence and $\E [ X_i ] = R_S^{\trans} \Pi R_T$ assumptions in the main text. Furthermore, from \eqref{eqMatrixInvariant} and Algorithm \ref{algMatrixEstimator}, $\Pi(S,T) - \hat{\Pi}(S,T) = R_S^{\trans} \Pi R_t - \frac{1}{w} \sum_{i=1}^w X_i$. We may therefore write
\begin{align} 
& \P \left[ \left\| \Pi(S,T) - \hat{\Pi}(S,T) \right\|_2 > \epsilon \max \{ \| \Pi(S,T) \|_2 , 1 \} \right]  \leq 2l \exp \left(  \tfrac{-3 w \left( \epsilon \max \{ \| \Pi(S,T) \|_2 , 1 \} \right)^2 }{ 6 m_2(X_i) + 4 M \epsilon \max \{ \| \Pi(S,T) \|_2 , 1 \}  } \right)  \\
&  \quad \leq 2l \exp \left(  \frac{-3 w \epsilon^2  }{ 6 \frac{m_2(X_i)}{\max \{ \| \Pi(S,T) \|_2 , 1 \}} + 4 M \epsilon  } \right)   \leq 2l \exp \left(  \frac{-3 w \epsilon^2  }{ 6 \frac{m_2(X_i)}{ \| \Pi(S,T) \|_2 } + 4 M \epsilon  } \right) . \label{eqBernstein2}
\end{align}
where we have also used the inequalities $\max \{ \| \Pi(S,T) \|_2 , 1 \} \geq 1$, $\max \{ \| \Pi(S,T) \|_2 , 1 \} \geq \| \Pi(S,T) \|_2$.

Now to prove the theorem, we aim to find $M$ s.t.\ $\| X_i \|_2 \leq M$ a.s. and to compute $m_2(X_i)$ such that \eqref{eqBernstein2} is bounded by $p_{\fail}$, in each of the following cases:
\begin{align} 
& \textrm{(Case 1)} \quad \sigma = \sigma_{\avg}, \quad w \geq \tfrac{ l^2 \sqrt{\srank(\Pi(S,T))} \log ( 2 l / p_{\textrm{fail}} ) r^s_{\max} r^t_{\max} (6+4\epsilon) }{3\epsilon^2} . \label{eqWassumptionCase1} \\
& \textrm{(Case 2)} \quad \sigma = \sigma_{\max}, \quad w \geq \tfrac{ l^{3/2} \| \Sigma \|_{\infty,1} \log ( 2 l / p_{\textrm{fail}} ) r^s_{\max} r^t_{\max} (6+4\epsilon)}{3\epsilon^2} . \label{eqWassumptionCase2}
\end{align}

We begin with Case 1. By Lemma \ref{lemL2bound}, we may take $M = l^{3/2} r^s_{\max} r^t_{\max}$, and by Lemma \ref{lem2ndMoment}, we have $m_2(X_i) \leq l^2 r^s_{\max} r^t_{\max} \| \Pi(S,T) \|_F$. We can then write
\begin{align}
6 \tfrac{m_2(X_i)}{ \| \Pi(S,T) \|_2 } + 4 M \epsilon & \leq l^2 r^s_{\max} r^t_{\max} \left( 6 \tfrac{\| \Pi(S,T) \|_F}{ \| \Pi(S,T) \|_2} +  \tfrac{4}{\sqrt{l}} \epsilon \right) = l^2 r^s_{\max} r^t_{\max} \left( 6 \sqrt{\srank(\Pi(S,T))} +  \tfrac{4}{\sqrt{l}} \epsilon \right) \\
&   \leq l^2 r^s_{\max} r^t_{\max} \sqrt{\srank(\Pi(S,T))} ( 6 + 4 \epsilon ) \leq \tfrac{3 w \epsilon^2}{\log(2l/p_{\fail})} , \label{eqDenBoundCase1_4}
\end{align}
where the equality is definition of $\srank$, the penultimate inequality holds since $l , \srank(\Pi(S,T)) \geq 1$, and the final inequality is by \eqref{eqWassumptionCase1}. Substituting \eqref{eqDenBoundCase1_4} into \eqref{eqBernstein2} establishes the desired result.

For Case 2, we take $M = l^{3/2} \| \Sigma \|_{\infty,1} r^s_{\max} r^t_{\max}$ (Lemma \ref{lemL2bound}), and by Lemma \ref{lem2ndMoment} we have 
\begin{equation}
m_2(X_i) \leq l \| \Sigma \|_{\infty,1} r^s_{\max} r^t_{\max} \max \{ \| \Pi(S,T) \|_{\infty} , \| \Pi(S,T) \|_{1} \} .
\end{equation}
We then obtain
\begin{align}
6 \tfrac{m_2(X_i)}{ \| \Pi(S,T) \|_2 } + 4 M \epsilon & \leq l \| \Sigma \|_{\infty,1} r^s_{\max} r^t_{\max} \left( 6 \tfrac{\max \{ \| \Pi(S,T) \|_{\infty} , \| \Pi(S,T) \|_{1} \}}{ \| \Pi(S,T) \|_2} +  4 \sqrt{l} \epsilon \right) \nonumber \\
&  \leq l^{3/2} \| \Sigma \|_{\infty,1}  r^s_{\max} r^t_{\max} \left( 6 +  4  \epsilon \right)  \leq \tfrac{3 w \epsilon^2}{\log(2l/p_{\fail})} \label{eqDenBoundCase2_2}
\end{align}
where the second inequality is a standard norm equivalence inequality (for $A \in \R^{l \times l}$, $\| A \|_{\infty}, \| A \|_1 \leq \sqrt{l} \| A \|_2$), and the third inequality is by \eqref{eqWassumptionCase2}. Substituting \eqref{eqDenBoundCase2_2} into \eqref{eqBernstein2} completes the proof.

\begin{lemma}  \label{lemL2bound}
If $\sigma = \sigma_{\avg}$, $\| X_i \|_2 \leq l^{3/2} r^s_{\max} r^t_{\max}$ a.s.; if $\sigma = \sigma_{\max}$, $\| X_i \|_2 \leq l^{3/2} \| \Sigma \|_{\infty,1} r^s_{\max} r^t_{\max}$ a.s.
\end{lemma}
\begin{proof}
Observe $X_i = a_i b_i^{\trans}$, where $a_i,b_i \in \R^l$ with $a_i(j) = r^{s_j}(\mu_i) / \sigma(\mu_i), b_i(j) = r^{t_j}(\nu_i)$. $X_i$ has rank 1, and we may write its singular value decomposition as
\begin{equation}
X_i = \left( \| a_i \|_2 \| b_i \|_2 \right) \left( \tfrac{a_i}{\| a_i \|_2} \right) \left( \tfrac{b_i}{\| b_i \|_2} \right)^{\trans} ,
\end{equation}
so the nonzero singular value of $X_i$ is $\| a_i \|_2 \| b_i \|_2$. Using the well-known fact that a matrix's 2-norm equals its largest singular value, $\| X_i \|_2 = \| a_i \|_2 \| b_i \|_2$, so we seek bounds on $\| a_i \|_2$ and $\| b_i \|_2$. 

First, we assume $\sigma = \sigma_{\avg}$. Then we can write the following:
\begin{equation} \label{eqAsBound1}
\sigma(\mu_i)  = \tfrac{1}{l} \tsum_{s \in S} \tfrac{r^s(\mu_i)}{\|r^s\|_1} \geq \tfrac{1}{l r^s_{\max}} \tsum_{s \in S} r^s(\mu_i)   \geq \tfrac{1}{l r^s_{\max}} \left( \tsum_{s \in S} r^s(\mu_i)^2 \right)^{1/2} = \tfrac{1}{l r^s_{\max}} \| a_i \|_2 \sigma(\mu_i) .
\end{equation}
Here the first equality holds by definition \eqref{eqSigmaAvgSigmaMaxDefn}, the first inequality uses the terminating condition of Algorithm \ref{algApproxPR} ($\| r^s \|_1 \leq r^s_{\max}$), the second inequality is by nonnegativity, and the second equality is by definition of $a_i$. We conclude $\| a_i \|_2 \leq l r^s_{\max}$. To bound $\| b_i \|_2$, we have
\begin{equation} \label{eqAsBoundBi}
\| b_i \|_2 \leq \sqrt{l} \| b_i \|_{\infty} \leq \sqrt{l} r^t_{\max} ,
\end{equation}
where we have used a well-known vector norm inequality and the terminating condition of Algorithm \ref{algApproxCont} ($\| r^t \|_{\infty} \leq r^t_{\max}$). Hence, $\| X_i \|_2 \leq l^{3/2} r^s_{\max} r^t_{\max}$ follows.

Next, we assume $\sigma = \sigma_{\max}$. We have
\begin{align}
\sigma(\mu_i) & = \tfrac{1}{\| \Sigma \|_{\infty,1}} \max_{s \in S} \tfrac{r^s(\mu_i)}{\|r^s\|_1} \geq \tfrac{1}{\| \Sigma \|_{\infty,1} r^s_{\max}} \max_{s \in S} r^s(\mu_i) \geq \tfrac{1}{l \| \Sigma \|_{\infty,1} r^s_{\max}} \tsum_{s \in S} r^s(\mu_i)  = \tfrac{1}{l \| \Sigma \|_{\infty,1} r^s_{\max}} \| a_i \| \sigma(\mu_i) , \label{asBound2}
\end{align}
which is justified similarly to \eqref{eqAsBound1}. Combining with \eqref{eqAsBoundBi} gives $\| X_i \|_2 \leq l^{3/2} \| \Sigma \|_{\infty,1} r^s_{\max} r^t_{\max}$.
\end{proof}

\begin{lemma} \label{lem2ndMoment}
If $\sigma = \sigma_{\avg}$, then $m_2(X_i) \leq l^2 r^s_{\max} r^t_{\max} \| \Pi(S,T) \|_F$; if instead $\sigma = \sigma_{\max}$, then $m_2(X_i) \leq l \| \Sigma \|_{\infty,1} r^s_{\max} r^t_{\max} \max \{ \| \Pi(S,T) \|_{\infty} , \| \Pi(S,T) \|_{1} \}$.
\end{lemma}

\begin{proof}
We first assume $\sigma = \sigma_{\avg}$. Using Jensen's inequality, and since $X_i = a_i b_i^{\trans}$, we have $\| \E [ X_i X_i^{\trans} ] \|_2 \leq \E [ \| X_i X_i^{\trans} \|_2 ] = \E [ \| a_i \|_2^2 \| b_i \|_2^2 ]$; similarly, $\| \E [ X_i X_i^{\trans} ] \|_2 \leq \E [ \| a_i \|_2^2 \| b_i \|_2^2 ]$. Thus, 
\begin{align}
m_2(X_i) & \leq \E [ \| a_i \|_2^2 \| b_i \|_2^2 ]  = \tsum_{u , v \in V} \sigma(u) \pi_u(v) \left( \tfrac{1}{\sigma(u)^2} \tsum_{s \in S} r^s(u)^2 \right) \left( \tsum_{t \in T} r^t(v)^2 \right)  \\
& \leq r^t_{\max} \tsum_{u , v \in V} \tfrac{\pi_u(v)}{\sigma(u)} \left( \tsum_{s \in S} r^s(u) \right)^2 \tsum_{t \in T} r^t(v) \\
& \leq l r^s_{\max} r^t_{\max} \tsum_{s \in S} \tsum_{t \in T} \tsum_{u, v \in V} r^s(u) \pi_u(v) r^t(v) \leq l r^s_{\max} r^t_{\max} \tsum_{s \in S} \sum_{t \in T} \pi_s(t) ,  \label{eqBernMomIneq3}
\end{align}
where the second inequality uses the terminating condition of Algorithm \ref{algApproxCont} ($r^t(v) \leq r^t_{\max}$) and the nonnegativity of $r^s(u)$, the third follows from \eqref{eqAsBound1}, and the fourth uses the invariant \eqref{eqFwBwInvariant}. Finally, letting $vec(\Pi(S,T))$ denote the $l^2$-length vector with entries $\{ \pi_s(t) \}_{s \in S, t \in T}$, we have
\begin{equation*}
\tsum_{s \in S} \tsum_{t \in T} \pi_s(t) = \| vec(\Pi(S,T)) \|_1 \leq  l \| vec(\Pi(S,T)) \|_2 = l \| \Pi(S,T) \|_F ,
\end{equation*}
where the first equality is by nonnegativity, the inequality is a standard norm inequality, and the second inequality is by definition of Frobenius norm. Substituting into \eqref{eqBernMomIneq3} establishes the result.

We next assume $\sigma = \sigma_{\max}$ and bound $\| \E [ X_i X_i^{\trans} ] \|_2$. We observe that by definition,
\begin{align}
& X_i X_i^{\trans} = \tfrac{\left( \sum_{t \in T} r^t(\nu_i)^2 \right)}{\sigma(\mu_i)^2}  \begin{bmatrix} r^{s_1}(\mu_i) & \cdots & r^{s_l}(\mu_i)  \end{bmatrix}^{\trans} \begin{bmatrix} r^{s_1}(\mu_i)&& \cdots & r^{s_l}(\mu_i) \end{bmatrix} \\
& \Rightarrow \E [ X_i X_i^{\trans} ] = \tsum_{u , v \in V} \tfrac{\pi_u(v) }{\sigma(u)} \tsum_{t \in T} r^t(v)^2  \begin{bmatrix} r^{s_1}(u) & \cdots & r^{s_l}(u)  \end{bmatrix}^{\trans} \begin{bmatrix} r^{s_1}(u) & \cdots & r^{s_l}(u)  \end{bmatrix} .
\end{align}
Letting $1_l$ denote the all ones vector of length $l$, we also have 
\begin{equation} \label{eqXXTrowSums}
\E [ X_i X_i^{\trans} ] 1_l = \tsum_{u , v \in V}  \tfrac{\pi_u(v)}{\sigma(u)}  \tsum_{t \in T} r^t(v)^2   \tsum_{s \in S} r^s(u)  \begin{bmatrix} r^{s_1}(u) & \cdots & r^{s_l}(u) \end{bmatrix}^{\trans} .
\end{equation}
Now since $\E [ X_i X_i^{\trans}]$ is symmetric, its 2-norm is its largest eigenvalue; since it is nonnegative, the Perron-Frobenius Theorem states this eigenvalue is bounded by its maximum row sum. Therefore,
\begin{align}
\| \E [ X_i X_i^{\trans} ] \|_2 & \leq  \max_{j \in \{ 1 , 2 , ... , l \} } \tsum_{u , v \in V}  \tfrac{\pi_u(v)}{\sigma(u)}   \tsum_{t \in T} r^t(v)^2 \tsum_{s \in S} r^s(u) r^{s_j}(u) \label{eqXXTbound1}  \\
&  \leq l \| \Sigma \|_{\infty,1} r^s_{\max} r^t_{\max} \max_{j \in \{ 1 , 2 , ... , l \} } \tsum_{t \in T} \tsum_{u , v \in V} r^{s_j}(u) \pi_u(v) r^t(v) \label{eqXXTbound2}  \\
&  \leq l \| \Sigma \|_{\infty,1} r^s_{\max} r^t_{\max} \max_{j \in \{ 1 , 2 , ... , l \} } \tsum_{t \in T} \pi_{s_j} ( t )  = l \| \Sigma \|_{\infty,1} r^s_{\max} r^t_{\max} \| \Pi(S,T) \|_{\infty} , \label{eqXXTbound4}
\end{align}
where \eqref{eqXXTbound1} uses the row sums derived in \eqref{eqXXTrowSums}, \eqref{eqXXTbound2} uses  \eqref{asBound2} from the proof of Lemma \ref{lemL2bound} and the terminating condition of Algorithm \ref{algApproxCont} ($\| r^t \|_{\infty} \leq r^t_{\max}$), and \eqref{eqXXTbound4} uses the invariant \eqref{eqFwBwInvariant}. We can use the same idea to bound $\| \E [ X_i^{\trans} X_i ] \|_2$. The steps to obtain the expression analogous to \eqref{eqXXTbound1} follow the same approach so we omit them. We then have
\begin{align}
\| \E [ X_i^{\trans} X_i ] \|_2 & \leq  \max_{j \in \{ 1 , 2 , ... , l \} } \tsum_{u , v \in V}  \frac{\pi_u(v)}{\sigma(u)}   \tsum_{s \in S} r^s(u)^2 \tsum_{t \in T} r^t(v) r^{t_j}(v) \nonumber \\ 
&  \leq \tsum_{u , v \in V}  \tfrac{\pi_u(v)}{\sigma(u)} \tsum_{s \in S} r^s(u) \max_{s' \in S} r^{s'}(u) \tsum_{t \in T} r^t(v) r^{t_j}(v) \label{eqXTXbound1} \\
&  \leq l \| \Sigma \|_{\infty,1} r^s_{\max} r^t_{\max} \max_{j \in \{ 1 , 2 , ... , l \} } \tsum_{s \in S} \tsum_{u , v \in V} r^s(u) \pi_u(v) r^{t_j}(v) \label{eqXTXbound2} \\
&  \leq l \| \Sigma \|_{\infty,1} r^s_{\max} r^t_{\max} \max_{j \in \{ 1 , 2 , ... , l \} } \tsum_{s \in S} \pi_s ( t_j )  = l \| \Sigma \|_{\infty,1} r^s_{\max} r^t_{\max} \| \Pi(S,T) \|_{1} , \label{eqXTXbound4} 
\end{align}
where \eqref{eqXTXbound1} is immediate, \eqref{eqXTXbound2} uses \eqref{asBound2} from the proof of Lemma \ref{lemL2bound} and the terminating condition of Algorithm \ref{algApproxCont} ($\| r^t \|_{\infty} \leq r^t_{\max}$), and \eqref{eqXTXbound4} uses the invariant \eqref{eqFwBwInvariant}. We conclude from \eqref{eqXXTbound4} and \eqref{eqXTXbound4} that 
\begin{align}
\max \{ \| \E [ X_i^{\trans} X_i ] \|_2 , \| \E [ X_i X_i^{\trans} ] \|_2 \}  \leq l \| \Sigma \|_{\infty,1} r^s_{\max} r^t_{\max} \max \{ \| \Pi(S,T) \|_{\infty} , \| \Pi(S,T) \|_{1} \} . \quad\quad \qedhere
\end{align}
\end{proof}

\iftoggle{arxiv}{ 

\section{Proof of Theorem \ref{thmSbmSources}} \label{appProofThmSbmSources}

The theorem relies on two key lemmas. The first of these (Lemma \ref{lemSbmDegConc}) shows that the out-degrees in our stochastic block model concentrate, in the sense that these degrees are all close to $p \sqrt{n}$ with high probability. Lemma \ref{lemSbmDegConc} also bounds the maximal number of outgoing edges pointing to other communities, i.e.\ $\max_{v \in V_n}  d_{\outT}^-(v)$. The proof, deferred to Appendix \ref{appProofLemSbmDegConc}, is a modified version of a similar (standard) result for similar random graph families (such as the Erd\H{o}s-R\'enyi model). 

\begin{lemma} \label{lemSbmDegConc}
Let $\{ G_n = (V_n,E_n) \}_{n \in \N : \sqrt{n} \in \N}$ be the sequence of stochastic block models defined in Section \ref{secManySource}, with $p_n = p$ for some constant $p \in (0,1)$. For $\epsilon, C > 0$, define the following events:
\begin{gather}
\mathcal{E}_{n,\epsilon} = \cap_{v \in V_n} \left\{ d_{\outT}(v) \in \left( (1-\epsilon) p \sqrt{n} , (1+\epsilon) p \sqrt{n} \right) \right\} , \\
\mathcal{F}_{n,C} = \left\{ \max_{v \in V_n}  d_{\outT}^-(v) \leq C q_n n \right\} , \quad \mathcal{G}_{n,C} = \left\{ \max_{v \in V_n}  d_{\outT}^-(v) \leq \frac{C \log n}{\log \log n} \right\} .
\end{gather}
Then the following hold:
\begin{itemize}
\item If $q_n = o ( 1 / \sqrt{n} )$, then for any constant $\epsilon > 0$, $\lim_{n \rightarrow \infty} \P ( \mathcal{E}_{n,\epsilon} ) = 1$.
\item If $q_n = \Omega ( \log n / n )$, then for some constant $C > 0$, $\lim_{n \rightarrow \infty} \P ( \mathcal{F}_{n,C} ) = 1$.
\item If $q_n = \Theta ( 1/n )$, then for some constant $C > 0$, $\lim_{n \rightarrow \infty} \P ( \mathcal{G}_{n,C} ) = 1$.
\end{itemize}
\end{lemma}
\begin{proof}
See Appendix \ref{appProofLemSbmDegConc}.
\end{proof}

The second key lemma (Lemma \ref{lemSigVecBounds}) contains useful bounds regarding the vector $\sigma^s_k = r^s_k / \| r^s_k \|_1$, where $r^s_k$ is the $r^s$ vector in the $k$-th iteration of Algorithm \ref{algApproxPR}. (Here and moving forward, we explicitly denote the current iteration of Algorithm \ref{algApproxPR} via subscripts, as in Algorithm \ref{algApproxPROriginal} from Appendix \ref{PROOF_FWBWMCMC_COMP}). In fact, these bounds hold more generally than will be required for the theorem; namely, we formulate the lemma for any deterministic graph on $n$ nodes for which the out-degree condition $\mathcal{E}_{n,\epsilon}$ holds. The proof is somewhat tedious so is deferred to Appendix \ref{appProofLemSigVecBounds}.
\begin{lemma} \label{lemSigVecBounds}
Let $G_n = (V_n = \{1,\ldots,n\} ,E_n)$ be a deterministic graph satisfying
\begin{equation}  \label{eqKeyAssumption2}
d_{\outT}(v) \in \left( (1-\epsilon) p \sqrt{n} , (1+\epsilon) p \sqrt{n} \right)\ \forall\ v \in V_n 
\end{equation}
for some $p, \epsilon \in (0,1)$, and let $k \in \{1,\ldots,\lceil (1-\epsilon)^2 p \sqrt{n}  (1-\alpha)/(2e) \rceil \}$. Then for any $s \in V_n$,
\begin{equation} \label{eqSigVecBoundsInComm}
\sigma_k^s(v) < \frac{1}{\sqrt{n}}  \frac{ e }{  (1-\alpha)(1-\epsilon)^2 p  - 2 e k / \sqrt{n}   }\ \forall\ v \in V_n ,
\end{equation}
and for any $S_n \subset V_n$ s.t.\ $s \in S_n$,
\begin{equation} \label{eqSigVecBoundsOutComm}
\sum_{v \in V_n \setminus S_n} \sigma_k^s ( v ) < \frac{ \max_{s' \in S_n} | N_{\outT}(s') \setminus S_n | }{ \sqrt{n} }  \frac{  1+ 2 e k / \sqrt{n}  }{ (1-\epsilon)p \left(  (1-\alpha) (1-\epsilon)^2 p  - 2 e k / \sqrt{n} \right)  } .
\end{equation}
\end{lemma}
\begin{proof}
See Appendix \ref{appProofLemSigVecBounds}.
\end{proof}

We now turn to the proof of the theorem. First, suppose all sources belong to the same community, and consider the sub-case $q_n = o ( 1 / \sqrt{n} ), q_n = \Omega ( \log n / n )$. Then for any $\epsilon \in (0,1)$, Lemma \ref{lemSigVecBounds} implies that any realization of $G_n$ satisfying $\mathcal{E}_{n,\epsilon}$ also satisfies
\begin{align}
\sum_{v \in V_n} \max_{s \in S_n} \sigma_k^s(v) & \leq \sum_{v \in S_n } \max_{s \in S_n} \sigma_k^s(v) + \sum_{s \in S_n} \sum_{v \in V_n \setminus S_n }  \sigma_k^s(v) \\
& \leq |S_n| \times \frac{1}{\sqrt{n}} \frac{ e }{  (1-\alpha)(1-\epsilon)^2 p  - 2 e k / \sqrt{n} } \\
& \quad\quad +  |S_n| \times  \frac{ \max_{s \in S_n} d_{\outT}^-(s)}{ \sqrt{n} } \frac{  1+ 2 e k / \sqrt{n}  }{ (1-\epsilon)p \left(  (1-\alpha) (1-\epsilon)^2 p  - 2 e k / \sqrt{n} \right)  } .
\end{align}
Recall $\alpha, \epsilon, p$ are constants and $|S_n| = \sqrt{n} , k = o ( \sqrt{n} )$ in the statement of the theorem. Hence, for some $C'' > 0$ and all $n$ large, 
any realization  of $G_n$ satisfying $\mathcal{E}_{n,\epsilon}$ also satisfies
\begin{equation}
\sum_{v \in V_n} \max_{s \in S_n} \sigma_k^s(v) \leq C'' \max_{s \in S_n} d_{\outT}^-(s) .
\end{equation}
Now let $C' > 0, C = C' C''$. Then for $n$ large, any realization satisfying $\mathcal{E}_{n,\epsilon}$ and $\mathcal{F}_{n,C'}$ also satisfies
\begin{equation}
\sum_{v \in V_n} \max_{s \in S_n} \sigma_k^s(v) \leq C q_n n .
\end{equation}
In other words, we have shown that for some $C > 0$ and any $C' > 0$, 
\begin{equation}
\lim_{n \rightarrow \infty} \P \left( \sum_{v \in V_n} \max_{s \in S_n} \sigma_k^s(v) \leq C q_n n \middle| \mathcal{E}_{n,\epsilon} , \mathcal{F}_{n,C'} \right) = 1 . 
\end{equation}
Finally, for $C'$ satisfying the second statement of Lemma \ref{lemSbmDegConc}, we obtain
\begin{align}
\P \left( \sum_{v \in V_n} \max_{s \in S_n} \sigma_k^s(v) \leq C q_n n \right)  \geq  \P \left( \sum_{v \in V_n} \max_{s \in S_n} \sigma_k^s(v) \leq C q_n n \middle| \mathcal{E}_{n,\epsilon} , \mathcal{F}_{n,C'} \right) \P \left( \mathcal{E}_{n,\epsilon} , \mathcal{F}_{n,C'} \right) \xrightarrow[n \rightarrow \infty]{} 1 .
\end{align}
In the sub-case $q_n = \Theta ( 1 / n )$, a similar argument implies that for some $C, C' > 0$,
\begin{align}
& \P \left( \sum_{v \in V_n} \max_{s \in S_n} \sigma_k^s(v) \leq \frac{C \log n}{ \log \log n } \right) \\
& \quad\quad \geq \P \left( \sum_{v \in V_n} \max_{s \in S_n} \sigma_k^s(v) \leq \frac{C \log n}{ \log \log n } \middle| \mathcal{E}_{n,\epsilon} , \mathcal{G}_{n,C'} \right) \P \left( \mathcal{E}_{n,\epsilon} , \mathcal{G}_{n,C'} \right) \xrightarrow[n \rightarrow \infty]{} 1 .
\end{align}

We next consider the case for which all sources belong to different communities, i.e.\ $S_n = \{ \sqrt{n} , 2 \sqrt{n} , \ldots , n \}$ (which is without loss of generality by symmetry). Then clearly
\begin{equation} \label{eqDistinctCommInitialBound}
\sum_{v \in V_n} \max_{s \in S_n} \sigma_k^s(v) \geq \sum_{i=1}^{\sqrt{n}} \sum_{v=1+(i-1)\sqrt{n}}^{i\sqrt{n}} \sigma_k^{i \sqrt{n}}(v) .
\end{equation}
Furthermore, for any $\epsilon \in (0,1)$, Lemma \ref{lemSigVecBounds} implies that any realization satisfying $\mathcal{E}_{n,\epsilon}$ satisfies
\begin{equation}
\sum_{v=1}^{\sqrt{n}} \sigma_k^{\sqrt{n}}(v) \geq 1 - \frac{ \max_{v \in V_n } d_{\outT}^-(v)}{ \sqrt{n} } \frac{  1+ 2 e k / \sqrt{n}  }{ (1-\epsilon)p \left(  (1-\alpha) (1-\epsilon)^2 p  - 2 e k / \sqrt{n} \right)  } .
\end{equation}
Now suppose $q_n = o ( 1 / \sqrt{n} ), q_n = \Omega ( \log n / n )$, and let $\delta \in (0,1)$ be a constant. Then for $C> 0$ and $n$ sufficiently large, any realization satisfying $\mathcal{E}_{n,\epsilon}$ and $\mathcal{F}_{n,C}$ will also satisfy
\begin{equation}
\sum_{v=1}^{\sqrt{n}} \sigma_k^{\sqrt{n}}(v) \geq 1 - \frac{C q_n n}{\sqrt{n}} = 1 - C q_n \sqrt{n} \geq 1 - \delta .
\end{equation}
where we again used the fact that $\alpha,\epsilon,p$ are constant and $k = o(\sqrt{n})$. The same argument holds for all summands in the summation over $i$ in \eqref{eqDistinctCommInitialBound}. It follows that, for appropriate choice of $C > 0$,
\begin{equation}
\lim_{n \rightarrow \infty} \P \left( \sum_{v \in V_n} \max_{s \in S_n} \sigma_k^s(v) \geq (1-\delta) \sqrt{n} \middle| \mathcal{E}_{n,\epsilon} , \mathcal{F}_{n,C} \right) \P \left( \mathcal{E}_{n,\epsilon} , \mathcal{F}_{n,C} \right) = 1 .
\end{equation}
A similar approach establishes the desired result in the case $q_n = \Theta ( 1 / n )$.

Note that (perhaps surprisingly) the only feature of the stochastic block model used above was the degree concentration of Lemma \ref{lemSbmDegConc}. In other words, we considered the number of edges for each node, while ignoring how exactly these edges were connected. Consequently, the same analysis can be used to obtain results for sequences of \textit{deterministic} graphs $\{ G_n = (V_n,E_n) \}_{n \in \N : \sqrt{n} \in \N}$. For example, if such a sequence satisfies $\mathcal{E}_{n,\epsilon}, \mathcal{G}_{n,C}$ for some constants $\epsilon, C$ and for all $n$ large, the analysis above implies $\| \Sigma_{S_n} \|_{\infty,1} = O ( \log n / \log \log n )$ when $\sqrt{n}$ sources belong to the same community, whereas $\| \Sigma_{S_n} \|_{\infty,1} = \Omega(\sqrt{n})$ when $\sqrt{n}$ sources belong to different communities.

\subsection{Proof of Lemma \ref{lemSbmDegConc}} \label{appProofLemSbmDegConc}

For the first statement, we begin by showing $d_{\outT}(1)$ concentrates around $p \sqrt{n}$; we will then use the union bound to establish the lemma. Towards this end, first note that since edges from node $1$ to each $v \in \{2,\ldots,\sqrt{n}\}$ are present with probability $p$, and since edges from node $1$ to each $v \in \{\sqrt{n}+1,\ldots,n\}$ are present with probability $q_n$, we have
\begin{equation} \label{eqAsympRegExp}
\E \left[ d_{\outT}(1) \right] = p \left( \sqrt{n} - 1 \right) + q_n \left( n - \sqrt{n} \right) = p \sqrt{n} + \left( q_n \left( n - \sqrt{n} \right) - p \right) .
\end{equation}
Next, since $q_n = o ( 1/ \sqrt{n}  )$ and $p$ is constant by assumption, we have for $n$ sufficiently large,
\begin{equation}
\frac{q_n (n - \sqrt{n} ) - p }{ p \sqrt{n} } \leq \frac{ \epsilon/2 }{ 1 + \epsilon/2 } \quad \Rightarrow \quad \left( 1 + \frac{\epsilon}{2} \right) \left( q_n (n - \sqrt{n} ) - p \right)  \leq \frac{\epsilon  }{2}  p \sqrt{n} .
\end{equation}
Thus, combining the previous two lines, we obtain (for such $n$),
\begin{equation}
\left( 1 + \frac{\epsilon}{2} \right) \E \left[ d_{\outT}(1) \right] = \left( 1 + \frac{\epsilon}{2} \right) p \sqrt{n} + \left( 1 + \frac{\epsilon}{2} \right) \left( q_n \left( n - \sqrt{n} \right) - p \right) \leq (1+\epsilon) p \sqrt{n} .
\end{equation}
We can then use monotonicity and \eqref{eqChernoff1} from Appendix \ref{PROOF_FWBWMCMC_ACC} to obtain
\begin{equation}
\P \left( d_{\outT}(1) > (1+\epsilon) p \sqrt{n} \right) \leq \P \left( d_{\outT}(1) > \left( 1 + \frac{\epsilon}{2} \right) \E \left[ d_{\outT}(1) \right] \right)  \leq \exp \left( - \frac{ \epsilon^2 p }{12} \sqrt{n} \right) ,
\end{equation}
where we also used $\E [ d_{\outT}(1) ] \geq p \sqrt{n}$ by \eqref{eqAsympRegExp}. Using the same argument for the lower tail, and then using the union bound, we thus obtain
\begin{equation}
\P \left(  d_{\outT}(1) \notin \left[ (1-\epsilon) p \sqrt{n} , (1+\epsilon) p \sqrt{n} \right] \right) \leq 2 \exp \left( - \frac{ \epsilon^2 p }{12} \sqrt{n} \right) .
\end{equation}
Finally, by this bound, the fact that $\{ d_{\outT}(v) \}_{v \in V}$ are identically-distributed, and the union bound,
\begin{equation}
\P \left( \cup_{v \in V} \left\{ d_{\outT}(v) \notin \left[ (1-\epsilon) p \sqrt{n} , (1+\epsilon) p \sqrt{n} \right] \right\} \right) \leq 2 n \exp \left( - \frac{ \epsilon^2 p }{12} \sqrt{n} \right) \xrightarrow[n \rightarrow \infty]{} 0 ,
\end{equation}
which, by the law of complements, completes the proof of the first statement.

For the second statement, we similarly begin with a tail bound for $d_{\outT}^-(1)$. First note that, since $q_n = \Omega ( \log n / n )$, we can find $C' > 0$ such that for all $n$ sufficiently large,
\begin{equation}
\frac{ q_n }{ \log n / n } > C' \Rightarrow q_n n > C' \log n .
\end{equation}
Now let $C > \max \{ 2 e , 2 / ( C'  \log 2 ) \}$. Then clearly
\begin{equation}
C q_n n > 2 e q_n \left( n - \sqrt{n} \right) = 2 e \E \left[ d_{\outT}^-(1) \right] .
\end{equation}
Hence, we can use \eqref{eqChernoff2} from Appendix \ref{PROOF_FWBWMCMC_ACC} to obtain
\begin{equation}
\P \left( d_{\outT}^-(1) > C q_n n \right) \leq 2^{ - C q_n n } .
\end{equation}
By the union bound argument used above, we then have
\begin{equation}
\P \left( \max_{v \in V_n} d_{\outT}^-(v) > C q_n n \right) \leq n \P \left( d_{\outT}^-(1) > C q_n n \right) = 2^{ - ( C q_n n - \log_2 n ) } .
\end{equation}
Also, by our choice of $C$ and for $n$ sufficiently large (so that $q_n n > C' \log n$),
\begin{equation}
C q_n n - \log_2 n  > \frac{2}{C' \log 2} C' \log n - \log_2 n = \log_2 n .
\end{equation}
Combining the previous two inequalities then yields, for $n$ sufficiently large,
\begin{equation}
\P \left( \max_{v \in V_n} d_{\outT}^-(v) > C q_n n \right) \leq 1 / n ,
\end{equation}
from which the second statement clearly follows.

For the third statement, we again derive a tail bound for $d_{\outT}^-(1)$ and invoke the union bound, but the tail bound requires a slightly different approach. First, for any $M \in \{1,\ldots, \lfloor n - \sqrt{n} \rfloor \}$, the event $\{ d_{\outT}^-(1) \geq M \}$ means that node $1$ has outgoing edges to $M$ nodes in other communities, so
\begin{align}
\P \left( d_{\outT}^-(1) \geq M \right) & = \P \left( \cup_{U_n \subset \{ 1 + \sqrt{n} , \ldots , n \} : |U_n| = M} \{ 1 \rightarrow u \in E_n\ \forall\ u \in U_n \} \right) \\
& \leq \sum_{ U_n \subset \{ 1 + \sqrt{n} , \ldots , n \} : |U_n| = M } \P \left( 1 \rightarrow u \in E_n\ \forall\ u \in U_n \right) = { n - \sqrt{n} \choose M } q_n^M \leq { n \choose M } q_n^M  ,
\end{align}
where the first inequality is the union bound, the second equality holds by definition of our stochastic block model, and the second inequality is immediate. Now by assumption $q_n = \Theta ( 1 / n )$, we can find $C_1$ such that $q_n n \leq C_1$ for $n$ sufficiently large; combined with the standard binomial coefficient approximation ${ n \choose M } \leq ( \frac{n e}{M} )^M$, we can further bound the above as
\begin{equation}
\P \left( d_{\outT}^-(1) \geq M \right) \leq \left( \frac{n q_n e}{M} \right)^M \leq \left( \frac{C_2}{M} \right)^M 
\end{equation}
for all $n$ large (we also defined $C_2 = C_1 e$). Thus, by the union bound and the fact that $\{ d_{\outT}^-(v) \}_{v \in V}$ are identically-distributed, we obtain for all $n$ large and any constant $C > 0$,
\begin{equation}
\P \left( \max_{v \in V_n} d_{\outT}^-(v) \geq \frac{C \log n}{ \log \log n } \right) \leq n \left( \frac{C_2 \log \log n}{ C \log n } \right)^{ C \log n / \log \log n }  .
\end{equation}
Next, we note
\begin{align}
& \log \left( n \left( \frac{C_2 \log \log n}{ C \log n } \right)^{ C \log n / \log \log n } \right) = \log n + \frac{C \log n}{\log \log n} \left( \log ( C_2 \log \log n ) - \log ( C \log n ) \right) \\
& \quad = \log n \left( 1 +  \frac{ C \log \log ( \log n )^{C_2} }{ \log \log n } - \frac{ C \log \log n^C }{ \log \log n } \right) .
\end{align}
Choosing any $C \geq 1$ clearly implies
\begin{equation}
\frac{ \log \log n^C }{ \log \log n } \geq 1 .
\end{equation}
Also, since $C_2 > 0$ is a constant, we have for all $n$ large (for example)
\begin{equation}
\frac{ \log \log ( \log n )^{C_2} }{ \log \log n } < \frac{1}{2} .
\end{equation}
Combining the previous four lines, we then obtain, for all $n$ large,
\begin{equation}
\log \P \left( \max_{v \in V_n} d_{\outT}^-(v) > \frac{C \log n}{ \log \log n } \right) \leq  \left( 1 - C / 2 \right) \log n  ,
\end{equation}
so that choosing any $C > 2$ establishes the third statement.

\subsection{Proof of Lemma \ref{lemSigVecBounds}} \label{appProofLemSigVecBounds}

We begin with another lemma, which in fact holds for any underlying graph $G$.
\begin{lemma} \label{lemMaxRvLowerUpper}
For any graph $G = (V,E)$, any source node $s \in V$, and any iteration $k \in \{1,\ldots, d_{\outT}(s) \}$ of Algorithm \ref{algApproxPR},
\begin{equation}
\frac{1-\alpha}{ \max_{v \in V} d_{\outT}(v) } \leq \max_{v \in V} r^s_k(v) \leq \frac{1-\alpha}{ \min_{v \in V} d_{\outT}(v) } \exp \left( \frac{(1-\alpha) (k-1) }{ \min_{v \in V} d_{\outT}(v)} \right) .
\end{equation}
\end{lemma}
\begin{proof}
For the lower bound, first note $r_1^s(v) = (1-\alpha) / d_{\outT}(s)\ \forall\ v \in N_{\outT}(s)$. Furthermore, for each such $v$, $r_k^s(v)$ is non-decreasing in $k$ for $k < k_v$, where $k_v$ is the first iteration $k$ for which $v^*_k = v$. Also, since $v_1^* = s$, we must have $k_v \geq d_{\outT}(s)+1$ for some $v \in N_{\outT}(s)$. Hence, for any $k \in \{1,\ldots, d_{\outT}(s) \}$, we can find some $v \in N_{\outT}(s)$ for which $k_v > k$, which implies $r_k^s(v) \geq r_1^s(v) = (1-\alpha) / d_{\outT}(s)$. Since also $d_{\outT}(s) \leq \max_{v \in V} d_{\outT}(v)$, the lower bound follows.

For the upper bound, we use induction. For the base of induction, simply note
\begin{equation}
r_1^s(v) = \frac{1-\alpha}{d_{\outT}(s)} \leq \frac{1-\alpha}{ \min_{v \in V} d_{\outT}(v) }  \ \forall\ v \in N_{\outT}(s) .
\end{equation}
Now assuming the upper bound holds for $k-1$, we have for any $v \in V$,
\begin{align}
r^s_k(v) & \leq r^s_{k-1}(v) + \frac{1-\alpha}{d_{\outT}(v^*_k)} r^s_{k-1}(v_k^*) \leq \left( 1 +  \frac{1-\alpha}{ \min_{v \in V} d_{\outT}(v) } \right) \max_{v' \in V} r^s_{k-1}(v') \\
& \leq \left( 1 +  \frac{1-\alpha}{ \min_{v \in V} d_{\outT}(v) } \right) \frac{1-\alpha}{ \min_{v \in V} d_{\outT}(v) } \exp \left( \frac{(1-\alpha) (k-2) }{ \min_{v \in V} d_{\outT}}(v) \right) \\
& \leq \frac{1-\alpha}{ \min_{v \in V} d_{\outT}(v) }  \exp \left( \frac{(1-\alpha) (k-1) }{ \min_{v \in V} d_{\outT}(v)} \right) ,
\end{align}
where the first inequality uses the iterative update rule in Algorithm \ref{algApproxPR}, the second is immediate, the third uses the inductive hypothesis, and the fourth uses the standard inequality $1+x \leq e^x$. 
\end{proof}

We next state and prove a corollary of Lemma \ref{lemMaxRvLowerUpper}, which translates the $r^s_k$ bounds from Lemma \ref{lemMaxRvLowerUpper} to bounds regarding $\sigma^s_k$ (the actual vector of interest in the theorem).
\begin{corollary}  \label{corRvecBounds}
Let $G_n = (V_n = \{1,\ldots,n\} ,E_n)$ be a graph satisfying
\begin{equation}  \label{eqKeyAssumption}
d_{\outT}(v) \in \left( (1-\epsilon) p \sqrt{n} , (1+\epsilon) p \sqrt{n} \right)\ \forall\ v \in V_n 
\end{equation}
for some $p, \epsilon \in (0,1)$. Then for any $k \in \{ 1, \ldots ,  \lfloor (1-\epsilon) p \sqrt{n} \rfloor \}$ and any $s \in V_n$, 
\begin{gather} 
\frac{1-\alpha}{ 2 \sqrt{n} } < \max_{v \in V_n} r^s_k(v) < \frac{ e }{ (1-\epsilon) p \sqrt{n} } , \\
\frac{(1-\epsilon)(1-\alpha)}{ 4 \sqrt{n} } < r^s_k(v^*_{k+1}) < \frac{ 2 e }{ (1-\epsilon)^2 p \sqrt{n} } , \\
(1 - \alpha)- \frac{ 2 e (k-1)  }{ (1-\epsilon)^2 p \sqrt{n} } \leq \| r^s_k \|_1 \leq (1 - \alpha) - \frac{ (k-1) (1-\epsilon)(1-\alpha)}{ 4 \sqrt{n} } .
\end{gather}
\end{corollary}
\begin{proof}
Fix $k \in \{ 1, \ldots ,  \lfloor (1-\epsilon) p \sqrt{n} \rfloor \}$ and $s \in V_n$. Then $k < d_{\outT}(s)$ by \eqref{eqKeyAssumption} and the choice of $k$. We can then use the assumption \eqref{eqKeyAssumption}, Lemma \ref{lemMaxRvLowerUpper}, and the choice of $k$ to obtain
\begin{align} \label{eqFirstPairProof}
\frac{1-\alpha}{ (1+\epsilon) p \sqrt{n} } < \frac{1-\alpha}{ \max_{v \in V} d_{\outT}(v) } & \leq \max_{v \in V} r^s_k(v) \\
& \leq \frac{1-\alpha}{ \min_{v \in V} d_{\outT}(v) } \exp \left( \frac{(1-\alpha) (k-1) }{ \min_{v \in V} d_{\outT}(v)} \right) < \frac{ (1-\alpha) e^{1-\alpha} }{ (1-\epsilon) p \sqrt{n} } .
\end{align}
Finally, $\epsilon, p, \alpha \in (0,1)$ yields the first pair of inequalities. Next, by definition of $v_{k+1}^*$ in Algorithm \ref{algApproxPR},
\begin{equation} \label{eqSecondPairVstarDefn}
r^s_k(v^*_{k+1}) = d_{\outT}(v^*_{k+1}) \frac{ r^s_k(v^*_{k+1})  }{ d_{\outT}(v^*_{k+1}) } = d_{\outT}(v^*_{k+1}) \max_{v \in V_n} \frac{ r^s_k(v)  }{ d_{\outT}(v) } = \max_{v \in V_n} \frac{  d_{\outT}(v^*_{k+1}) }{  d_{\outT}(v) } r^s_k(v) .
\end{equation}
On the other hand, by the assumption \eqref{eqKeyAssumption}, and since $\epsilon \in (0,1)$,
\begin{equation} \label{eqSecondPairDegRatio}
\frac{1-\epsilon}{2} < \frac{1-\epsilon}{1+\epsilon} < \frac{  d_{\outT}(v^*_{k+1}) }{  d_{\outT}(v) } < \frac{1+\epsilon}{1-\epsilon} < \frac{2}{1-\epsilon}
\end{equation}
Combining \eqref{eqSecondPairVstarDefn} and \eqref{eqSecondPairDegRatio}, and using the first pair of inequalities, yields the second pair of inequalities. For the third pair of inequalities, we first assume $k > 1$ and use \eqref{eqFDPl1Decrease} from Appendix \ref{PROOF_FWBWMCMC_COMP} to obtain
\begin{equation}
\| r^s_k \|_1 = \| r^s_{k-1} \|_1 - \alpha r^s_{k-1}(v_k^*)  = \cdots = \| r^s_0 \|_1 - \alpha r^s_0(v_1^*) - \alpha \sum_{j=1}^{k-1} r^s_j(v^*_{j+1}) = 1 - \alpha - \alpha \sum_{j=1}^{k-1} r^s_j(v^*_{j+1}) , 
\end{equation}
where we also used $r^s_0 = e_s, v^*_1 = s$ by Algorithm \ref{algApproxPR}. We can then use the second pair of inequalities to obtain the third pair of inequalities. If instead $k=1$, we immediately have $\| r^s_k \|_1 = 1-\alpha$, which is precisely the third pair of inequalities in the case $k=1$.
\end{proof}

We can now prove Lemma \ref{lemSigVecBounds}. For the first bound, note the assumptions of Lemma \ref{lemSigVecBounds} are stronger than those of Corollary \ref{corRvecBounds}, so we can use Corollary \ref{corRvecBounds} to obtain
\begin{equation}
\sigma_k^s(v) = \frac{ r_k^s(v) }{ \| r_k^s \|_1 } < \frac{ \frac{ e }{ (1-\epsilon) p \sqrt{n} } }{ (1 - \alpha)- \frac{ 2 e (k-1)  }{ (1-\epsilon)^2 p \sqrt{n} } }  = \frac{1}{\sqrt{n}} \frac{ (1-\epsilon)e }{ (1-\alpha)(1-\epsilon)^2 p - 2 e (k-1) / \sqrt{n} } .
\end{equation}
Using the trivial inequalities $1-\epsilon < 1, k-1<k$ in the final expression then yields the first upper bound. (Note the assumed upper bound on $k$ ensures the denominator is non-negative.)

For the second bound, let $S_n \subset V_n$ be a set containing $s$. We begin by showing
\begin{equation} \label{eqRcecBoundsOutComm}
\sum_{v \in V_n \setminus S_n} r_k^s ( v ) < \frac{ \sqrt{n} + 2 e (k-1) }{ (1-\epsilon)^3 p^2 n } \max_{s' \in S_n} \left| N_{\outT}(s') \setminus S_n \right| .
\end{equation}
To prove \eqref{eqRcecBoundsOutComm}, we use induction. For $k=1$, the $r^s$ update in Algorithm \ref{algApproxPR} implies
\begin{equation}
r^s_1(v) = \frac{1-\alpha}{d_{\outT}(s)} \mathbbm{1}_{ \{ v \in N_{\outT}(s) \} } \Rightarrow \sum_{v \in V_n \setminus S_n} r^s_k(v) = \frac{1-\alpha}{d_{\outT}(s)} \left| N_{\outT}(s) \setminus S_n \right| ,
\end{equation}
which, using the assumption \eqref{eqKeyAssumption2} and $\alpha, \epsilon, p \in (0,1)$, can clearly be bounded as
\begin{equation}
\sum_{v \in V_n \setminus S_n} r^s_k(v) < \frac{1}{ (1-\epsilon)^3 p^2 \sqrt{n} } \max_{s' \in S_n} \left| N_{\outT}(s') \setminus S_n \right| ,
\end{equation}
which establishes \eqref{eqRcecBoundsOutComm} when $k=1$. Now assume \eqref{eqRcecBoundsOutComm} holds for $k-1$ and consider two cases:
\begin{enumerate}
\item $v_k^* \notin S_n$: We can write the $r^s$ update in Algorithm \ref{algApproxPR} as
\begin{equation} \label{eqPushUpdateForKeyCor}
r_k^s(v) = r_{k-1}^s(v) + \frac{1-\alpha}{d_{\outT}(v_k^*)} r_{k-1}(v_k^*) \mathbbm{1}_{ \{ v \in N_{\outT}(v_k^*) \} } - r_{k-1}(v_k^*) \mathbbm{1}_{ \{ v = v_k^* \} }\ \forall\ v \in V_n ,
\end{equation}
where $\mathbbm{1}_A$ is the indicator function of the event $A$. This clearly implies
\begin{equation}
\sum_{v \in V_n \setminus S_n} r_k^s ( v )  = \sum_{v \in V_n \setminus S_n} r_{k-1}^s ( v ) + r_{k-1}(v_k^*) \left( \frac{1-\alpha}{ d_{\outT}(v_k^*)} \left| N_{\outT}(v_k^*) \setminus S_n \right| - 1 \right)  < \sum_{v \in V_n \setminus S_n} r_{k-1}^s ( v ) ,
\end{equation}
from which the inductive hypothesis completes the proof, since the upper bound in \eqref{eqRcecBoundsOutComm} increases with $k$.
\item $v_k^* \in S_n$: Again using \eqref{eqPushUpdateForKeyCor}, we observe
\begin{equation} \label{eqPushUpdateForKeyCorCase2}
\sum_{v \in V_n \setminus S_n} r_k^s ( v ) = \sum_{v \in V_n \setminus S_n} r_{k-1}^s ( v ) + \frac{1-\alpha}{d_{\outT}(v_k^*)} r_{k-1}(v_k^*) \left| N_{\outT}(v_k^*) \setminus S_n \right| .
\end{equation}
(Note the final term in \eqref{eqPushUpdateForKeyCor} does not appear in \eqref{eqPushUpdateForKeyCorCase2}, since $\sum_{v \in V_n \setminus S_n} \mathbbm{1}_{ \{ v = v_k^* \} } = 0$ when $v_k^* \in S_n$.) For the second summand in \eqref{eqPushUpdateForKeyCorCase2}, we use the second upper bound from Corollary \ref{corRvecBounds}, the assumption \eqref{eqKeyAssumption2}, $\alpha \in (0,1)$, and the assumed case $v_k^* \in S_n$ to obtain
\begin{equation}
\frac{1-\alpha}{d_{\outT}(v_k^*)} r_{k-1}(v_k^*) \left| N_{\outT}(v_k^*) \setminus S_n \right| < \frac{2e}{ (1-\epsilon)^3 p^2 n } \max_{s' \in S_n} \left| N_{\outT}(s') \setminus S_n \right| ,
\end{equation}
Substituting into \eqref{eqPushUpdateForKeyCorCase2} and using the inductive hypothesis yields
\begin{equation}
\sum_{v \in V_n \setminus S_n} r_k^s ( v )  < \left( \frac{ \sqrt{n} + 2 e (k-2) }{ (1-\epsilon)^3 p^2 n } + \frac{2e}{ (1-\epsilon)^3 p^2 n } \right) \max_{s' \in S} \left| N_{\outT}(s') \setminus S_n \right| ,
\end{equation}
which completes the proof.
\end{enumerate} 
Combining \eqref{eqRcecBoundsOutComm} with the lower bound for $\| r^s_k \|_1$ from Corollary \ref{corRvecBounds} gives
\begin{align}
\sum_{v \in V_n \setminus S_n} \sigma_k^s ( v ) & = \frac{\sum_{v \in V_n \setminus S_n} r_k^s ( v )}{ \| r^s_k \|_1 } < \frac{\frac{ \sqrt{n} + 2 e (k-1) }{ (1-\epsilon)^3 p^2 n } \max_{s' \in S} | N_{\outT}(s') \setminus S_n |}{ (1 - \alpha)- \frac{ 2 e (k-1)  }{ (1-\epsilon)^2 p \sqrt{n} }  } \\
& = \frac{ \max_{s' \in S_n} | N_{\outT}(s') \setminus S_n | }{ \sqrt{n} }  \frac{  1+ 2 e (k-1) / \sqrt{n}  }{ (1-\epsilon)p \left(  (1-\alpha) (1-\epsilon)^2 p  - 2 e (k-1) / \sqrt{n} \right)  } ,
\end{align}
from which the trivial bound $k-1<k$ completes the proof.

\section{Choosing order of targets in Algorithm \ref{algApproxContMany}} \label{secTargetOrder}

As mentioned at the end of Section \ref{secManyTarget}, the performance of Algorithm \ref{algApproxContMany} can significantly depend on the order in which the targets $t_1, t_2, \ldots, t_{|T|}$ are chosen. For instance, suppose there exists $t^* \in T$ such that $\pi_{t^*}(t') > r^t_{\max}\ \forall\ t' \in T$, but $\pi_{t}(t') \leq r^t_{\max}\ \forall\ t \in T \setminus \{ t^* \}, t' \in T$. Then choosing $t_1 = t^*$ implies $c_T = |T|-1$, while choosing $t_{|T|} = t^*$ implies $c_T = 0$. More generally, the algorithm is most efficient when any $t$ satisfying $\pi_t(t') > r^t_{\max}$ for many $t' \in T$ is chosen ``early'' in the algorithm, i.e.\ $t_i = t$ for small $i$. However, because $\pi_t(t')$ is unknown, optimizing the order $t_1, t_2, \ldots, t_{|T|}$ at runtime is difficult. A possible workaround is to use $p^{t'}(t)$ as a proxy for $\pi_t(t')$, since $p^{t'}(t) \in [ \pi_{t}(t') - r^t_{\max}, \pi_{t}(t')]$ by the invariant \eqref{eqBwInvariant}. Unfortunately, even this proxy is difficult to utilize at runtime. This is because we would like to choose $t_i$ such that $\pi_{t_j}(t_i)$ is large for many $j < i$, but the proxy $p^{t_i}(t_j)$ of $\pi_{t_j}(t_i)$ is only known \textit{after} choosing $t_i$. (Loosely speaking, we have a ``chicken and egg'' scenario.) Hence, we do not suspect there is a provably optimal method, or even a simple heuristic but suboptimal method, for choosing the order of targets at runtime.

\section{Details on Section \ref{secExperiments} experiments} \label{secExperimentDetails}

\textbf{Datasets:} \texttt{Direct-ER} is a directed Erd\H{o}s-R\'enyi graph with parameters $n = 2000, p = 0.005$ (edge $v \rightarrow u$ is present with probability $p$, independent of other edges, $\forall\ v , u \in V , v \neq u$). \texttt{Direct-SBM} is a directed stochastic block model; there are $n = 2000$ nodes partitioned into $k = 20$ disjoint communities, each of size $\frac{n}{k} = 100$; directed edges occur with probability $9 / ( \frac{n}{k} - 1 )$ between distinct nodes in the same community and with probability $1 / ( n - \frac{n}{k} )$ between nodes in different communities (so that each node has nine neighbors in its own community and one neighbor in another community, in expectation, yielding a highly modular graph). The real graphs used are available from the Stanford Network Analysis Platform (SNAP) \cite{snapnets}; see Table \ref{tabDatasets} for further details.

\begin{table}
\caption{Dataset details} \label{tabDatasets}
\begin{tabular}{|l|l|c|c|c|}
\hline 
Dataset & Description & $n$ & $m$ \\ \hline 
\texttt{com-Amazon} & Amazon co-purchasing & 334863 & 925872  \\ \hline 
\texttt{com-dblp} & Scientific co-authorship & 317080 & 1049866  \\ \hline 
\texttt{roadNet-PA} & Roads in Pennsylvania & 1087532 & 1541514 \\ \hline
\texttt{Slashdot} & Friendships on technology news site & 71307 & 912381 \\ \hline
\texttt{web-BerkStan} & berkley.edu, stanford.edu web graph & 334857 & 4523232  \\ \hline 
\texttt{web-Google} & Partial web crawl & 434818 & 3419124  \\ \hline 
\texttt{Wiki-Talk} & Friendships among Wikipedia editors & 111881 & 1477893 \\ \hline
\end{tabular} 
\end{table}

\vspace{5pt} \noindent \textbf{Parameters:} For the scalar estimation experiments in Sections \ref{secExpSynScal} and \ref{secExpRealScal}, we use the algorithmic parameters shown in Table \ref{tabParam}. More specifically, \texttt{FW-BW-MCMC} uses Algorithm \ref{algApproxPROriginal} for forward DP with parameter $\tilde{r}^s_{\max}$ and samples $w \| \tilde{r}^s \|_1$ random walk starting node locations for each source $s$ (as in Algorithm \ref{algOurEstimatorPractical}), uses the walk sharing scheme from Section \ref{secManySource} to sample walks jointly across $S$, and uses Algorithm \ref{algApproxContMany} with parameter $r^t_{\max}$ for the targets; for \texttt{Bidirectional-PPR}, we sample $w$ walks separately for each source and run Algorithm \ref{algApproxCont} separately for each target. In practice, we find that $w$ given by the accuracy guarantee (Theorem \ref{STATE_FWBWMCMC_ACC}) is overly pessimistic, so we instead set $w = \frac{c r^t_{\max} }{\delta}$ for both methods, with $c$ given in the table. For the matrix experiments in Sections \ref{secExpSynMat} and \ref{secExpRealMat}, we use the same $\tilde{r}^s_{\max}$ and $r^t_{\max}$ values. Furthermore, we set $w = l \tfrac{c r^t_{\max} }{\delta}$, $w = \| \Sigma \|_{\infty,1} \tfrac{c r^t_{\max} }{\delta}$, and $w = \sqrt{l\ \srank(P_T(S,:)+P_S^{\trans} R_T)} \tfrac{c r^t_{\max} }{\delta}$ for the baseline, $\sigma_{\max}$, and $\sigma_{\avg}$ schemes, respectively.

\vspace{5pt} \noindent \textbf{Single pair performance:} The parameters in Table \ref{tabParam} were chosen so the primitives \texttt{FW-BW-MCMC- Practical} and \texttt{Bidirectional-PPR} offer similar accuracy in the single pair case and balance runtime between dynamic programming (DP) and Monte Carlo (MC). To demonstrate this, we show statistics in Table \ref{tabParam}. We obtained the statistics by averaging across $10^3$ trials of the following procedure. First, we sample $t \in V$ uniformly. Next, we sample a ``significant'' source $s$ (i.e.\ $s$ satisfying $\pi_s(t) > \delta$) and an ``insignificant'' source $s'$ (i.e.\ $s'$ satisfying $\pi_{s'}(t) < \delta$). Since Theorem \ref{STATE_FWBWMCMC_ACC} bounds relative and absolute error for significant and insignificant pairs, respectively, we compute relative and absolute error for the $\pi_s(t)$ and $\pi_{s'}(t)$ estimates, respectively. (We do not report absolute error statistics as no insignificant estimate violated the absolute error guarantee.) For real datasets, we cannot compute $\pi_s(t)$ to test error performance; instead, we run Algorithm \ref{algApproxCont} with $r^t_{\max}$ replaced by $\eta = \frac{1}{n}$, denote the output $p^t_{\eta}, r^t_{\eta}$, and bound relative error for significant pairs as
\begin{align}
\tfrac{  | \hat{\pi}_s(t) - \pi(t) | }{\pi_s(t)} \leq \tfrac{ | \hat{\pi}_s(t) - p^t_{\eta}(s) |  + | p^t_{\eta}(s) - \pi_s(t) | }{ p^t_{\eta}(s) }  \leq \tfrac{ | \hat{\pi}_s(t) - p^t_{\eta}(s) | + \| r^t_{\eta} \|_{\infty} }{ p^t_{\eta}(s) }  < \tfrac{ | \hat{\pi}_s(t) - p^t_{\eta}(s) | }{ p^t_{\eta}(s) } + \tfrac{1}{10} ,
\end{align}
where we have used $p_{\eta}^t(s) \in [ \pi_s(t) - \| r^t_{\eta} \|_{\infty} , \pi_s(t)]$ (which holds by \eqref{eqBwInvariant}), $\| r^t_{\eta} \|_{\infty} < \eta = \frac{1}{n}$ (which holds by Algorithm \ref{algApproxCont}), and $p_{\eta}^t(s) \geq \delta = \frac{10}{n}$ (which holds by choice of $s,t$). In the same manner, we can bound absolute error for insignificant pairs as $| \hat{\pi}_s(t) - \pi(t) | \leq | \hat{\pi}_s(t) - p_{\eta}^t(s) | + \frac{1}{n}$. (Note we choose significant pairs as those $(s,t)$ satisfying $p^t_{\eta}(s) \geq \delta$, since then $\pi_s(t) \geq \delta$ by \eqref{eqBwInvariant}; similarly, we choose insignificant pairs as those $(s',t)$ satisfying $p^t_{\eta}(s') < \delta - \eta$, since then $\pi_{s'}(t) < \delta$ by \eqref{eqBwInvariant}.)

\begin{table}
\caption{Experiment parameters and single pair performance} \label{tabParam}
\begin{tabular}{ |C{0.84in}|C{1.3in}|C{0.25in}|C{0.25in}|C{0.28in}|C{0.18in}|C{0.3in}|C{0.3in}|C{0.3in}| } \hline
Graph & Algorithm & $\tilde{r}^s_{\max}$ $\times 10^3$ & $r^t_{\max}$ $\times 10^3$ & $\delta$ & $c$ & DP time (ms) & MC time (ms) & Error \\ \hline
\texttt{Direct-ER} & \texttt{FW-BW-MCMC-Prac} & $1.8$ & $3.8$ & $1/n$ & $7$ & 10.61 & 7.63 & 0.075 \\ \hline
\texttt{Direct-ER} & \texttt{Bidirectional-PPR} & N/A & $1.6$ & $1/n$ & $12$ & 6.94 & 7.52 & 0.072 \\ \hline
\texttt{Direct-SBM} & \texttt{FW-BW-MCMC-Prac} & $1$ & $4 $ & $1/n$ & $7$ & 15.43 & 7.08 & 0.052 \\ \hline
\texttt{Direct-SBM} & \texttt{Bidirectional-PPR} & N/A & $3 $ & $1/n$ & $10$ & 10.19 & 12.01 & 0.061 \\ \hline
\texttt{com-amazon} & \texttt{FW-BW-MCMC-Prac} & $3.6$ & $18.2 $ & $10/n$ & $12 $ & 22.55 & 22.54 & 0.12 \\ \hline
\texttt{com-amazon} & \texttt{Bidirectional-PPR} & N/A & $7.4 $ & $10/n$ & $13$ & 22.13 & 22.21 & 0.11 \\ \hline
\texttt{com-dblp} & \texttt{FW-BW-MCMC-Prac} & $2.9$ & $14.3 $ & $10/n$ & $13$ & 20.27 & 20.31 & 0.12 \\ \hline
\texttt{com-dblp} & \texttt{Bidirectional-PPR} & N/A & $6 $ & $10/n$ & $15 $ & 20.03 & 19.65 & 0.11 \\ \hline
\texttt{roadNet-PA} & \texttt{FW-BW-MCMC-Prac} & $15.1$ & $34.8 $ & $10/n$ & $6 $  & 55.04 & 56.58 & 0.11 \\ \hline
\texttt{roadNet-PA} & \texttt{Bidirectional-PPR} & N/A & $12.8 $ & $10/n$ & $6$ & 53.19 & 55.96 & 0.10 \\ \hline
\texttt{Slashdot} & \texttt{FW-BW-MCMC-Prac} & $2$ & $12.2 $ & $10/n$ & $7$ & 3.08 & 3.38 & 0.10 \\ \hline
\texttt{Slashdot} & \texttt{Bidirectional-PPR} & N/A & $4.2 $ & $10/n$ & $17$ & 3.30 & 4.03 & 0.11 \\ \hline
\texttt{web-BerkStan} & \texttt{FW-BW-MCMC-Prac} & $6.9$ & $23 $ & $10/n$ & $3$ & 11.13 & 11.02 & 0.12 \\ \hline
\texttt{web-BerkStan} & \texttt{Bidirectional-PPR} & N/A & $11.6 $ & $10/n$ & $3$ & 8.40 & 8.42 & 0.12 \\ \hline
\texttt{web-Google} & \texttt{FW-BW-MCMC-Prac} & $4.5$ & $17.6 $ & $10/n$ & $8$ & 23.33 & 22.83 & 0.11 \\ \hline
\texttt{web-Google} & \texttt{Bidirectional-PPR} & N/A & $6.7 $ & $10/n$ & $11$ & 26.07 & 22.29 & 0.11 \\ \hline
\texttt{WikiTalk} & \texttt{FW-BW-MCMC-Prac} & $2.3$ & $7.5 $ & $10/n$ & $8$ & 4.40 & 3.99 & 0.11 \\ \hline
\texttt{WikiTalk} & \texttt{Bidirectional-PPR} & N/A & $2.9 $ & $10/n$ & $20$ & 5.84 & 5.10 & 0.11 \\ \hline
\end{tabular}
\end{table}

\vspace{5pt} \noindent {\textbf{Additional Erd\H{o}s-R\'enyi results:}} We also ran the first experiment from Section \ref{secExpSynScal} for Erd\H{o}s-R\'enyi graphs with $n \in \{ 4000, 8000 \}$, each with edge formation probability $10/n$. For \texttt{FW-BW-MCMC}, we used  parameters $(\tilde{r}_{\max}^s, r^t_{\max}) = ( 1.5 , 3.5 ) \times 10^{-3}$ when $n = 4000$ and $(\tilde{r}_{\max}^s, r^t_{\max}) = ( 1.2 , 3.2 ) \times 10^{-3}$ when $n = 8000$ (choosing smaller parameters for larger $n$ gave more balanced runtime than using the $n = 2000$ parameters from Table \ref{tabParam}). Similarly, for \texttt{Bidirectional-PPR}, we used $r^t_{\max} = 1.1 \times 10^{-3}$ when $n = 2000$ and $r^t_{\max} = 0.8 \times 10^{-3}$ when $n = 8000$. As in Table \ref{tabParam}, we ensured these parameters gave similar accuracy for both algorithms. Results are shown in Fig.\ \ref{figOtherEr}. As mentioned in Section \ref{secExpSynScal}, the plots are qualitatively similar across $n$; however, they improve slightly as $n$ grows. For instance, in the extreme case $|S| = |T| = n/2$, \texttt{FW-BW-MCMC-Prac} was (on average) 2.9, 4.5, and 5.8 times faster than \texttt{Bidirectional-PPR} for $n = 2000$, $n = 4000$, and $n = 8000$, respectively.

\begin{figure}
\centering
\begin{subfigure}[b]{\columnwidth}
\centering\
\includegraphics[height=\plotHeight]{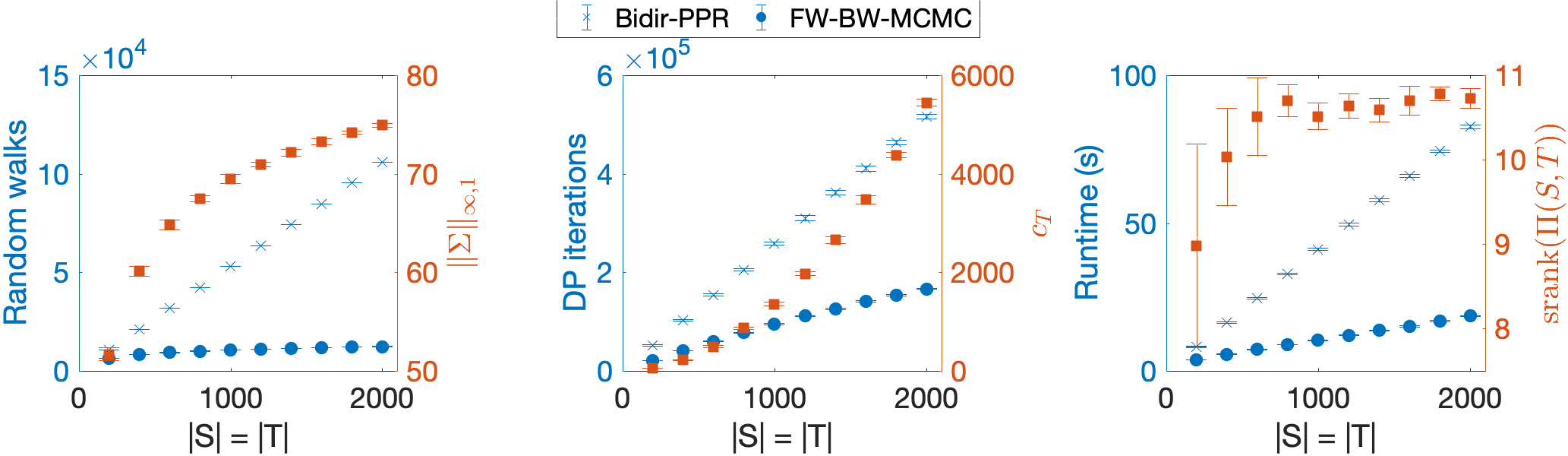} 
\end{subfigure}
\begin{subfigure}[b]{\columnwidth}
\centering\
\includegraphics[height=\plotHeight]{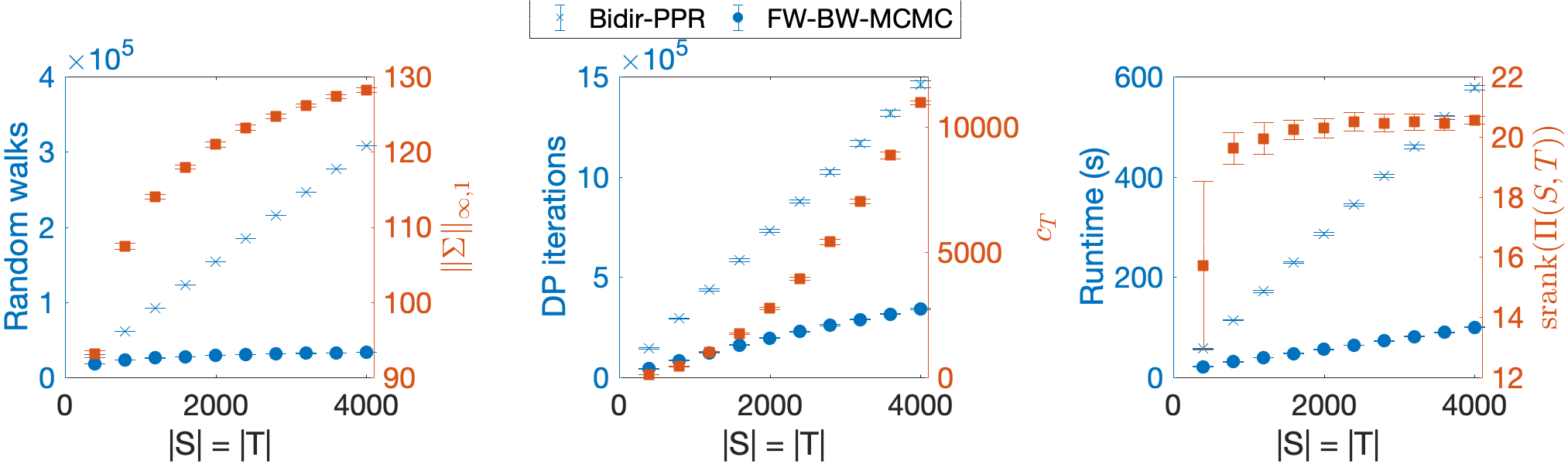} 
\end{subfigure}
\caption{Replicating Erd\H{o}s-R\'enyi experiment from Section \ref{secExpSynScal} with $n = 4000$ (top) and $n = 8000$ (bottom).} \label{figOtherEr}
\end{figure}

\vspace{5pt} \noindent \textbf{Building clustered subsets:} As mentioned in Section \ref{secExpReal}, we use a simple algorithm to randomly construct clustered subsets of nodes for experiments; Algorithm \ref{algBuildSet} provides a formal definition.

\begin{algorithm}
\caption{ $U = \texttt{Construct-Clustered-Set} (G,l)$} \label{algBuildSet}
 Choose $u \in V$ uniformly at random, let $U = \{ u \}$ \\
\For{ $i = 2$ \KwTo $l$}{
	Sample $w \in \left( \cup_{u \in U} N_{\textrm{out}}(u) \right) \setminus U$ with probability proportional to $\frac{\sum_{u \in U} \mathbbm{1}_{ \{ w \in N_{\textrm{out}} (u) \} } }{N_{\textrm{in}}(w)}$; set $U \leftarrow U \cup \{ w \}$
}
\end{algorithm}

\section{Additional experiments for distributed setting} \label{secAdditionalDist}

\subsection{Matrix approximation, $\sigma_{\avg}$ approach} \label{secDsMat}

In this section, we describe a scheme to use the $\sigma_{\avg}$ variant of Algorithm \ref{algMatrixEstimator} in the distributed setting from Section \ref{secDistributed}. Our scheme is quite similar to that defined in Section \ref{secDistributed} and proceeds as follows. First, we arbitrarily partition $S$ into $k$ subsets of size $|S|/k$, and we use the $i$-th machine to run forward DP (Algorithm \ref{algApproxPR}) for each source $s$ belonging to the $i$-th subset. Next, we create another partition $\{ S_i \}_{i = 1}^k$ of $S$ and use the $i$-th machine to sample random walks for $S_i$ using the $\sigma_{\avg}$ variant of Algorithm \ref{algMatrixEstimator}. Finally, we construct the estimate $\hat{\Pi}(S,T)$ of $\Pi(S,T)$ as in Algorithm \ref{algMatrixEstimator}.


It remains to specify the construction of $\{ S_i \}_{i = 1}^k$. For this, we first use the output $p^s$ of Algorithm \ref{algApproxPR} to define $\surr_s = P_T(s,:) + (p^s)^{\trans} R_T$ for each $s \in S$; here $P_T$ and $R_T$ are the matrices with columns $\{ p^t \}_{t \in T}$ and $\{ r^t \}_{t \in T}$, respectively (with each $(p^t, r^t)$ computed offline via Algorithm \ref{algApproxCont} as in Section \ref{secDistributed}). Note that $\surr_s$ is a row of the surrogate matrix $P_T(S,:)+P_S^{\trans} R_T$ discussed at the conclusion of Section \ref{secManyPairMatrix}. For $S' \subset S$, we also define  $\surr_{S'}$ be the matrix with rows $\{ \surr_s \}_{s \in S'}$. Now, as in Section \ref{secExpRealMat}, the number of walks sampled on the $i$-th machine will be set proportional to $\sqrt{ |S_i| \srank ( \surr_{S_i} ) ) }$; hence, our goal is to construct $\{ S_i \}_{i = 1}^k$ so as to minimize
\begin{equation}\label{eqObjectiveSrank}
\max_{i \in \{1,\ldots,k\}} \sqrt{ |S_i| \srank ( \surr_{S_i} ) } .
\end{equation}

To approximate the solution of this minimization problem, we consider a heuristic method defined in Algorithm \ref{algHeuristicSrankExact}. Note this is similar to Algorithm \ref{algHeuristic} in Section \ref{secDistributed}: first, we assign one source to each $S_i$, while attempting to choose these $s$ with $\surr_s$ vectors far apart; next, we iteratively assign the remaining $|S|-k$ nodes to some $S_i$, while attempting to minimize the cost of this assignment. In light of \eqref{eqObjectiveSrank}, we here define the cost of assigning $s$ to $S_i$ as $\tilde{d}(s,S_i) = \sqrt{ ( |S_i| + 1 ) \srank ( \surr_{S_i \cup \{ s \} } ) }$.

\begin{algorithm}
\caption{ $\{ S_i \}_{i=1}^k$ = \texttt{Source-Partition-$\sigma_{\avg}$} ($\{ \surr_s \}_{s \in S}$, $k$)  } \label{algHeuristicSrankExact}
Draw $s \sim S$ uniformly, set $S_1 = \{ s \}$; set $S_i = \emptyset\ \forall\ i \in \{2,\ldots,k\}$ \\
\For{$i = 2$ \KwTo $k$}{
	 Draw $s \sim S$ with probability proportional to $\min_{j \in \{1,\ldots,i-1\}} \| \surr_s - \surr_{S_j} \|_1$; set $S_i = \{ s \}$
}
\For{$i = k+1$ \KwTo $|S|$}{
	 Choose any $s \in S \setminus ( \cup_{j=1}^k S_j )$ (any $s$ not yet assigned); compute $\tilde{d}(s,S_j)\ \forall\ j \in \{1,\ldots,k\}$ \\
	 Let $j^* \in \argmin_j  \tilde{d}(s,S_j)$ , $S_{j^*} = S_{j^*} \cup \{ s \}$.
}
\end{algorithm}

Unfortunately, Algorithm \ref{algHeuristicSrankExact} requires the singular value decomposition (SVD) of $\surr_{ S_j \cup \{ s \} }$ to be computed, so that $\tilde{d}(s,S_j)$ can computed in the second for loop of Algorithm \ref{algHeuristicSrankExact}. (In contrast, computing $d(s,S_j)$ in the $\sigma_{\max}$ partitioning scheme, Algorithm \ref{algHeuristic}, only requires subtracting one vector from another.) Hence, we also propose an alternative partitioning method that avoids this SVD computation. This method is based on two observations. First, we have
\begin{align}
\| \surr_{S_j \cup \{ s \} } \|_2^2 & = \lambda_{\max} \left( \begin{bmatrix} \surr_{S_j}^{\trans} & \surr_s^{\trans} \end{bmatrix} \begin{bmatrix} \surr_{S_j} \\ \surr_s \end{bmatrix} \right) = \lambda_{\max} \left( \sum_{s' \in S_j} \surr_{s'}^{\trans} \surr_{s'} + \surr_s^{\trans} \surr_s \right) \\
& \leq \max_{t \in T} \sum_{t' \in T} \left( \sum_{s' \in S_j} \surr_{s'}^{\trans} \surr_{s'} + \surr_s^{\trans} \surr_s \right)(t,t') \\
& = \max_{t \in T} \left( \sum_{s' \in S_j} \surr_{s'}(t) \| \surr_{s'} \|_1 + \surr_s(t) \| \surr_s \|_1 \right) ,
\end{align}
where the first equality is a well-known result, the inequality follows from the Perron-Frobenius Theorem, and the remaining equalities are straightforward. Second, by definition of $\| \cdot \|_F$, we have
\begin{equation}
\| \surr_{S_j \cup \{ s \} } \|_F^2 = \sum_{s' \in S_j} \| \surr_{s'} \|_2^2 + \| \surr_s \|_2^2 .
\end{equation}
Combining these observations, we obtain
\begin{equation}\label{eqDtildeDefn}
\tilde{d}(s,S_j) \geq \hat{d}(s,S_j) = \sqrt{ \frac{ \left( |S_j| + 1 \right) \left( \sum_{s' \in S_j} \| \surr_{s'} \|_2^2 + \| \surr_s \|_2^2 \right) }{ \max_{t \in T} \left( \sum_{s' \in S_j} \surr_{s'}(t) \| \surr_{s'} \|_1 + \surr_s(t) \| \surr_s \|_1 \right) } } .
\end{equation}
This expression allows us to estimate $\tilde{d}(s,S_j)$ more efficiently than it can be computed exactly. In Algorithm \ref{algHeuristicSrankApprox}, we give a partitioning scheme that leverages this insight. Note that the computation of $\hat{d}(s,S_j)$ in Algorithm \ref{algHeuristicSrankApprox} can be performed as
\begin{equation}
\hat{d}(s,S_j) =  \sqrt{ \frac{ \left( |S_j| + 1 \right) \left( x_j + \| \surr_s \|_2^2 \right) }{ \max_{t \in T} \left( y_j(t) + \surr_s(t) \| \surr_s \|_1 \right) } } ,
\end{equation}
i.e.\ the terms $\sum_{s' \in S_j} \| \surr_{s'} \|_2^2$ and $\sum_{s' \in S_j} \surr_{s'}(t) \| \surr_{s'} \|_1$ in \eqref{eqDtildeDefn} have already been computed as $x_j$ and $y_j(t)$ when $\hat{d}(s,S_j)$ is computed; furthermore, $x_j$ and $y_j(t)$ are updated (rather than being computed in full) each time some $s$ is added to $S_j$ (last line of Algorithm \ref{algHeuristicSrankApprox}).

\begin{algorithm}
\caption{ $\{ S_i \}_{i=1}^k$ = \texttt{Source-Partition-$\sigma_{\avg}$-alt} ($\{ \surr_s \}_{s \in S}$, $k$)  } \label{algHeuristicSrankApprox}
Draw $s \sim S$ uniformly, set $S_1 = \{ s \}$, $x_1 = \| \surr_s \|_2^2$, $y_1(t) = \surr_s(t) \| \surr_s \|_1\ \forall\ t \in T$ \\
Set $S_i = \emptyset, x_i = y_i = 0\ \forall\ i \in \{2,\ldots,k\}$ \\
\For{$i = 2$ \KwTo $k$}{
	 Draw $s \sim S$ with probability proportional to $\min_{j \in \{1,\ldots,i-1\}} \| \surr_s - \surr_{S_j} \|_1$ \\
	 Set $S_i = \{ s \}$, $x_i = \| \surr_s \|_2^2$, $y_i(t) = \surr_s(t) \| \surr_s \|_1\ \forall\ t \in T$
}
\For{$i = k+1$ \KwTo $|S|$}{
	 Choose any $s \in S \setminus ( \cup_{j=1}^k S_j )$ (any $s$ not yet assigned); compute $\hat{d}(s,S_j)\ \forall\ j \in \{1,\ldots,k\}$  \\
	 Let $j^* \in \argmin_j  \hat{d}(s,S_j)$ \\
	 Set $x_{j^*} = x_{j^*} + \| \surr_s \|_2^2$, $y_{j^*}(t) = y_{j^*}(t) + \surr_s(t) \| \surr_s \|_1\ \forall\ t \in T$, $S_{j^*} = S_{j^*} \cup \{ s \}$ 
}
\end{algorithm}

In Fig.\ \ref{figDsMat}, we present empirical results for the $\sigma_{\avg}$ matrix approximation scheme in the distributed setting. In particular, we show results for the scheme described above with the partition $\{ S_i \}_{i = 1}^k$ constructed via Algorithm \ref{algHeuristicSrankExact} (``Heuristic'' in Fig.\ \ref{figDsMat}) and via Algorithm \ref{algHeuristicSrankApprox} (``Alt Heuristic'' in Fig.\ \ref{figDsMat}). For both schemes, we show the maximum forward DP and random walk sampling time across machines, the maximum number of walks sampled across machines, and the value of the objective function \eqref{eqObjectiveSrank}. The first two quantities are shown relative to the respective quantities for a baseline scheme, which arbitrarily partitions $S$ into subsets of size $|S|/k$ and uses the $i$-th machine to run the baseline matrix approximation scheme from Section \ref{secExpRealMat} for the $i$-th subset (recall no forward DP is used for this baseline scheme, i.e.\ walks are not shared across sources). For this experiment, we let $S = \{ \tilde{S}_i \}_{i = 1}^k$, where $k = 10$ and each $\tilde{S}_i$ is a clustered subset satisfying $| \tilde{S}_i | = 100$; we also compare to an oracle scheme that sets $S_i = \tilde{S}_i$ (as in Section \ref{secDistributed}). In general, Fig.\ \ref{figDsMat} conveys the same message as Fig.\ \ref{figDistSetting} in Section \ref{secDistributed}: our  methods perform similarly to the oracle method and noticeably outperform the baseline. Here we also note that the heuristic outperforms the oracle across graphs, while the oracle in turn outperforms the alternative heuristic. Nevertheless, the alternative heuristic offers similar performance as the other schemes, while avoiding the SVD computation of the heuristic (which we expect would become prohibitively costly as $S$ grows).

\begin{figure}
\centering
\includegraphics[width=\columnwidth]{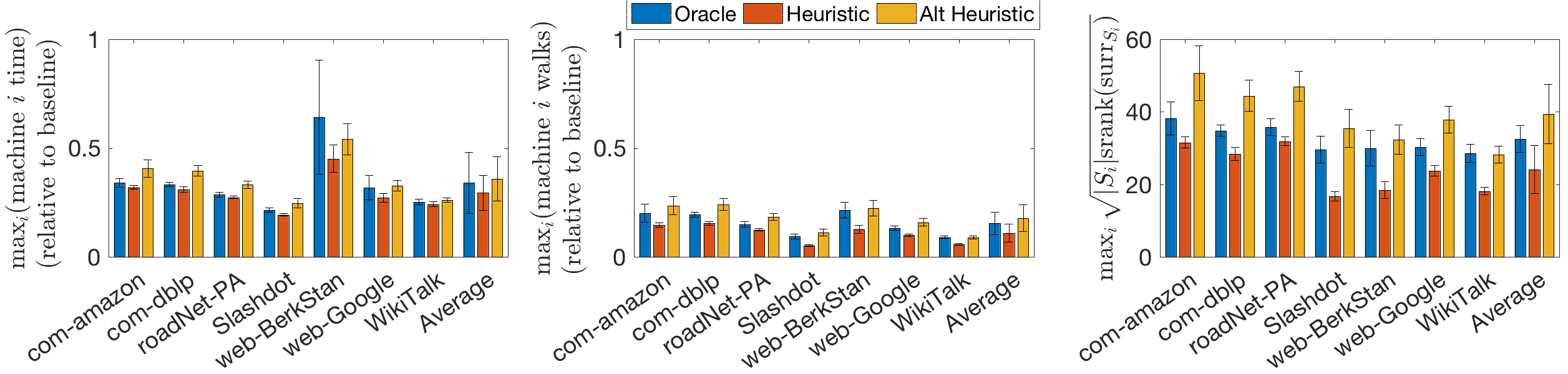} 
\caption{The $\sigma_{\avg}$ matrix approximation scheme is typically 2-3 times faster than the baseline scheme in the distributed setting of Section \ref{secDistributed}, and our heuristic partitioning schemes (Algorithms \ref{algHeuristicSrankExact} and \ref{algHeuristicSrankApprox}) perform similar to the oracle method.} \label{figDsMat}
\end{figure}

\subsection{Other results for source partitioning schemes}

As discussed at the conclusion of Section \ref{secDistributed}, it is crucial that our source partitioning schemes (Algorithms \ref{algHeuristic}, \ref{algHeuristicSrankExact}, and \ref{algHeuristicSrankApprox}) balance the number of sources assigned to each machine. To see why, note that the baseline schemes have objective function value $|S| / k$; hence, if some machine $i$ is assigned $O(|S|)$ sources using our schemes, we may only outperform the baseline when clustering is extreme. Luckily, we find that the partitions are typically quite balanced in practice, despite the lack of explicit balance constraints in Algorithms \ref{algHeuristic}, \ref{algHeuristicSrankExact}, and \ref{algHeuristicSrankApprox}. To demonstrate this, we show the maximum and minimum number of sources assigned to machines for the three partitioning schemes in Fig.\ \ref{figBalance}. Averaged across graphs, Algorithms \ref{algHeuristic}, \ref{algHeuristicSrankExact} and \ref{algHeuristicSrankApprox} typically produce partitions with $|S_i| \in [ 85, 122 ]$, $|S_i| \in [ 55, 188 ]$, and $|S_i| \in [ 75, 134 ]$, respectively (the red line shows $|S| / k = 100$, i.e.\ a perfectly balanced partition). We also note that, while Algorithm \ref{algHeuristicSrankExact} typically produces the least balanced partition, its overall performance is similar to that for Algorithm \ref{algHeuristicSrankApprox} (see Fig.\ \ref{figDsMat}), which we have argued is more useful in practice for large $S$.

\begin{figure}
\centering
\includegraphics[height=\plotHeight]{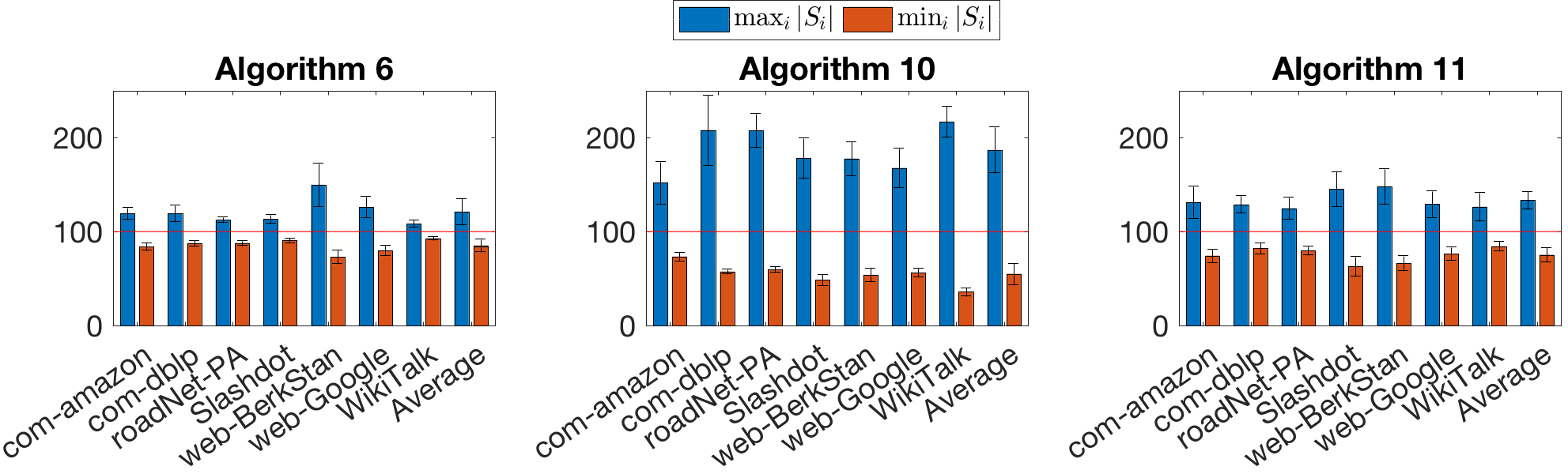}
\caption{All of our source partitioning schemes produce partitions $\{ S_i \}_{i=1}^k$ with $|S_i|$ reasonably close to $|S| / k = 100$ for every subset $i$ (where $|S| / k = 100$ is the case of perfectly balanced partition).} \label{figBalance}
\end{figure}

}

\end{appendices}

\bibliographystyle{ACM-Reference-Format}
\bibliography{references} 

\end{document}